\documentclass[11pt,reqno]{amsart}
\usepackage[margin=1in]{geometry}
\usepackage{amsmath,amsthm,amssymb}
\usepackage{bm}
\usepackage{mathrsfs} 
\usepackage{enumitem}
\usepackage{mathrsfs} 
\usepackage{multirow}
\usepackage{booktabs}
\usepackage{algorithm,algorithmic}

\usepackage{pgf,tikz,pgfplots}
\usetikzlibrary{arrows}
\usepackage[unicode]{hyperref}
\hypersetup{colorlinks=true, linkcolor=blue, citecolor=blue, filecolor=blue, urlcolor=blue}
\usepackage{chngcntr}
\usepackage{mathabx}
\usepackage{bm}
\usepackage{etoolbox}
\usepackage{enumitem}
\usepackage{subcaption}
\usepackage{seqsplit}

\usepackage[comma,sort&compress, numbers]{natbib}

\newcommand{\E}{\mathbb{E}}
\renewcommand{\P}{\mathbb{P}}
\newcommand{\R}{\mathbb{R}}
\newcommand{\cB}{\mathcal{B}}
\newcommand{\cE}{\mathcal{E}}
\newcommand{\cG}{\mathcal{G}}
\newcommand{\cN}{\mathcal{N}}
\newcommand{\cS}{\mathcal{S}}
\newcommand{\cV}{\mathcal{V}}
\newcommand{\bA}{\bm{A}} 
\newcommand{\bX}{\bm{X}}
\newcommand{\bY}{\bm{Y}}
\newcommand{\bZ}{\bm{Z}}

\newcommand{\bx}{\bm{x}}

\newcommand{\sE}{\mathscr{E}}
\newcommand{\sX}{\mathscr{X}}

\newcommand{\gX}{\tilde{\mathscr{X}}}

\newcommand{\ra}{\rightarrow}
\newcommand{\pto}{\overset{P}{\rightarrow}}
\newcommand{\tr}{\mathrm{Tr}}
\renewcommand{\d}{\mathrm{d}}

\newcommand{\Var}{\mathrm{Var}}
\newcommand{\vep}{\varepsilon}

\numberwithin{equation}{section} 

\newtheorem{theorem}{Theorem}[section]

\newtheorem{corollary}{Corollary}[section]

\newtheorem{lemma}{Lemma}[section] 
\newtheorem{proposition}{Proposition}[section]

\newtheorem{assumption}{Assumption}[section]

\theoremstyle{definition} 

\newtheorem{remark}{Remark}[section]
\newtheorem{example}{Example}[section]
\DeclareMathOperator{\Tr}{Tr} 
\newcommand{\one}{\bm{1}}

\allowdisplaybreaks

\begin{document}

\title[Conditional Mean Independence, Global Sensitivity Analysis, and Variable Screening]{Conditional Mean Independence, Global Sensitivity Analysis, and Variable Screening Using nearest neighbor Graphs}
\author[Chatterjee, Niu, and Bhattacharya]{Anirban Chatterjee, Ziang Niu, and Bhaswar B. Bhattacharya} 
\address{Department of Mathematics and Statistics\\ Boston University\\ Boston\\ MA 02215\\ United States}
\email{achatter@bu.edu}
\address{Department of Statistics and Data Science\\ University of Pennsylvania\\ Philadelphia\\ PA 19104\\ United States}
\email{ziangniu@wharton.upenn.edu}
\address{Department of Statistics and Data Science\\ University of Pennsylvania\\ Philadelphia\\ PA 19104\\ United States}
\email{bhaswar@wharton.upenn.edu}

\begin{abstract}  
Quantifying how much of the variation in a response is explained by its conditional mean is central to many statistical problems. In this paper, we develop a unified framework for this problem, centered on a normalized conditional mean discrepancy that coincides with the classical Sobol’ index for univariate responses and its natural trace-based extension for multivariate responses. We propose a simple estimator of this measure based on nearest-neighbor graphs that can be computed in near-linear time. We establish its consistency and derive its rate of convergence. Further, under the null hypothesis of conditional mean independence, a studentized version of the statistic is asymptotically standard normal. This leads to a computationally efficient test for conditional mean independence that attains the correct asymptotic level and is universally consistent, without requiring bootstrap calibration or sample splitting. Building on the same estimator, we develop a model-free sequential variable-screening procedure that selects a sufficient set with high probability and an exponential error bound. We also discuss extensions of the framework to quantifying interaction effects through higher-order Sobol’ indices. Simulations and real-data experiments demonstrate strong power, effective variable screening, and substantial computational gains over existing methods.   
\end{abstract}


\maketitle


\section{Introduction}

Understanding how a collection of covariates influences a response variable is central to many statistical problems, including variable selection \citep{szekely2014partial,park2015partial,yan2021rare,lavergne2000nonparametric,tang2018testing}, feature screening \citep{variablescreening,fan2014nonparametric,zhang2017feature,yan2021rare,shao2014martingale,li2012feature}, graphical modeling \citep{li2018bayesian,li2019expectation,gan2019bayesian,fan2020projection}, and applications in the biological sciences \citep{subramanian05,efron07,newton07}. This has motivated the study of feature importance and dependence measures, an area that has attracted renewed attention in recent years owing to the increasing reliance on black-box models for handling the complexity of large-scale modern problems. While a substantial body of work has focused on general notions of dependence (see \citep{josse2016measuring,chatterjee2024survey} for recent surveys), in regression-type problems the primary object of interest is often the conditional mean, since it captures the systematic effect of the covariates on the response \cite{cook2002dimension} (see \cite{verdinelli2024feature} for a recent survey of various notions of feature importance in regression). 
%
%
A natural diagnostic step toward understanding whether a collection of covariates $\bX$ affects an outcome variable $\bY$ is to determine whether $\bX$ contributes to the conditional mean of $\bY$. This entails testing the null hypothesis: 
\begin{align}\label{eq:H0}
    H_0 : \E[\bY\mid \bX] = \E[\bY] \text{ almost surely}, 
\end{align}
based on i.i.d. samples $\{(\bY_i, \bX_i)\}_{1 \leq i \leq n}$. This is the problem of testing conditional mean independence, which may be viewed as the nonparametric analogue of the regression goodness-of-fit problem. Several approaches to testing this hypothesis have been developed over the years. One approach is to estimate a (weighted) $L_2$-discrepancy between $\E[\bY\mid \bX]$ and $\E[\bY]$ using techniques from kernel density estimation \cite{lavergne2000nonparametric,fan1996consistent,delgado2001significance,zhu2018dimension,ait2001goodness,tian2025variation}. However, in most cases, for such methods convergence rate of the test statistic under $H_0$ depends on the kernel bandwidth and typically worsens with dimension. A major advance was the introduction of the martingale difference divergence (MDD) \citep{shao2014martingale}, an elegant extension of the celebrated distance covariance \citep{szekely2007measuring} for measuring conditional mean dependence. However, MDD-based tests typically require bootstrap calibration \cite{lai2021kernel,lee2020testing}, because their null distributions are non-pivotal and do not admit closed-form critical values. Another emerging approach, which is particularly well suited to high-dimensional problems, leverages machine-learning algorithms to estimate the conditional mean function \cite{williamson2023general,dai2022significance,lundborg2024projected,zhang2025testing,cai2025test,williamson2021nonparametric}. These methods, however, typically require sample splitting to estimate the conditional mean function, and the performance of the resulting tests depends on the accuracy of the underlying nonparametric estimators.

Quantifying the influence of input variables on a model output is also the central objective in Global Sensitivity Analysis (GSA) \citep{de2008uncertainty,saltelli2000sensitivity,wagner1995global}. Classical approaches in this area rely on an analysis-of-variance (ANOVA) decomposition, with Sobol' indices \citep{sobol1993sensitivity,sobol2001global} being among the most widely used metrics for measuring variable importance. 
A broad range of methods for estimating Sobol' indices has been developed over the years, including Pick--Freeze estimators, orthogonal basis expansions, and quasi-Monte Carlo methods 
(see \cite{daveiga2021basics} for a comprehensive overview). More recently, inspired by the elegant ideas of \citet{chatterjee2021new}, there has been renewed interest in rank and nearest neighbor-based approaches to Sobol' index estimation \citep{chhaibi2026martingale,lin2022limit,gamboa2022global}. Moreover, while Sobol' indices were originally formulated for scalar outputs, a growing body of work has extended these notions to the multivariate setting \citep{gamboa2013sensitivity,garcia2014global,cheng2019multivariate,milton2025sobol}.

In this paper, using variance based importance measures such as Sobol' indices
as the common basic primitive, we propose a unified computationally efficient approach for testing conditional mean independence, variable screening, and global sensitivity analysis, using a nearest neighbor graph based approach, that does not require bootstrap resampling or sample splitting. 
The following section summarizes the main methodological and theoretical results obtained in the paper. In particular, we present the proposed nearest neighbor estimators, establish their consistency and convergence rates, develop a test for conditional mean independence, and develop an algorithm for variable screening. We also discuss how the same framework naturally extends to measuring higher-order interaction effects.

\subsection{Summary of Results}

Our starting point is a basic measure of conditional mean dependence that quantifies the proportion of variability in the response explained by the regression function. This quantity has three fundamental properties one would expect from a mean dependence measure: it takes values in the interval $[0, 1]$, vanishes if and only if conditional mean independence holds, and attains its maximal value of 1 when the response is completely determined by the covariates (see Proposition \ref{ppn:normalized}). Moreover, this measure coincides with the classical Sobol' index, when the response is univariate, and its natural trace-based extension, when the response is multivariate \citep{gamboa2013sensitivity,gamboa2014sensitivity} (see Section \ref{sec:connect_sobol} for details). 
In this paper, we propose a nearest neighbor-based estimator of this index and establish the following properties:

\begin{itemize} 

\item The estimate has a simple, interpretable form, which does not require any
estimation of density or distribution functions. Moreover, the estimate can be computed in near-linear time (with a fixed number of nearest neighbors), irrespective of the dimension of the data (see Remark \ref{remark:graphcomputation}).

\item The estimator is consistent for the population measure under mild moment conditions (Theorem \ref{thm:etan_consistent}). Furthermore, in Theorem \ref{thm:rate_of_convg} we obtain the rate at which the estimator converges to the population measure as the sample size increases. As a consequence, one obtains an analogous convergence guarantee for the estimated Sobol' index, linking the finite-sample behavior of our estimator directly to that of the corresponding population sensitivity measure.

\item Under the null hypothesis of conditional mean independence, the estimator is asymptotically standard Gaussian, after a simple data-driven standardization (Theorem \ref{thm:null_dist}). This allows us to readily select the rejection threshold (based on the asymptotic distribution), resulting in a test that asymptotically has the correct level and is universally consistent against fixed alternatives (Corollary \ref{cor:test_consistency}). Consequently, our method is significantly faster than existing tests for conditional mean independence that rely on non-pivotal limiting distributions and therefore require bootstrap calibration (see Remark \ref{remark:distributioncomparison}). Another important advantage of our approach is that it does not require sample splitting. This is particularly important in practice, since sample splitting can lead to a loss of power in finite samples. Indeed, our empirical results show that the proposed test often exhibits improved power compared to recent methods that rely on sample splitting to estimate the regression function. 

\end{itemize}

Next, in Section \ref{sec:variable_screening_theory}, we develop a model-free variable screening algorithm based on our nearest neighbor estimator of the conditional dependence index. The procedure adds variables sequentially according to their estimated contribution to the conditional mean signal. In Theorem \ref{thm:variable_screening} we show that the proposed algorithm selects a sufficient set (see \eqref{eq:sufficient_S}) with high probability, with an exponential error bound in the sample size. Then, in Section \ref{sec:higher_sobol}, we discuss how the nearest neighbor framework can be extended beyond the estimation of conditional mean dependence and global sensitivity measures to the estimation of interaction effects through higher-order Sobol' indices \citep{gamboa2013sensitivity, gamboa2014sensitivity}. In particular, we focus on the second-order case and propose a method for quantifying the proportion of response variation attributable purely to interaction effects, beyond the separate contributions of the individual inputs. Finally, in Section \ref{sec:experiments}, we present empirical results on both simulated and real data, comparing our methods with existing approaches. The following are the summary of our findings: 

\begin{itemize} 

\item In Section \ref{sec:cmi_test}, we compare our proposed test with other popular methods for conditional mean independence, specifically the MDD based test \citep{shao2014martingale,lee2020testing} and the partial Mean Independence Test (pMIT) \citep{cai2025test}, as well as with tests designed for the more general null hypothesis of conditional independence, such as distance covariance (dCov) \cite{szekely2007measuring} and the Azadkia--Chatterjee coefficient \cite{azadkia2021simple,chatterjee2021new}. The experiments show that the proposed procedure achieves strong power against a wide range of nonlinear alternatives, controls Type I error in settings where more general dependence tests over-reject, and is computationally much faster than the competing methods. 

\item In the variable screening experiments, the proposed method is competitive with existing model-free screening procedures in simulation studies and yields strong predictive performance on the augmented California Housing dataset (Section \ref{sec:variable_screening}).

\end{itemize} 

The proofs of the main results are given in the Appendix. The code for the experiments can be found in \href{https://github.com/anirbanc96/ncmd}{\texttt{https://github.com/anirbanc96/ncmd}}.

\section{Conditional Mean Dependence and Global Sensitivity Measures }

In this section, we introduce a simple measure of conditional mean dependence and show how it relates to certain multivariate generalizations of Sobol' indices.

\subsection{A Simple Measure of Conditional Mean (In)dependence}\label{sec:def_eta}


Let $(\bY,\bX)$ be random variables taking values in $\R^p \times \R^d$ with joint distribution $P_{\bY\bX}$ and marginal distributions $P_{\bX}$ and $P_{\bY}$, respectively. Throughout we will assume that the conditional distribution $P_{\bY\mid\bX}$ exists and that $\E[\|\bY\|_2^2] < \infty$. We also assume that $\bY$ is not almost surely a constant. Then the hypothesis of conditional mean independence can be stated as follows. 
\begin{align}\label{eq:h0_formal}
    H_0 : P_{\bX}\left(\E[\bY\mid \bX] = \E[\bY]\right)=1 
    \quad \text{ versus } \quad 
    H_1 : P_{\bX}\left(\E[\bY\mid \bX] = \E[\bY]\right)<1.
\end{align}
To test this hypothesis, the first step is to quantify the extent to which $\E[\bY\mid \bX]$ departs from $\E[\bY]$. A natural choice is the squared $L_2$ distance between the  vectors $\E[\bY\mid \bX]$ and $\E[\bY]$, that is, $\|\E[\bY\mid \bX]-\E[\bY]\|_2^2$. Note that this quantity depends on the particular value of $\bX$ under conditioning. To obtain a global measure of discrepancy, we therefore take its expectation with respect to the marginal distribution of $\bX$: 
\begin{align}\label{eq:EYX}
\E\left[\left\|\E\left[\bY\mid \bX\right] - \E[\bY]\right\|_2^2\right].
\end{align}
Note that the above quantity is zero if and only if $H_0$ in \eqref{eq:h0_formal} holds. However, under $H_1$ it may be unbounded, and therefore may fail to provide a meaningful measure of the strength of the mean-level association between $\bY$ and $\bX$. To remedy this, we consider the normalized measure:
\begin{align}\label{eq:def_eta}
    \eta = \dfrac{\E\left[\left\|\E\left[\bY\mid \bX\right] - \E[\bY]\right\|_2^2\right]}{\E\left[\left\|\bY - \E[\bY]\right\|_2^2\right]}.
\end{align} 
Note that this quantity is well defined whenever the denominator in \eqref{eq:def_eta} is nonzero, that is, whenever $\bY$ is not almost surely constant (as assumed above). We will refer to the measure as the {\it Normalized Conditional Mean Discrepancy} (NCMD). The following proposition collects some basic properties of the NCMD measure. The proof is 
given in Appendix \ref{sec:propertiespf}. (Also, proposed in \cite{williamson2021nonparametric,tian2025variation})

\begin{proposition}\label{ppn:normalized}
    Suppose $\E[\|\bY\|_2^2]<\infty$ and $\bY$ is not almost surely a constant. Then the measure $\eta$ defined in \eqref{eq:def_eta} satisfies the following: 
    \begin{enumerate}[label=$(\mathrm{P}\arabic*)$,ref=$(\mathrm{P}\arabic*)$]
    \item \label{item:P1} $\eta\in [0,1]$. 
    \item \label{item:P2} $\eta = 0$ if and only if $\bm{H}_0$ in \eqref{eq:h0_formal} holds.
    \item \label{item:P3} $\eta = 1$ if and only if $\bY$ is a measurable function of $\bX$ almost surely.
\end{enumerate}
\end{proposition}


The above properties show that $\eta$ may be interpreted as a measure of the strength of the influence of $\bX$ on the regression function $\E[\bY\mid \bX]$. The two extreme values admit natural interpretations. At one extreme, $\eta=0$ corresponds to conditional mean independence, in which case the regression function $\E[\bY\mid \bX]$ is constant and, hence, does not depend on $\bX$. At the other extreme, $\eta=1$ corresponds to the case in which $\bY$ is completely determined by $\bX$. These   properties are the analogues of R\'enyi's axioms \cite{renyi1959measures} in the context of measuring mean dependence and are closely related to those satisfied by the nonparametric measures of (conditional) dependence introduced in \cite{chatterjee2021new,azadkia2021simple,dette2013copula,deb2020measuring,huang2022kernel}, among others.

\subsection{Global Sensitivity Measures }\label{sec:connect_sobol}

In this section, we relate the measure $\eta$ to multivariate generalizations of the Sobol' index \citep{sobol1993sensitivity,sobol2001global}, which play a central role in global sensitivity analysis (GSA) \citep{wagner1995global,de2008uncertainty,saltelli2000sensitivity}. To that end, consider a function $f:\R^d\times \R^m\ra \R$ and independent random variables $\bX\in \R^d$ and $\bZ\in \R^m$. Let $Y = f(\bX, \bZ)$, that is, $Y$ is a scalar output of the model $f$ with input variables $\bX$ and $\bZ$. The Sobol' index of the variables $\bX$ is then defined as (see \cite{gamboa2013sensitivity, gamboa2022global, janon2014asymptotic} and references therein), 
\begin{align}\label{eq:univariate_sobol}
    S^{\bX} = \frac{\Var\left[\E\left[Y\mid \bX\right]\right]}{\Var[Y]}.
\end{align}
This index provides a natural measure of the influence of the input $\bX$ on the output $Y$. In this setting, the measure $\eta$ defined in \eqref{eq:def_eta} can be expressed as follows: 
\begin{align*}
    \eta = \frac{\E\left[\left|\E\left[Y\mid \bX\right] - \E[Y]\right|^2\right]}{\E\left[\left|Y - \E[Y]\right|^2\right]} = \frac{\Var\left[\E[Y\mid\bX]\right]}{\Var[Y]} = S^{\bX}.
\end{align*}
Hence, when the response is univariate, the measure $\eta$ is precisely the Sobol' index \eqref{eq:univariate_sobol}. This measure is also known as the nonparametric coefficient of determination \cite{DoksumSamarov1995}, and its numerator is often referred to as the residual variance \cite{DevroyeGyorfiLugosiWalk2018}. Sobol' indices are also building blocks of Shapley values, which are widely used in machine learning as variable importance measures for model interpretation \cite{owen2014sobol}. 


Next, consider a multivariate response model $\bY = f(\bX, \bZ)$, where $f:\R^d\times \R^m\ra\R^p$. Then, by the Hoeffding decomposition \citep{van2000asymptotic},  
\begin{align}\label{eq:hoeffdin_decomp_1}
    \bY = f(\bX,\bZ) = \E[\bY] + f_{\bX} + f_{\bZ} + f_{\bX\bZ},
\end{align}  
where $f_{\bX} = \E\left[\bY\mid\bX\right] - \E[\bY]$, $f_{\bZ} = \E\left[\bY\mid\bZ\right] - \E[\bY]$, and $f_{\bX\bZ} = \bY - f_{\bX} - f_{\bZ} - \E[\bY]$. Using $L_2$ orthogonality and taking covariances on both sides gives
\begin{align}\label{eq:decomp_var}
    \Sigma_{\bY} = \Sigma_{f_{\bX}} + \Sigma_{f_{\bZ}} + \Sigma_{f_{\bX\bZ}},
\end{align}  
where $\Sigma_{\bm W}$ denotes the covariance matrix of a random variable $\bm W$. For scalar outputs (that is when $p = 1$), the covariance matrices reduce to scalar variances, and \eqref{eq:decomp_var} can be interpreted as the decomposition of the total variance of $Y$ into the variance due to the input factors in $\bX$, the variance due to the input factors in $\bZ$, and the variance due to interactions between $\bX$ and $\bZ$. The (univariate) Sobol' index then represents the sensitivity of $Y$ to the inputs in $\bX$. A natural multivariate extension of this definition is to consider the ratio $\tr(\Sigma_{f_{\bX}})/\tr(\Sigma_{\bY})$ \citep{gamboa2013sensitivity, gamboa2014sensitivity}. It is now straightforward to observe that  
\begin{align}\label{eq:eta_trace_def}
    \eta = \dfrac{\E\left[\left\|\E\left[\bY\mid \bX\right] - \E[\bY]\right\|_2^2\right]}{\E\left[\left\|\bY - \E[\bY]\right\|_2^2\right]} = \frac{\tr(\Sigma_{f_{\bX}})}{\tr(\Sigma_{\bY})} . 
\end{align}  
This shows that, even when the response is multivariate, the NCMD measure defined in \eqref{eq:def_eta} can be interpreted as a multivariate Sobol' index, quantifying the proportion of the total variance of the model $f
$ explained by the variables $\bm{X}$. A similar measure, motivated differently through projecting the covariance into a scalar, has been studied in \citet{gamboa2013sensitivity, gamboa2014sensitivity}. In particular, \citet[Proposition 3.1]{gamboa2013sensitivity} shows that this measure is optimal for defining a general global sensitivity index via such scalar projections.

\subsection{Connections to Other Variable Importance Measures}

In recent years, several measures of variable importance have been proposed. mong those perhaps most closely related to the measure $\eta$, is the minimum Mean Squared Error (mMSE) gap \cite{zhang2020floodgate,williamson2021nonparametric,williamson2023general}, which measures the importance of the variable $\bX$ for predicting a univariate response $Y$ in the presence of a confounding variable $\bZ$. In the setting of this paper, where there are no confounders, the mMSE gap statistic takes the form $\E[\left(Y - \E[Y]\right)^2] - \E[(Y - \E[Y\mid \bX])^2]$. By the law of total variance, this quantity is equal to $\operatorname{Var}[\E[Y\mid \bX]]$, which is the numerator of the statistic $\eta$ in the univariate-response case.

In another direction, \citet{borgonovo2025convexity,borgonovo2025global} proposed a family of optimal transport (OT)-based measures of variable importance, defined in terms of the expected optimal transport distance between the marginal distribution of $Y$ and the conditional distribution of $Y\mid \bX$. This is reminiscent of the measure of association introduced in \citet{deb2020measuring}, which is defined as the expected distance between the mean embeddings of the same distributions in a reproducing kernel Hilbert space. In fact, \cite{borgonovo2025global} extends this idea to an OT-based measure of global sensitivity, which satisfies properties analogous to those in Proposition~\ref{ppn:normalized}. 
However, it should be noted that, unlike Sobol' indices such as $\eta$, which capture mean dependence, OT-based measures typically characterize full statistical dependence.  

\section{Estimating NCMD using Nearest Neighbors: Consistency and Rate of Convergence}\label{sec:estimation_nn}
In this section, we propose a method to estimate the NCMD measure, given independent samples $(\bY_1, \bX_1),\ldots, (\bY_n,\bX_n)$ from the joint distribution $P_{\bY\bX}$, and establish its consistency and rate of convergence. To this end, first note that by an application of the law of large numbers,
\begin{align}\label{eq:Dn}
D_n := \frac{1}{n}\sum_{u=1}^{n}\left\|\bY_u\right\|^2 - \frac{1}{n(n-1)}\sum_{1 \leq u \neq v \leq n}\bY_u^\top \bY_v \pto \E\left[\left\|\bY - \E\left[\bY\right]\right\|_2^2\right], 
\end{align}
that is, $D_n$ is a consistent estimator of the denominator of $\eta$ (recall \eqref{eq:def_eta}). To estimate the numerator of $\eta$,  consider the following decomposition:
\begin{align}\label{eq:numerator_decomp}
    \E\left[\left\|\E\left[\bY\mid\bX\right] - \E\left[\bY\right]\right\|_2^2\right] = \E\left[\E\left[\bY^\top \bY^\prime\mid \bX\right]\right] - \left\|\E\left[\bY\right]\right\|_2^2 , 
\end{align}
where $\bX\sim P_{\bX}$ and $\bY,\bY^\prime$ are generated independently from the conditional distribution $P_{\bY\mid\bX}$. A consistent estimate of the second term in \eqref{eq:numerator_decomp} can be obtained easily using 
\begin{align}\label{eq:N2}
\frac{1}{n(n-1)}\sum_{1 \leq u \neq v \leq n}\bY_u^\top \bY_v \pto \left\|\E\left[\bY\right]\right\|_2^2 . 
\end{align} 
To estimate the first term in \eqref{eq:numerator_decomp}, we fix $\bX = \bX_u$ for $1\leq u\leq n$, and consider the conditional expectation $\E\left[\bY^\top\bY^\prime\mid\bX = \bX_u\right]$. The idea then is to estimate this quantity by averaging the inner product over indices corresponding to observations that are ``close'' to $\bX_u$. A natural way to quantify such proximity is through nearest neighbor graphs. Specifically, fix $K \geq 1$ and consider the directed $K$-nearest neighbor ($K$-NN) graph $G(\sX_n)$ associated with the data points $\sX_n = \{\bX_1,\ldots,\bX_n\}$, in which each $\bX_u \in \sX_n$ is connected by directed edges to its $K$ nearest neighbors in $\sX_n \setminus {\bX_u}$. We denote the presence of a directed edge from $\bX_u$ to $\bX_v$ in $G(\sX_n)$ by $\bX_u\ra\bX_v$ and the presence of directed edges both from $\bX_u$ to $\bX_v$ and $\bX_v$ to $\bX_u$ by $\bX_u \leftrightarrow \bX_v$. 
%
%
Moreover, for $\bX_u \in \sX_n$, denote its set of neighbors in $G(\sX_n)$ as: 
\begin{align}\label{eq:NGX_n}
    N_{G(\sX_n)}(u) = \left\{ v \in [n]: \bX_u\ra\bX_v\text{ is an edge in }G(\sX_n)\right\} . 
\end{align}
Then the $K$-NN-based estimator of the first term in \eqref{eq:numerator_decomp} is given by 
\begin{align}\label{eq:N1estimate}
\frac{1}{n}\sum_{u=1}^{n} \frac{1}{K} \sum_{v \in N_{G(\sX_n)}(u)}\bY_u^\top\bY_v.
\end{align}
Combining the above with \eqref{eq:Dn} and \eqref{eq:N2} we define our estimator of $\eta$ as follows: 
%
%
\begin{align}\label{eq:def_eta_hat}
    \hat\eta_n:= \frac{\frac{1}{n}\sum_{u=1}^{n} \frac{1}{K}\sum_{v \in N_{G\left(\sX_n\right)}(u)}\bY_u^\top \bY_v - \frac{1}{n(n-1)}\sum_{1 \leq u \neq v \leq n}\bY_u^\top \bY_v}{\frac{1}{n}\sum_{u=1}^{n}\left\|\bY_u\right\|^2 - \frac{1}{n(n-1)}\sum_{1 \leq u \neq v \leq n}\bY_u^\top \bY_v}.
\end{align}
Note this estimator is well defined, since under our assumptions $\bY$ is almost surely not a constant.

\begin{remark}\label{remark:graphcomputation} Note that the estimate \eqref{eq:def_eta_hat} can be computed easily in $O(Kn \log n)$ time, in any dimensions. This is because $K$-NN graph can be computed in $O(Kn \log n)$ time (see, for example, \cite{friedman1977algorithm}) and, given the graph, \eqref{eq:N1estimate} can be computed in $O(Kn)$ time, since a $K$-NN graph has $O(Kn)$ edges. Also, observe that the second term in the numerator and the denominator can be computed in $O(n)$ time using the identity 
\begin{align}\label{eq:YY_simple}
   \sum_{1 \leq u \neq v \leq n}\bY_u^\top \bY_v =  \left\|\sum_{u=1}^{n}\bY_u\right\|_2^2 - \sum_{u=1}^{n}\left\|\bY_u\right\|_2^2 . 
\end{align} 
\end{remark}

\begin{remark} Given the connection between the NCMD measure and the classical Sobol' index, it is natural to ask whether existing estimation strategies for Sobol' indices can be adapted to this setting. The classical Pick-and-Freeze approach \citep{janon2014asymptotic, gamboa2016statistical} is not directly applicable beyond the setting of GSA, as it is a resampling-based method that requires generating additional independent samples and repeated evaluations of the model. 
To avoid the high computational cost of model evaluations, there have been alternative approaches based on kernel estimators that rely solely on the observed samples (see \citep{da2009local,DaVeigaGamboaKleinLagnouxPrieur2026} and the references therein). More recently, \cite{gamboa2022global} drew inspiration from the rank-based approach of \cite{chatterjee2021new} to propose a rank-based estimator of Sobol' indices in the univariate setting, where both $\bY$ and $\bX$ are univariate. \citet{chhaibi2026martingale} further generalized this to the multivariate response setting, by defining a coordinate-wise Sobol' index, and established asymptotic normality of the corresponding estimator. \citet{lin2022limit} also defined the same estimator when $\bX$ is allowed to be multivariate, but $\bY$ is still univariate. In this univariate setting, our estimator (with $K=1$) can be viewed as a comparable approach to the rank-based estimator. However, it is important to emphasize that our method is more general and remains well defined when both $\bY$ and $\bX$ are allowed to be multi-dimensional. Moreover, \cite{gamboa2022global} demonstrates that, in simple settings with linear associations, the rank-based approach outperforms the classical Pick-and-Freeze method. This observation suggests that, in the multi-dimensional setting, our nearest neighbor-based approach may yield similar improvements.
\end{remark}

We now proceed to establish the consistency of the estimator $\hat\eta_n$. To this end, we make the following assumption on the conditioning variable $\bX$.

\begin{assumption}\label{assumption:normX_continuous}
    The random variable $\bX$ takes values in $\R^d$ for some $d\geq 1$ and $\left\|\bX - \bX^\prime\right\|_2$ has a continuous distribution, where $\bX,\bX^\prime$ are generated independently from $P_{\bX}$.
\end{assumption}

This assumption ensures that the $K$-NN graph defined using the Euclidean distance  $\|\cdot\|_2$ is well defined and the degrees of the vertices scale proportionally with $K$ \citep{jaffe2020randomized}. The following theorem establishes the consistency of $\hat\eta_n$. The proof is given in Appendix \ref{appendix:proofof_etan_consistent}.

\begin{theorem}\label{thm:etan_consistent}
    Suppose Assumption \ref{assumption:normX_continuous} holds and $\E[\|\bY\|_2^{4+\delta}]<\infty$, for some $\delta>0$. Then, 
    \begin{align}\label{eq:NCMDconsistency}
        \hat\eta_n\pto \eta , 
    \end{align}
    where $\eta$ and $\hat\eta_n$ are defined in \eqref{eq:def_eta} and \eqref{eq:def_eta_hat}, respectively.
\end{theorem}

Having established consistency, the next natural question is to determine the rate of convergence in \eqref{eq:NCMDconsistency}. For this, we make the following assumptions: 

\begin{assumption}\label{assumption:rate_of_convg}
For $(\bX, \bY) \sim P_{\bX \bY}$ the following holds: 
    \begin{enumerate}
        \item[$(1)$] Suppose $\E[\bX] = 0$ and there exist constants $C_1,C_2>0$ such that,         \begin{align*}
            \P\left(\left\|\bX\right\|_2>t\right)\leq C_1 e^{-C_2t} \quad \text{ and } \quad \P\left(\left\|\bY - \E\bY\right\|_2>t\right)\leq C_1e^{-C_2t} , 
        \end{align*} 
   for all $t>0$.
        \item[$(2)$] Define $g:\R^d\ra \R^p$ as $g(\bm x) = \E\left[\bY\mid\bX = \bm x\right]$. Then there exists $\beta, C_3>0$ such that,
        \begin{align*}
            \left|g(\bm x)^\top\left(g(\bm x_1) - g(\bm x_2)\right)\right|\leq C_3\left(1 + \|\bm x\|_2^{\beta} + \|\bm x_1\|_2^\beta + \|\bm x_2\|_2^\beta\right)\left\|\bm x_1 - \bm x_2\right\|_2 , 
        \end{align*} 
        for all $\bm x, \bm x_1, \bm x_2 \in \R^d$. 
    \end{enumerate}
\end{assumption}


Under the above assumptions, the following theorem establish the rate of $\hat \eta_n$. Note that since $\eta$ can be also interpreted as a Sobol' index, this result also establishes the convergence rates for estimation of Sobol' indices.

\begin{theorem}\label{thm:rate_of_convg}
    Suppose Assumption \ref{assumption:normX_continuous} and \ref{assumption:rate_of_convg} hold. Then,
    \begin{align}\label{eq:etaconvergence}
        \left|\hat\eta_n - \eta\right| = O_P\left(\max\left\{\frac{(\log n)^2}{\sqrt{n}}, \frac{(\log n)^{1+1/d}}{n^{1/d}}\right\}\right).
    \end{align}
    where $\eta$ and $\hat\eta_n$ are defined in \eqref{eq:def_eta} and \eqref{eq:def_eta_hat} respectively.
\end{theorem}

The proof of Theorem \ref{thm:rate_of_convg} is given in Appendix \ref{sec:rate_of_convgpf}. The first term in \eqref{eq:etaconvergence} corresponds to the variance, while the second term corresponds to the bias. Hence, for $d \leq 2$, the variance dominates 
and $\hat{\eta}_n$ attains a near-parametric convergence rate of order $O(1/\sqrt{n})$ (up to a $\mathrm{polylog}(n)$ factor). This is reminiscent of the rates obtained for nearest neighbor-based estimators in conditional independence testing \cite{azadkia2021simple}. Related convergence rates have also been established for Sobol' indices in settings where either one or both of $\bY$ and $\bX$ are univariate (see \cite{lin2022limit}). On the other hand, for $d \geq 3$, bias emerges as the dominant term and the rate of convergence deteriorates with the dimension. This is inherent to any procedure based on nonparametric estimation of conditional distributions, such as nearest neighbors or kernel-density methods (see, for example, \cite{azadkia2021simple, deb2020measuring, huang2022kernel}, for related results). Nonetheless, an interesting feature of the proposed estimate is that, under the null hypothesis of conditional mean independence, it achieves the parametric convergence rate of $O(1/\sqrt{n})$ in all dimensions, as shown in the next section.

\section{Testing Conditional Mean Independence}\label{sec:cmi_test_theory}

In this section, we develop an asymptotic test for the conditional mean independence hypothesis in \eqref{eq:h0_formal} based on the estimator of the NCMD measure introduced in the previous section. To this end, recall that the null hypothesis $H_0$ holds if and only if the quantity in \eqref{eq:EYX}, which is the numerator of $\eta$, is equal to zero. Hence, one can construct a test for conditional mean independence based on the numerator of estimate $\hat{\eta_n}$, which we denote as follows (recall \eqref{eq:def_eta_hat}):
\begin{align}\label{eq:Tn}
T_n = \frac{1}{n}\sum_{u=1}^{n} \frac{1}{K}\sum_{v \in N_{G\left(\sX_n\right)}(u)}\bY_u^\top \bY_v - \frac{1}{n(n-1)}\sum_{1 \leq u \neq v \leq n}\bY_u^\top \bY_v . 
\end{align} 
Denote by $\mathcal F(\sX_n)$ the $\sigma$-algebra generated by $\sX_n = \{X_1, X_2, \ldots, X_n\}$. Then, observe that, under $H_0$, 
\begin{align}\label{eq:ETn}
    \E_{H_0}\left[T_n\mid \mathcal{F}(\sX_n)\right] = \frac{1}{n K}\sum_{u=1}^{n}\sum_{v \in N_{G(\sX_n)}(u)} \|\E[\bY]\|_2^2 - \frac{1}{n(n-1)}\sum_{1 \leq u \neq v \leq n} \|\E[\bY]\|_2^2 = 0 , 
\end{align}
since under $H_0$, $\E[\bY\mid\bX] = \E[\bY]$. 

\begin{remark}\label{remark:unbaisedTn}
Recall, from the discussion following Theorem \ref{thm:rate_of_convg}, that the bias is the dominant term in the estimation error when $d \geq 3$. However, \eqref{eq:ETn} shows under $H_0$ the bias vanishes. This is a key feature that enables the construction of an asymptotically valid test for conditional mean independence, as explained below.  
\end{remark}

The next theorem establishes the asymptotic null distribution of $T_n$. Specifically, we show that $T_n$ (after appropriate normalization) converges to a limiting normal distribution in the Kolmogorov distance. Throughout, $\Phi$ will denote the cumulative distribution function of the standard Gaussian distribution.


\begin{theorem}\label{thm:null_dist}
    Suppose Assumption \ref{assumption:normX_continuous} holds and $\E[\|\bY\|_2^{8+\delta}]<\infty$, for some $\delta > 0$. Then, under $\bm H_0$, as $n\ra\infty$,
    \begin{align}\label{eq:cltH0} 
        \sup_{z\in\R}\left|\P_{H_0}\left(\frac{\sqrt{n}{T}_n}{\hat{\sigma}_n}\leq z\right) - \Phi(z)\right| \ra 0 , 
    \end{align} 
    where 
    \begin{align}\label{eq:varianceH0}
\hat \sigma_n^2 & := \frac{1}{nK^2} \sum_{1 \leq u \ne v \leq n} ((\bY_u - \bar{\bY}_n)^\top (\bY_v - \bar{\bY}_n))^2  \left( \one\left\{\bX_u \rightarrow \bX_v\right\} + \one\left\{ \bX_u \leftrightarrow \bX_v \right\} \right) \nonumber \\ 
& \hspace{1.35in} + \frac{1}{n} \sum_{u=1}^{n} \left( \frac{\bar{d}_u}{K} - 1 \right)^2 \, \bar{\bY}^\top (\bY_u - \bar{\bY}_n) (\bY_u - \bar{\bY}_n)^\top \bar{\bY} , 
\end{align}
with $\bar{\bY}_n = \frac{1}{n}\sum_{u=1}^{n}\bY_u$ and $\bar{d}_u:= | \left\{v \in [n]:\bX_v\ra\bX_u \right\}|$, for $1\leq u \leq n$. 
\end{theorem} 

Observe that the normalizing factor $\hat{\sigma}_n^2$ in \eqref{eq:varianceH0} above depends solely on the observed data. In fact, as discussed below, it is a consistent estimator of the variance of $\sqrt{n} T_n$ under $H_0$. Consequently, the result in \eqref{eq:cltH0} can be used directly to construct an asymptotically valid and universally consistent test for conditional mean independence (see Corollary \ref{cor:test_consistency}). The proof of Theorem \ref{thm:null_dist}, which is given in Appendix \ref{sec:cltpf}, proceeds in the following steps:

\begin{itemize} 

\item[(1)] First we use a second moment computation to decompose $T_n$ as follows (see Lemma \ref{lm:Rn}): 
\begin{align}\label{eq:RnTn}
\sqrt{n}T_n = R_n + o_{L_2}(1), 
\end{align} 
where 
\begin{align}\label{eq:Rn}
    R_n = \frac{1}{\sqrt{n} K}\sum_{u=1}^{n}\sum_{v \in N_{G(\sX_n)}(u)}\bY_u^\top \bY_v - \frac{1}{\sqrt{n}}\sum_{u=1}^{n}\E\left[\bY\right]^\top\left(2\bY_u - \E\left[\bY\right]\right).
\end{align}

\item[(2)] Then the key step in the proof is to apply Stein's method based on dependency graphs \cite{chen2004normal} to establish the CLT of $R_n$, specifically, 
\begin{align*}
    \sup_{z\in\R}\left|\P_{H_0}\left(\frac{R_n}{{\sigma}_n}\leq z\right) - \Phi(z)\right| \ra 0, 
\end{align*}
where $\sigma_n^2 := \Var_{H_0}\left[R_n\middle| \mathcal{F}(\sX_n)\right] $. 

\item[(3)] Note that $\sigma_n^2$ depends on the unknown conditional distribution of $\bY$ given $\bX$ under $H_0$ (see Lemma \ref{lm:varianceRn} for the precise expression). In Proposition \ref{ppn:varianceestimate} we show that $\hat \sigma_n^2$ as defined in \eqref{eq:varianceH0} is a consistent estimate of $\sigma_n^2$, that is, 
\begin{align*}
        \left|\frac{\hat{\sigma}_n^2}{\sigma_n^2}-1\right| = o_P(1) . 
    \end{align*}
The result in \eqref{eq:cltH0} then follows by replacing $\sigma_n$ with $\hat\sigma_n$ and replacing $R_n$ with $\sqrt{n}T_n$  (see Lemma \ref{lemma:replace_lemma} for details).

\end{itemize}


To construct an asymptotically valid test for the hypothesis \eqref{eq:h0_formal} using Theorem \ref{thm:null_dist}, fix $\alpha \in (0,1)$ and consider the test function 
\begin{align}\label{eq:conditional_mean_test}
    \phi_n = \one\left\{\left|\frac{\sqrt{n}T_n}{\hat\sigma_n}\right|>z_{\frac{\alpha}{2}}\right\},
\end{align}
where $z_{\frac{\alpha}{2}}$ denotes the $1-\frac{\alpha}{2}$-quantile of the $N(0,1)$ distribution. The following result is now an immediate consequence of Theorem \ref{thm:null_dist} and Theorem \ref{thm:etan_consistent}.


\begin{corollary}\label{cor:test_consistency} 
Suppose Assumption \ref{assumption:normX_continuous} holds and $\E[\|\bY\|_2^{8+\delta}]<\infty$, for some $\delta > 0$. Then the following hold: 
    
    \begin{itemize} 

\item (Asymptotic level $\alpha$) $\lim_{N\rightarrow\infty} \P_{H_0}(\phi=1) = \alpha$. 

\item (Universal consistency) For any $P_{\bY\bX} \in H_1$, $\lim_{N\rightarrow\infty} \P_{P_{\bY\bX}}(\phi=1)= 1$. 
\end{itemize}

\end{corollary}

\begin{remark}\label{remark:eta_hat_asymp_dist} 
From Theorem \ref{thm:null_dist} we can also obtain the limiting null distribution of the estimator $\hat{\eta}_n$ (recall \eqref{eq:def_eta_hat}). For this, note that 
$\hat\eta_n = T_n / D_n$, where $T_n$ and $D_n$ are the numerator and denominator of  $\hat \eta_n$, respectively (recall \eqref{eq:Dn} and \eqref{eq:Tn}). The distributional convergence established in Theorem \ref{thm:null_dist} together with \eqref{eq:Dn}, implies that under $H_0$ (that is, whenever $\eta = 0$),
\begin{align*}
\frac{\sqrt{n}}{\hat\sigma_n}\hat\eta_n
\overset{D}{\rightarrow}
 N\left(0,
\left(\frac{1}{\E[\|\bY - \E\bY\|_2^2]}\right)^2
\right) . 
\end{align*}
\end{remark}

\begin{remark}\label{remark:distributioncomparison} 
An important feature of the test obtained above is that the rescaled test statistic $\sqrt{n} T_n/\hat{\sigma}_n$ has a standard normal distribution under $H_0$. As a result, the rejection threshold of the resulting test can be readily  obtained, without having to estimate any nuisance parameter or use resampling. Moreover, the proposed test admits an $O(n\log n)$ implementation (for fixed $K$), since both the statistic $T_n$ and the estimator $\hat{\sigma}_n^2$ can be computed in $O(n\log n)$ time (recall the discussion in Remark \ref{remark:graphcomputation}). In contrast, tests based on the martingale difference correlation and its kernel-based extensions \cite{lai2021kernel,lee2020testing,conditionalmeancovariates}, have non-Gaussian limiting null distributions (specifically, an infinite mixture of chi-squares). Closed form estimates for the quantiles of such distributions are not available, in general, which necessitates the use permutation/bootstrap techniques for determining the rejection thresholds. On the other hand, kernel-smoothing-based tests for conditional mean independence are often asymptotically normal under $H_0$, but their convergence to normality usually depends on the kernel bandwidth and becomes slower as the dimension increases (see \cite{lavergne2000nonparametric,fan1996consistent,delgado2001significance,zhu2018dimension,ait2001goodness,tian2025variation}, among others). Moreover, most of the aforementioned procedures are built on high-order $U$-statistics and are thus considerably more computationally intensive than the near-linear-time method proposed here. More recently, a growing body of work on conditional mean independence testing has employed machine-learning methods to estimate the conditional mean function, that are particularly effective in high-dimensional settings \cite{williamson2023general,dai2022significance,lundborg2024projected,zhang2025testing,cai2025test}. However, these methods typically require sample splitting for estimating the conditional mean function, which can reduce power in finite samples (see the empirical results in Section \ref{sec:cmi_test} and Appendix \ref{appendix:cmi_test}). It is also worth noting that many of the aforementioned methods study the conditional mean independence problem in the more general setting in which one controls for additional covariates. The nearest neighbor-based method described in Section \ref{sec:estimation_nn} can be adapted to this more general setting, however, the resulting estimator generally loses its unbiasedness under the null (recall Remark \ref{remark:unbaisedTn}). Consequently, additional debasing techniques will be necessary for constructing valid tests.
\end{remark}

\section{Variable Screening using NCMD }\label{sec:variable_screening_theory}

In this section, we propose a model-free variable screening procedure based on our estimator of the NCMD measure. Specifically, suppose $Y \in \R$ is a response variable and $\bX = (X_1,\ldots,X_d)$ is a set of covariates.
The objective of variable screening is to identify a parsimonious subset of covariates that preserves most of the explanatory power for the response. Concretely, the objective is to identify $S\subset[d]$ such that
\begin{align}\label{eq:sufficient_S}
    \E\left[Y\mid\bX\right] = \E\left[Y\mid\bX_{S}\right]\text{ almost surely} , 
\end{align}  
where $\bX_S = (X_i)_{i \in S}$, for $S \subset \{1,2,\ldots,d\}$. We call a subset $S\subset[d]$ as \textit{sufficient} if $S$ satisfies \eqref{eq:sufficient_S}. To motivate our approach, note that by a standard application of Jensen's inequality one has, 
\begin{align*}
    \left|\E[Y]\right|^2
    \leq \E\left[\left|\E[Y\mid\bX_{S^\prime}]\right|^2\right]
    \leq \E\left[\left|\E[Y\mid\bX_{S}]\right|^2\right]
    \leq \E\left[\left|\E[Y\mid\bX]\right|^2\right]
\end{align*}
for all $S^\prime\subseteq S\subseteq [d]$. This monotonicity suggests selecting the sufficient subset $S$ by maximizing
\[
V(S) = \E\left[\E[Y\mid\bX_{S}]^2\right].
\]
Now, suppose we are given i.i.d. samples $(Y_1, \bX_1), (Y_2, \bX_2), \ldots, (Y_n, \bX_n)$,  where $\bX_i = (X_{ij})_{1 \leq j \leq d} \in \mathbb R^d$. Based on this sample, our algorithm for variable screening then proceeds by adding one variable at a time as  follows: Suppose we have already selected the variables $\{s_1, \ldots, s_t\}$. The idea then is to choose the next variable $s_{t+1}$ such that $V\left(\{s_1,\ldots, s_t\}\bigcup\{s_{t+1}\}\right)$ is maximized. However, for any subset $S\subseteq[d]$, $V(S)$ is unknown, since it depends on the population conditional expectation. Hence, our choice is based on a consistent estimator of $V(S)$ (see \eqref{eq:nn_consistency}):
\begin{align*}
    \hat V_n\left(S\right)=\left|\frac{1}{n}\sum_{u=1}^{n} \frac{1}{K}\sum_{v \in N_{G\left(\sX_{S}\right)}(u)} Y_u Y_v \right| 
    \pto V(S),
\end{align*}
where $\sX_S = \{ (X_{ij})_{j \in S}: i\in [n]\}$ and $G\left(\sX_{S}\right)$ denotes the $K$-NN graph constructed using $\sX_S$. Hence, our variable screening algorithm proceeds as above with $V$ replaced by $\hat V_n$. The pseudocode of the method is given in Algorithm \ref{alg:greedy_screening}. The following result establishes consistency of the algorithm. 

\begin{algorithm}[t]
\caption{Variable screening with NCMD}
\label{alg:greedy_screening}
\begin{algorithmic}[1]
\STATE \textbf{Input:} Samples $(Y_1, \bX_1), (Y_2, \bX_2), \ldots, (Y_n, \bX_n)$ and the  number of neighbors $K \geq 1$. 

\STATE \textbf{Initialization: } $\hat S, s^\ast \gets \emptyset \text{ and } \hat V_n(\emptyset) \gets -\infty$. 

\WHILE{$\hat V_n(\hat S \cup \{s^\ast\}) \geq \hat V_n(\hat S)$ and $\hat S\subset [d]$}
    \STATE $\hat S \gets \hat S\cup\{s^\ast\}$
    \STATE Choose $s^\ast\in [d]\setminus \hat S$ such that $\hat V_n(\hat S \cup \{s\})$ is maximised. That is $$s^\ast \gets \arg\max_{s \in [d]\setminus \hat S} \hat V_n(\hat S \cup \{s\}).$$
\ENDWHILE
\STATE \textbf{Output:} $\hat S$
\vspace{0.5em}
\end{algorithmic}
\end{algorithm}

\begin{theorem}\label{thm:variable_screening} Assume the following holds: 
    \begin{enumerate}
        \item[(a)] There exists $\delta\in (0,1)$ such that for any insufficient subset $S\subset[d]$, there exists some $s$ such that $V\left(S\cup\left\{s\right\}\right)\geq V(S) + \delta$.
        \item[(b)] There exists $M>0$ such that $\E[Y^2]\leq M<\infty$. 
        \item[(c)] Fix $\kappa = \lfloor \frac{M}{\delta} + 1\rfloor$. Assumption \ref{assumption:normX_continuous} and  Assumption \ref{assumption:rate_of_convg} are satisfied with $\bX$ replaced by $\bX_S$, for any $S\subseteq[d]$ such that $|S|\leq \kappa$.
        \item[(d)] Assume that $m(\bx) = \E\left[Y\mid\bX = \bx\right]$ is uniformly bounded and $$\E\left[ e^{t(Y - m(\bX))} \mid \bX\right]\leq e^{\sigma^2t^2},$$ almost surely $P_{\bX}$, for some constant $\sigma>0$.
    \end{enumerate}
    Then there exist constants $L_1,L_2>0$, depending on $M,\delta, d$ (and the constants $\beta,C_1,C_2,C_3$ from Assumption \ref{assumption:rate_of_convg}), such that,
    \begin{align*}
        \P\left(\hat S\text{ is sufficient}\right)\geq 1-L_1d^\kappa e^{-L_2n} .
    \end{align*}
\end{theorem}

The proof of Theorem~\ref{thm:variable_screening} is given in Appendix~\ref{sec:vspf}. 
First we show that whenever $\hat V_n(\cdot)$ approximates $V(\cdot)$ accurately enough over all subsets selected up to iteration $\kappa$ (which is ensured by Assumption (c)) $\hat S$ is sufficient (The existence of an sufficient set of appropriate size is guaranteed by Assumption (a)). Then we show that this approximation holds with high probability, thereby completing the proof.

\begin{remark} 
The variable selection procedure in Algorithm~\ref{alg:greedy_screening} is motivated by the \texttt{FOCI} algorithm proposed in \cite{azadkia2021simple} and its extension \texttt{KFOCI} in \cite{huang2022kernel}, where similar sequential variable selection algorithms are proposed using their respective measures of conditional dependence. Our method, on the other hand, is based on mean dependence, which is more relevant for additive noise models. The idea of using measures of mean independence for variable selection has also been explored in \cite{shao2014martingale} and \cite{tian2025variation}, under the assumption that the marginal utilities associated with the active predictors do not decay too rapidly. In particular, \cite{tian2025variation} employs the same underlying measure of mean independence; see \eqref{eq:EYX}, but estimates the conditional expectation using kernel smoothing. We compare our approach with these existing methods empirically in Section~\ref{sec:variable_screening}.  
\end{remark}


\section{Higher Order Sobol' Indices}\label{sec:higher_sobol}

As discussed in Section \ref{sec:connect_sobol}, the NCMD index $\eta$ measures the proportion of the output variance attributable to a given collection of variables. In other words, this is the `total effect' of the variables on the output. A natural next step is to move beyond total effects and study `interaction effects', namely, the portion of the output variance that is explained by the joint action of groups of variables. Quantifying these effects is important for identifying non-additive structure and for assessing whether interactions among inputs play a substantive role in driving output variability. These interaction effects are naturally characterized by higher-order Sobol' indices, which arise from the higher-order terms in the Hoeffding decomposition of the regression function (similar to \eqref{eq:hoeffdin_decomp_1}). In this section, we discuss how our nearest neighbor-based strategy can be extended to estimate such higher-order indices. For notational simplicity, we focus on second-order Sobol' indices measuring pairwise interactions in presence of the three input variables. The extension to higher-order indices can be carried out similarly (see \cite{tissot2015randomized,chastaing2012generalized, saltelli2002making}).

Recalling the notation from Section \ref{sec:connect_sobol}, consider the model $\bY = f(\bX_1, \bX_2, \bZ)$, where $\bX_1$, $\bX_2$, and $\bZ$ are independent covariates taking values in $\R^{d_1}$, $\R^{d_2}$, and $\R^{d_{\bZ}}$, respectively, and $\bY$ is a response taking values in $\R^{p}$. By the Hoeffding decomposition (see \cite{van2000asymptotic}), 
\begin{align}\label{eq:three_hoeffding}
    \bY = \E[\bY] + f_{\bX_1} + f_{\bX_2} + f_{\bZ} + f_{\bX_1\bX_2} + f_{\bX_1\bZ} + f_{\bX_2\bZ} + f_{\bX_1\bX_2\bZ}.
\end{align}
where $f_{\bX_1}$, $f_{\bX_2}$, $f_{\bZ}$ are as defined in Section \ref{sec:connect_sobol}, $$f_{\bX_1\bX_2} = \E[\bY\mid\bX_1,\bX_2] - f_{\bX_1} - f_{\bX_2} - \E[\bY]$$ and $f_{\bX_a\bZ} = \E[\bY\mid\bX_a,\bZ] - f_{\bX_a} - f_{\bZ} - \E[\bY]$, for $a \in \{1, 2\}$. Finally, $f_{\bX_1\bX_2\bZ}$ is defined through the identity in \eqref{eq:three_hoeffding}. Using the orthogonality of the components in \eqref{eq:three_hoeffding}, the covariance of $\bY$ decomposes as
\begin{align}\label{eq:cov_decomp_3}
    \Sigma_{\bY} = \Sigma_{f_{\bX_1}} + \Sigma_{f_{\bX_2}} + \Sigma_{f_{\bZ}} + \Sigma_{f_{\bX_1\bX_2}} + \Sigma_{f_{\bX_1\bZ}} + \Sigma_{f_{\bX_2\bZ}} + \Sigma_{f_{\bX_1\bX_2\bZ}},
\end{align}
where the individual covariance matrices are the covariance matrices of the corresponding components in \eqref{eq:three_hoeffding}. Recall from \eqref{eq:eta_trace_def} that the main effects quantify only the proportion of the total variance in $\bY$ explained individually by each input. The covariance decomposition above allows us to define higher-order extensions of Sobol' indices that quantify the variability in $\bY$ explained by interactions among the inputs. For notational brevity, we consider only the interaction between $\bX_1$ and $\bX_2$. When the model output $\bY$ is scalar (that is, $p = 1$), the second-order Sobol' index is defined as $\eta_2 = \Var[f_{\bX_1\bX_2}]/\Var[\bY]$ (see \cite{chastaing2012generalized,tissot2015randomized}). For multivariate model outputs, this definition can be generalized in the same spirit as Section \ref{sec:connect_sobol} as:   
\begin{align}\label{eq:def_eta_2}
    \eta_2 := \frac{\Tr\left(\Sigma_{f_{\bX_1\bX_2}}\right)}{\Tr(\Sigma_{\bY})}.
\end{align}
Using the orthogonality of the components in the decomposition \eqref{eq:three_hoeffding} together with the definition of $f_{\bX_1\bX_2}$, it is straightforward to show that
\begin{align}\label{eq:eta_2_alternate}
    \eta_2 = \frac{\E\left[\left\|\E[\bY\mid \bX_1,\bX_2] - \E[\bY]\right\|_2^2\right]}{\E\left[\left\|\bY - \E[\bY]\right\|_2^2\right]} - \eta_{\bX_1} - \eta_{\bX_2}  ,  
\end{align}
where $\eta_{\bX_1}$ and $\eta_{\bX_2}$ are the Sobol' indices corresponding to the main effects of $\bX_1$ and $\bX_2$ from \eqref{eq:eta_trace_def}. This decomposition highlights that the second-order Sobol' index $\eta_2$ measures the proportion of the total variance of $\bY$ that is attributable solely to the interaction between $\bX_1$ and $\bX_2$, beyond their individual contributions. The following proposition establishes the key properties of the second-order Sobol' index, analogous to Proposition \ref{ppn:normalized}. The proof follows directly from the Hoeffding decomposition and the covariance decomposition in \eqref{eq:cov_decomp_3}, and is therefore omitted.

\begin{proposition}\label{prop:eta2_properties}
    Suppose $\E[\|\bY\|_2^2]<\infty$ and $\bY$ is not almost surely constant. Then the second-order Sobol' index $\eta_{2}$ defined in \eqref{eq:def_eta_2} satisfies the following properties:
    \begin{enumerate}[label=$(\mathrm{Q}\arabic*)$,ref=$(\mathrm{Q}\arabic*)$]
    \item \label{item:Q1} $\eta_{2}\in [0,1]$. 
    \item \label{item:Q2} $\eta_{2} = 0$ if and only if $f_{\bX_1\bX_2} = 0$ almost surely, that is, the conditional mean of $\bY$ given $(\bX_1, \bX_2)$ is additively separable.
    \item \label{item:Q3} $\eta_{2} = 1$ if and only if $f_{\bX_1}$, $f_{\bX_2}$, $f_{\bZ}$, $f_{\bX_1\bZ}$, $f_{\bX_2\bZ}$, and $f_{\bX_1\bX_2\bZ}$ are all zero almost surely, so that all variability in $\bY$ arises purely through the interaction between $\bX_1$ and $\bX_2$.
\end{enumerate}
\end{proposition}

As in Section \ref{sec:estimation_nn}, we estimate $\eta_2$ by a nearest neighbor approach as follows:  For notational convenience, denote $\tilde{\bm X} = (\bX_{1}, \bX_{2})$. Given i.i.d. samples $\{(\bY_i, \tilde{\bm X}_i ,\bZ_{i}): 1\leq i\leq n\}$, let $\gX_{n} = \{\tilde{\bm X}_i: 1\leq i\leq n\}$ with the corresponding $K$-nearest neighbor graph $G(\gX_n)$. Then using the identity from \eqref{eq:eta_2_alternate} and recalling the estimates of the main effects from \eqref{eq:def_eta_hat} we define the estimate of $\eta_2$ as,
\begin{align}\label{eq:def_hat_eta_2}
    \hat \eta_2 = \frac{\frac{1}{n}\sum_{u=1}^{n} \frac{1}{K}\sum_{v \in N_{G\left(\gX_n\right)}(u)}\bY_u^\top \bY_v - \frac{1}{n(n-1)}\sum_{1 \leq u \neq v \leq n}\bY_u^\top \bY_v}{\frac{1}{n}\sum_{u=1}^{n}\left\|\bY_u\right\|^2 - \frac{1}{n(n-1)}\sum_{1 \leq u \neq v \leq n}\bY_u^\top \bY_v} - \hat{\eta}_{\bX_1} - \hat{\eta}_{\bX_2}  , 
\end{align}
where $N_{G(\gX_n)}(\cdot)$ denotes the set of neighbors defined in \eqref{eq:NGX_n} for the graph $G(\gX_n)$ and $\hat \eta_{\bX_1}, \hat \eta_{\bX_2}$ are estimates of the main effects $\eta_{\bX_1},\eta_{\bX_2}$ defined using samples $\{(\bY_i,\bX_{1,i}): 1\leq i\leq n\}$ and $\{(\bY_i, \bX_{2,i}): 1\leq i\leq n\}$, respectively (see \eqref{eq:def_eta_hat}). The following result establishes the consistency and rate of convergence of $\hat \eta_2$. 

\begin{theorem}\label{thm:eta2_convg}
    Suppose $\bY$ is not almost surely a constant. Then,
    \begin{enumerate}
        \item [$(1)$] (Consistency) If Assumption \ref{assumption:normX_continuous} holds for $\bX_1$, $\bX_2$, and $(\bX_1, \bX_2)$, and let $\E[\|\bY\|_2^{4+\delta}]<\infty$ for some $\delta>0$. Then as $n\ra\infty$,
        \begin{align*}
            \hat \eta_{2}\pto \eta_{2}.
        \end{align*}
        \item [$(2)$] (Rate of convergence) Suppose that Assumption \ref{assumption:normX_continuous} and Assumption \ref{assumption:rate_of_convg} are satisfied for $\bX_1$, $\bX_2$, and $(\bX_1, \bX_2)$. Then 
        \begin{align*}
            \left|\hat \eta_{2} - \eta_{2}\right| = O_P\left(\max\left\{\frac{(\log n)^2}{\sqrt{n}}, \frac{(\log n)^{1 + 1/d_1}}{n^{1/d_1}}, \frac{(\log n)^{1 + 1/d_2}}{n^{1/d_2}}, \frac{(\log n)^{1 + 1/(d_1 + d_2)}}{n^{1/(d_1 + d_2)}}\right\}\right) . 
        \end{align*}
    \end{enumerate}
\end{theorem}

The proof of Theorem~\ref{thm:eta2_convg} follows along similar lines to that of Theorem~\ref{thm:etan_consistent} and is therefore omitted. Specifically, consistency follows by a straightforward adaptation of the proof of Theorem~\ref{thm:etan_consistent}, establishing convergence for each of the three terms in \eqref{eq:def_hat_eta_2}. Likewise, the convergence rate can be derived by applying the argument of Proposition~\ref{prop:rate_nn_est} to each of the three terms in \eqref{eq:def_hat_eta_2} and subsequently invoking Lemma~\ref{lemma:ratio_Op}, yields the desired result. Finally, as already mentioned at the beginning of this section, all results extend to arbitrary collections of input variables and higher-order interactions. 


\section{Simulations and Real Data Analysis}
\label{sec:experiments}

In this section, we evaluate the empirical performance of the proposed methods in various  experimental settings and benchmark datasets.

\subsection{Testing Conditional Mean Independence}\label{sec:cmi_test}

In this section, we compare the conditional mean independence test in \eqref{eq:conditional_mean_test} with several related methods in terms of Type I error, power, and computational time. For the proposed test (referred to as \texttt{NCMD} in the figures) we use the number of nearest neighbors $K\in\{5,10\}$. For comparison, we consider tests based on Martingale Difference Divergence (MDD) \citep{shao2014martingale} implemented using the multiplier bootstrap (as in \cite{lee2020testing}) with $B=200$ bootstrap samples, Distance Covariance (dCov) \citep{szekely2007measuring}, the Chatterjee correlation \citep{chatterjee2021new}, and the partial Mean Independence Test (pMIT) \citep[Section~4]{cai2025test}. For pMIT we estimate the conditional mean using XGBoost \citep{chen2016xgboost} with two data-splitting schemes: $r=0.8$ and the data-driven choice of \cite{dai2022significance}. All experiments use sample size $n=250$, nominal level $\alpha=0.05$, and the empirical Type I error/power is calculated based on $200$ Monte Carlo repetitions. Note that among these methods, MDD and pMIT directly test conditional mean independence, while dCov and the Chatterjee correlation test for general statistical dependence.

\begin{example}\label{example:univariate}
We consider the following response models with univariate predictors. In each case $\varepsilon\sim N(0,1)$, $\lambda\ge0$ controls the noise level, and we observe i.i.d.\ samples $\{(Y_i,X_i)\}_{i=1}^n$ from the respective joint distributions. 
\begin{itemize}
    \item \texttt{Linear:} $Y = 0.5X + 3\lambda \varepsilon$.
    
    \item \texttt{Step:} $Y = s(X) + 10\lambda \varepsilon$, where 
    $s(x)=-3$, for $x\le-0.5$, $s(x)=2$, for $-0.5<x\le0$, $s(x)=-4$, for $0<x\le0.5$, and $s(x)=-3$, for $x>0.5$.

    \item \texttt{W-Shaped:} $Y = w(X) + 0.75\lambda \varepsilon$, where 
    $w(x)=|x+0.5|$ for $x<0$ and $w(x)=|x-0.5|$ for $x\ge0$.

    \item \texttt{Sinusoid:} $Y = \cos(8\pi X) + 3\lambda \varepsilon$.

    \item \texttt{Circular:} $Y = Z\sqrt{(1-X^2)_+} + 0.9\lambda \varepsilon$, where $Z\sim 2\,\mathrm{Bernoulli}(0.5)-1$.

    \item \texttt{Heteroskedastic:} $Y = 3(\bm{1}\{|X|\le0.5\}(2-\lambda)+\lambda )\varepsilon$.
\end{itemize}
Observe that when $\lambda=0$ the signal is noiseless, and increasing $\lambda$ weakens the dependence between $X$ and $Y$. The first four settings evaluate empirical power, while \texttt{Circular} and \texttt{Heteroskedastic} serve as Type-I error benchmarks for \eqref{eq:h0_formal}. The results are shown in Figure~\ref{fig:power_time_uniform}, with  $X\sim\mathrm{Uniform}[-1,1]$. We also repeat the experiments with $X\sim N(0,1)$ and $X=2U-1$ where $U\sim\mathrm{Beta}(\alpha,\alpha)$ with $\alpha=0.1$. 
The corresponding results shown in Figures~\ref{fig:power_time_normal} and \ref{fig:power_time_beta} in Appendix~\ref{appendix:cmi_test}, respectively.
\end{example}

\begin{figure}[t]
    \centering
    \begin{minipage}{0.6\linewidth}
        \centering
        \includegraphics[width=\linewidth]{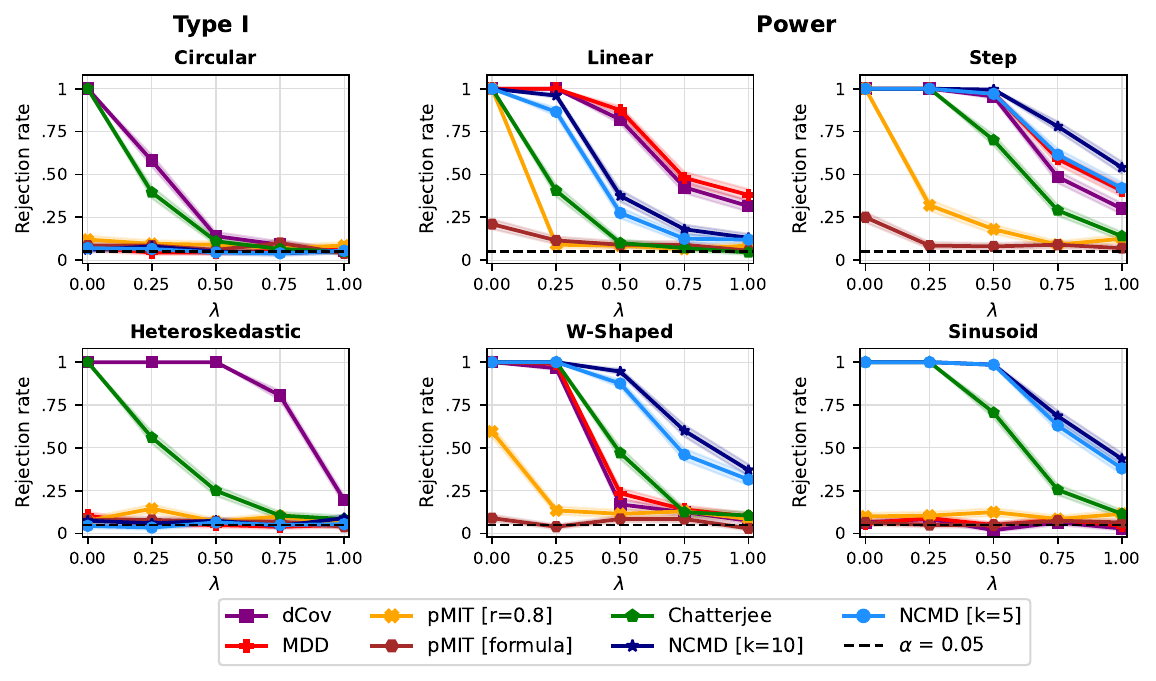} 
        \small{ (a) }  
    \end{minipage}%
    \begin{minipage}{0.4\linewidth}
        \centering
        \includegraphics[width=\linewidth]{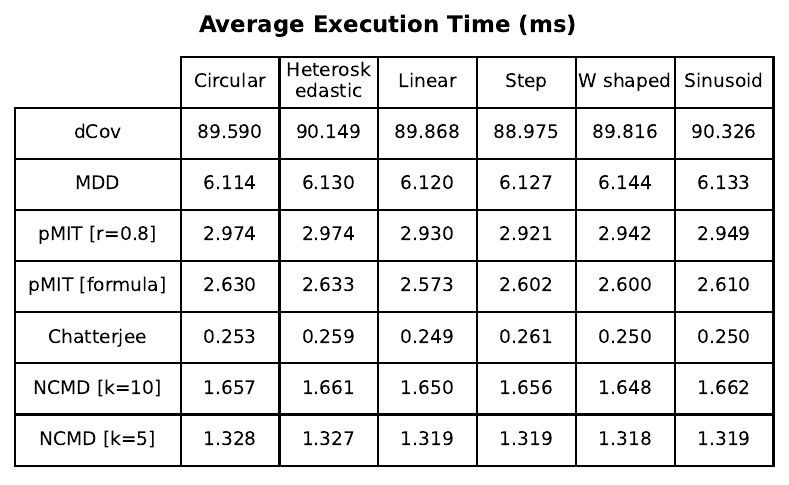}
        \small{ (b) } 
    \end{minipage}

    \caption{\small{Empirical Type I error/power and (b) computational time for conditional mean independence testing for the settings in Example \ref{example:univariate}, when $X\sim\mathrm{Uniform}[-1,1]$.}}
    \label{fig:power_time_uniform}
\end{figure}

\begin{example}\label{example:multivariate} 
Next, we consider the following models with a multivariate predictor $$\bX=(X_1,\ldots,X_d)\in\mathbb{R}^d$$ and a scalar response $Y$. Throughout, we take $d=10$, and assume that the coordinates of $\bX$ are generated independently from a common underlying distribution, which will be specified below. As before, $\varepsilon \sim N(0, 1)$, $\lambda \geq 0$ controls the noise level, and we observe i.i.d.\ samples $\{(Y_i,\bX_i)\}_{i=1}^n$ from the corresponding joint distribution.
\begin{itemize}

\item \texttt{Noise:} $Y = \lambda \varepsilon$.

\item \texttt{Heteroskedastic:} $Y = Z(1+2\|\bX\|_2^2) + \lambda \varepsilon$, where $Z\sim\mathrm{Uniform}\{-1,1\}$.

\item \texttt{Nonlinear Additive:} $Y = \sin(\pi X_1) + \log(|X_2|+1) + \lambda \varepsilon$.

\item \texttt{Interaction:} $Y = X_1X_2 + \lambda \varepsilon$.

\item \texttt{Radial:} $Y = \cos(r) + \lambda \varepsilon$, where $r=\frac{1}{\sqrt{5}}\|\bX_{S}\|_2$ and $S \subset[d]$ is chosen uniformly at random among subsets of size $5$.

\item \texttt{Nonlinear Interaction:} $Y = \sin(X_1) + \cos(X_2)X_3 + \lambda \varepsilon$.

\end{itemize}
\begin{figure}[t]
    \centering
    \begin{minipage}{0.6\linewidth}
        \centering
        \includegraphics[width=\linewidth]{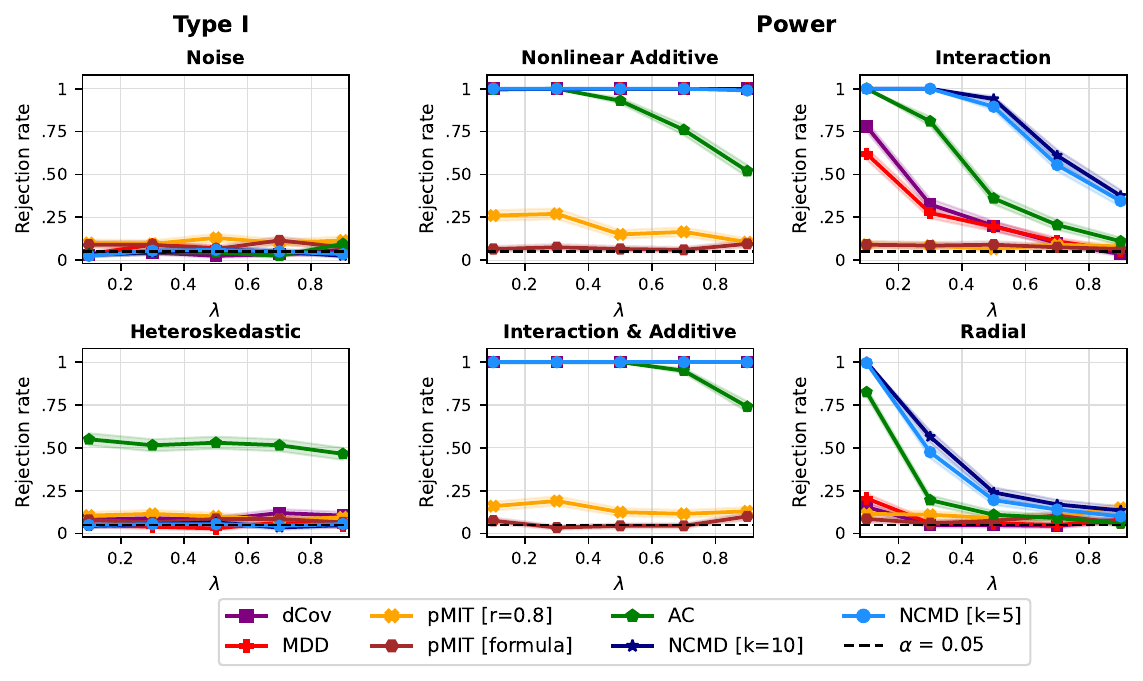} 
        \small{ (a) }  
    \end{minipage}%
    \begin{minipage}{0.4\linewidth}
        \centering
        \includegraphics[width=\linewidth]{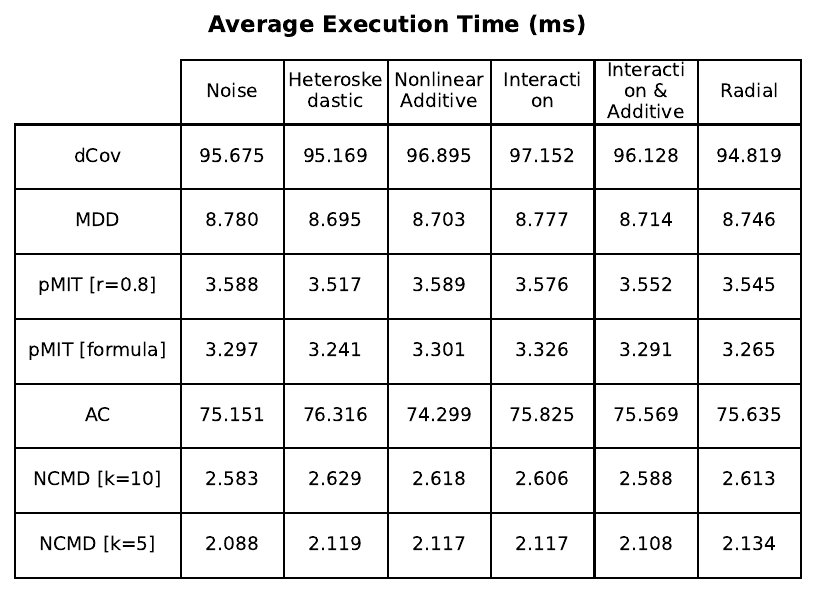} 
        \small{ (b) }  
    \end{minipage}

    \caption{\small{Empirical Type I error/power and (b) computational time for conditional mean independence testing for the settings in Example \ref{example:multivariate}, when $X\sim\mathrm{Uniform}([-1,1]^{10})$.}} 
    \label{fig:power_time_uniform_multivariate}
\end{figure}
Note that for the \texttt{Nonlinear Additive}, \texttt{Interaction}, \texttt{Radial}, and \texttt{Nonlinear Interaction} models, $Y$ depends only on the first few or a random subset of coordinates of $\mathbf{X}$, while the remaining coordinates are pure noise. Also, observe that \texttt{Noise} and \texttt{Heteroskedastic} settings correspond to Type-I error benchmarks and the remaining four cases is for power comparison. The comparison is performed with the same collection of methods as in Example \ref{example:univariate}. The only difference is that Chatterjee's correlation (which is a measure of association between two univariate random variables) is replaced by the Azadkia--Chatterjee coefficient \citep{azadkia2021simple}, implemented using a permutation test.
The results are displayed in Figure~\ref{fig:power_time_uniform_multivariate}, when the coordinates of $\bX$ are generated independently from $\mathrm{Uniform}[-1,1]$. The corresponding results for Gaussian and Beta distributed covariates are given in Figures~\ref{fig:power_time_normal_multivariate} and \ref{fig:power_time_beta_multivariate} in Appendix~\ref{appendix:cmi_test}, respectively.
\end{example} 

The following are the key findings from our experiments: 

\begin{itemize} 

\item Overall, the proposed \texttt{NCMD} test achieves higher empirical power than the competing methods in most of the simulation settings considered (see Figures~\ref{fig:power_time_uniform}(a), \ref{fig:power_time_uniform_multivariate}(a), as well as Figures~\ref{fig:power_time_normal}(a), \ref{fig:power_time_beta}(a), \ref{fig:power_time_normal_multivariate}(a), and \ref{fig:power_time_beta_multivariate}(a)  in Appendix~\ref{appendix:cmi_test}). The only exceptions are the \texttt{Linear} model and the \texttt{Nonlinear Additive Model} (when the covariates are Beta distributed), in which case the test based on MDD performs best, followed closely by the \texttt{dCov}-based test.

\item In the setting of Example \ref{example:univariate}, \texttt{Chatterjee} and \texttt{dCov} fail to control the Type I error in the \texttt{Circular} and \texttt{Heteroskedastic} models, since $Y$ remains dependent on $X$ even though $\mathbb{E}[Y\mid X]=\mathbb{E}[Y]$. By contrast, the conditional mean independence tests (\texttt{MDD}, \texttt{pMIT}, and \texttt{NCMD}) maintain correct Type I error control. A similar phenomenon is observed in Example \ref{example:multivariate}. In particular, under the \texttt{Noise} model, where $Y$ and $\bX$ are independent, all methods control the Type I error at the nominal level. However, the \texttt{Heteroskedastic Null} model satisfies conditional mean independence but not full independence. In this case, as expected, only the conditional mean independence tests (\texttt{MDD}, \texttt{pMIT}, and \texttt{NCMD}) provide valid Type I error control.

\item In terms of computational time, the proposed method is substantially faster than most of the competing procedures. This is because our test admits an efficient near-linear-time implementation, while the \texttt{MDD} and \texttt{dCov} based procedures require considerably more expensive permutation or bootstrap calibration (recall the discussion in Remark \ref{remark:distributioncomparison}). The \texttt{pMIT} test, which is also based on a simple asymptotic null distribution, is considerably faster than the resampling-based methods, but it is still slower than \texttt{NCMD} because it requires estimating the conditional mean function using a machine-learning method (here implemented using \texttt{XGBoost}). It is also worth noting that, in the univariate setting, the test based on Chatterjee's correlation is faster than \texttt{NCMD} (see Figure~\ref{fig:power_time_uniform}(b), and Figures~\ref{fig:power_time_normal}(b), and \ref{fig:power_time_beta}(b) in Appendix~\ref{appendix:cmi_test}). This is because it is likewise based on an asymptotically Gaussian test statistic and does not require estimating the asymptotic variance. In the multivariate setting, however, the test based on the Azadkia--Chatterjee coefficient implemented via a permutation procedure. As a result, the \texttt{NCMD} methods are substantially faster than the Azadkia--Chatterjee method in multivariate settings. 

\end{itemize}

\subsection{Variable Screening}\label{sec:variable_screening}

In this section, we evaluate the performance of the nearest neighbor based variable screening (\texttt{NNVS}) algorithm described in Algorithm~\ref{alg:greedy_screening} both in a simulation setting and on a real dataset. We compare the proposed \texttt{NNVS} algorithm with the following existing methods: \texttt{KFOCI} (with Gaussian kernel) \citep{huang2022kernel}, \texttt{MDCSIS} \citep{shao2014martingale}, \texttt{BcorSIS} \citep{pan2019generic}, and \texttt{Kfilter} \citep{mai2015fused}. \texttt{BcorSIS} and \texttt{Kfilter} are implemented using the R package \texttt{MFSIS} \citep{cheng2023generic}.

\subsubsection{Simulations}

We consider the following simulation settings. Generate $\mathbf{X} = (X_1, \ldots, X_d) \sim  N_d(\mathbf{0}, I_d)$, with $d \in \{10,25\}$. Then the response $Y$ (which only depends on $X_1$, $X_2$, and $X_3$) is obtained as follows. As before, $\varepsilon \sim N(0, 1)$. 

\begin{itemize}
    \item[(S1)] $Y = X_1 X_2 + X_1 - X_3 + \varepsilon$.
    
    \item[(S2)] $Y = \sin(X_1) + \cos(X_2) X_3 + \varepsilon$.
    
    \item[(S3)] In this setting, 
    $$Y =
    \begin{cases}
        \cos(X_1) + \sin(X_3) + \varepsilon, & X_2 < 0 , \\
        \sin(X_1) + \cos(X_3) + \varepsilon, & X_2 \ge 0 . 
    \end{cases} $$
    
    \item[(S4)] In this setting, 
    $$Y =
    \begin{cases}
        \cos(X_1)\exp(X_3) + \varepsilon, & X_2 < 0 , \\
        \sin(X_3)\exp(X_1) + \varepsilon, & X_2 \ge 0 . 
    \end{cases} $$ 
    
\end{itemize}
In our experiments we set $n = 300$ and use the number of nearest neighbors $K = 10$ for \texttt{NNVS} and \texttt{KFOCI} methods. The results are shown in Table~\ref{tab:screening_results}. Each cell reports, in order, the proportion of exact selections, the proportion of selections containing the correct set, and the average selection size, with each quantity computed from 100 independent repetitions of the experiment. For \texttt{MDCSIS}, \texttt{BcorSIS}, and \texttt{Kfilter}, the default number of selected variables is $n/\log n$. 
\begin{table}[t]
\small
\centering
\caption{\small{ Performance of the different variable screening algorithms in different settings. Each entry represents: exact selection/contains correct set/average selection size.} }
\label{tab:screening_results}
\begin{tabular}{|c|c|c|c|c|c|c|}
\hline
\textbf{Setting} & $d$ & \texttt{NNVS} & \texttt{KFOCI} & \texttt{MDCSIS} & \texttt{BcorSIS} & \texttt{Kfilter} \\
\hline
\multirow{2}{*}{(S1)} & 10 & $0.99/0.99/2.99$ & $1.00/1.00/3.00$ & $0.16/0.16/3.00$ & $0.95/0.95/3.00$ & $0.79/0.79/3.00$ \\
 & 25 & $0.99/1.00/3.01$ & $1.00/1.00/3.00$ & $0.08/0.08/3.00$ & $0.87/0.87/3.00$ & $0.65/0.65/3.00$ \\
\hline
\multirow{2}{*}{(S2)} & 10 & $0.86/0.94/3.03$ & $0.66/0.85/3.14$ & $0.15/0.15/3.00$ & $0.23/0.23/3.00$ & $0.22/0.22/3.00$ \\
 & 25 & $0.76/0.92/3.11$ & $0.51/0.75/3.26$ & $0.04/0.04/3.00$ & $0.13/0.13/3.00$ & $0.04/0.04/3.00$ \\
\hline
\multirow{2}{*}{(S3)} & 10 & $0.72/1.00/3.31$ & $0.60/0.99/3.45$ & $0.16/0.16/3.00$ & $0.19/0.19/3.00$ & $0.10/0.10/3.00$ \\
 & 25 & $0.45/0.97/3.69$ & $0.45/0.91/3.59$ & $0.03/0.03/3.00$ & $0.01/0.01/3.00$ & $0.07/0.07/3.00$ \\
\hline
\multirow{2}{*}{(S4)} & 10 & $0.90/0.95/3.01$ & $0.98/1.00/3.02$ & $0.54/0.54/3.00$ & $0.83/0.83/3.00$ & $0.69/0.69/3.00$ \\
 & 25 & $0.73/0.89/3.12$ & $0.84/1.00/3.16$ & $0.40/0.40/3.00$ & $0.74/0.74/3.00$ & $0.60/0.60/3.00$ \\
\hline
\multirow{2}{*}{(S5)} & 10 & $0.92/0.99/3.07$ & $0.47/1.00/3.92$ & $0.21/0.21/3.00$ & $0.88/0.88/3.00$ & $0.82/0.82/3.00$ \\
 & 25 & $0.87/1.00/3.13$ & $0.48/1.00/3.89$ & $0.05/0.05/3.00$ & $0.80/0.80/3.00$ & $0.68/0.68/3.00$ \\
\hline
\multirow{2}{*}{(S6)} & 10 & $0.48/0.80/3.25$ & $0.08/0.66/4.51$ & $0.02/0.02/3.00$ & $0.06/0.06/3.00$ & $0.08/0.08/3.00$ \\
 & 25 & $0.32/0.74/3.46$ & $0.10/0.77/4.80$ & $0.01/0.01/3.00$ & $0.04/0.04/3.00$ & $0.04/0.04/3.00$ \\
\hline
\end{tabular}
\end{table}
\normalsize 
From the results in Table~\ref{tab:screening_results}, we observe that \texttt{NNVS} and \texttt{KFOCI} outperform the other three methods (\texttt{MDCSIS}, \texttt{BcorSIS}, and \texttt{Kfilter}) across all settings. The performance of \texttt{NNVS} and \texttt{KFOCI} is comparable for (S1), whereas \texttt{NNVS} shows relatively better performance for (S2) and (S3), while \texttt{KFOCI} performs better in (S4). 

In addition, we also consider the following settings, in which the mean of $Y$ depends only on $X_1$, $X_2$, and $X_3$:
\begin{itemize}
    \item[(S5)] $Y = X_1X_2 + X_1 - X_3 + \varepsilon X_4X_5$.
    \item[(S6)] $Y = \sin(X_1) + \cos(X_2)X_3 + \varepsilon X_4X_5$.
\end{itemize}
The performance of all methods under settings (S5) and (S6) is also reported in Table \ref{tab:screening_results}. In both settings, \texttt{NNVS} achieves the highest exact selection proportion among the competing methods. This is because \texttt{NNVS} directly targets mean dependence, in contrast to methods such as \texttt{KFOCI}, \texttt{BcorSIS}, and \texttt{Kfilter}, which are based on measures of general dependence. Interestingly, \texttt{NNVS} also performs substantially better than \texttt{MDCSIS}, which is also a mean-dependence-based screening method, in these two settings.


\subsubsection{California Housing Dataset}

 In this section we apply our variable screening algorithm to the California Housing Dataset available in \texttt{sklearn}. This dataset contains $20{,}640$ observations describing housing districts in California from the 1990 U.S. Census. The response variable is the median house value (\texttt{MedHouseVal}), and the dataset includes eight predictive features: median income (\texttt{MedInc}), house age (\texttt{HouseAge}), average number of rooms (\texttt{AveRooms}), average number of bedrooms (\texttt{AveBedrms}), population (\texttt{Population}), average occupancy (\texttt{AveOccup}), latitude (\texttt{Latitude}), and longitude (\texttt{Longitude}).  To increase dimensionality, we augment the dataset as follows: we randomly select $7$ features and add Gaussian noise with standard deviation $\sigma \in \{0.1, 0.5, 1\}$, and generate $7$ additional variables as linear combinations of the original features. Finally, we add $500$ pure Gaussian noise variables, resulting in a total of $522$ features. We then apply variable screening methods mentioned above to this dataset. To this end, we split the dataset equally. Variable selection is performed on a random sample of $2,000$ observations drawn from the first half of the data, while the second half is used to fit an XGBoost model based on the selected variables. For model evaluation, we use $80\%$ of the second half for training and the remaining $20\%$ to report prediction performance in terms of mean squared error (MSE). For the \texttt{NNVS} method, we use the number of nearest neighbors $K \in \{10, 25\}$, and for \texttt{KFOCI}, we use $K= 10$. For \texttt{MDCSIS}, we use its default choice of selecting $n/\log n$ variables (reported as \texttt{MDCSIS} (own)). Moreover, for \texttt{BcorSIS}, \texttt{Kfilter}, and \texttt{MDCSIS} (match), the number of selected variables is fixed to match the maximum number of variables selected by \texttt{NNVS} and \texttt{KFOCI}.

\begin{table}[t]
\centering
\scriptsize
\setlength{\tabcolsep}{3pt}

\begin{minipage}{0.48\textwidth}
\centering
\begin{tabular}{lcccccc}
\toprule
 & \texttt{NNVS} & \texttt{NNVS} & \texttt{KFOCI} & \texttt{MDCSIS} & \texttt{BcorSIS} & \texttt{Kfilter} \\
\texttt{Feature} & $K{=}10$ & $K{=}25$ &  & (match) &  &  \\
\midrule
\texttt{MedInc}      & \checkmark & \checkmark & \checkmark & \checkmark & \checkmark & \checkmark \\
\texttt{Longitude}   & \checkmark & \checkmark & \checkmark &            & \checkmark &            \\
\texttt{Latitude}    & \checkmark & \checkmark & \checkmark &            &            &            \\
\texttt{AveOccup}    & \checkmark & \checkmark & \checkmark & \checkmark & \checkmark & \checkmark \\
\texttt{AveRooms}    &            &            &            & \checkmark & \checkmark & \checkmark \\
\texttt{AveBedrms}   &            & \checkmark &            &            &            &            \\
\midrule
\texttt{n\_MedInc}   &            &            &            & \checkmark & \checkmark & \checkmark \\
\texttt{n\_AveRooms} &            &            &            & \checkmark & \checkmark &            \\
\texttt{n\_AveOccup} &            & \checkmark &            & \checkmark &            &            \\
\texttt{n\_Latitude} &            &            &            &            &            & \checkmark \\
\texttt{n\_AveBedrms}&            &            &            &            &            & \checkmark \\
\bottomrule
\end{tabular}
\caption*{ \small{(a)} } 
\end{minipage}
\hspace{10pt}
\begin{minipage}{0.48\textwidth}
\centering
\scriptsize
\begin{tabular}{lccccc}
\toprule
Method & Selected & Real & Noise & MSE\\
\midrule
\texttt{NNVS} ($K{=}10$) & 4 & 4 & 0 & .2254\\
\texttt{NNVS} ($K{=}25$) & 6 & 5 & 1 & .2296\\
\texttt{KFOCI}           & 4 & 4 & 0 & .2254\\
\texttt{BcorSIS}         & 6 & 4 & 2 & .3564\\
\texttt{Kfilter}         & 6 & 3 & 3 & .4698\\
\texttt{MDCSIS} (match)     & 6 & 3 & 3 & .4944\\
\texttt{MDCSIS} (own)       & 263 & 8 & 255 & .2549\\
Oracle                   & 8 & 8 & 0 & .2036\\
\bottomrule
\end{tabular}
\caption*{ \small{(b)} }
\end{minipage}
\caption{\small{Variable screening in the California Housing Dataset: (a) features selected by each method with $\sigma = 0.5$, and (b) prediction performance of each method with $\sigma = 0.5$. }}
\label{table:variablesigma}
\end{table}

Table \ref{table:variablesigma} reports the results when $\sigma = 0.5$. In Table~\ref{table:variablesigma} (a), we show the variables selected by each method on the augmented California Housing Dataset. Notice that \texttt{NNVS} with $K=10$ and \texttt{KFOCI} select the same variables, whereas \texttt{NNVS} with $K=25$ selects two additional variables, one of which is an augmented noisy variable. The other three methods, \texttt{MDCSIS}, \texttt{Kfilter}, and \texttt{BcorSIS}, select six variables (matching the maximum number selected by \texttt{NNVS} and \texttt{KFOCI}), among which there are at least two noisy augmented variables.  In Table~\ref{table:variablesigma} (b), we report the mean squared error (MSE) of predictions using the selected variables. \texttt{NNVS} and \texttt{KFOCI} achieve the lowest prediction errors among all screening methods, significantly outperforming the others. This is expected, as the other methods include additional noisy features that reduce prediction accuracy.

The results for $\sigma = 0.1$ and $\sigma = 1$ are shown in Table \ref{table:variablesigma2} and Table \ref{table:variablesigma3}, respectively, in Appendix \ref{sec:variablesigma}. Consistent with the results above, the prediction performance achieved using variables selected by \texttt{NNVS} and \texttt{KFOCI} significantly outperforms \texttt{MDCSIS}, \texttt{Kfilter}, and \texttt{BcorSIS}. The latter methods tend to select a higher proportion of noisy variables, leading to degraded accuracy. Specifically, for $\sigma = 0.1$, \texttt{NNVS} identifies $4$ variables (including $1$ noise variable), while \texttt{KFOCI} selects $5$ variables (including the same noise variable). Notably, despite the similarity in selected features, the prediction accuracy using the \texttt{NNVS} subset is slightly superior to that of \texttt{KFOCI}. 

\small

\subsection*{Acknowledgements} B. B. Bhattacharya was supported by NSF CAREER grant DMS 2046393 and a Sloan Research Fellowship.

\subsection*{Use of Generative AI} During the preparation of this work, the authors used AI tools for assistance with coding and proofreading the manuscript. The authors take full responsibility for the final content of the manuscript.

\small

\bibliography{ref}
\bibliographystyle{abbrvnat}


\appendix

%
%

\normalsize 

\newpage

\section{Additional Experimental Results}

\subsection{Conditional Mean Independence}\label{appendix:cmi_test}

In this section we present additional results to complement the experiments in Section \ref{sec:cmi_test}. We consider the same settings as in Section \ref{sec:cmi_test}, but with different\ covariate distributions. In particular, for Example \ref{example:univariate} we consider $X \sim  N(0,1)$ and $X = 2U-1$, where $U \sim \mathrm{Beta}(0.1,0.1)$ (which concentrates mass near the boundaries of $[-1,1]$). 
The results are presented in Figure \ref{fig:power_time_normal} and Figure \ref{fig:power_time_beta}. In the setting of Example \ref{example:multivariate}, where we have a multivariate predictor $\bm{X} = (X_1,\ldots,X_{10})$, each coordinate is generated independently from the above distributions. The results are reported in Figures~\ref{fig:power_time_normal_multivariate} and \ref{fig:power_time_beta_multivariate}. These additional experiments confirm the trends observed in the main text: the proposed \texttt{NCMD} test maintains superior empirical power across most settings, controls Type-I error in the appropriate null scenarios, and is computationally faster than competing methods in the multivariate setting.

\begin{figure}[h]
    \centering
    \begin{minipage}{0.6\linewidth}
        \centering
        \includegraphics[width=\linewidth]{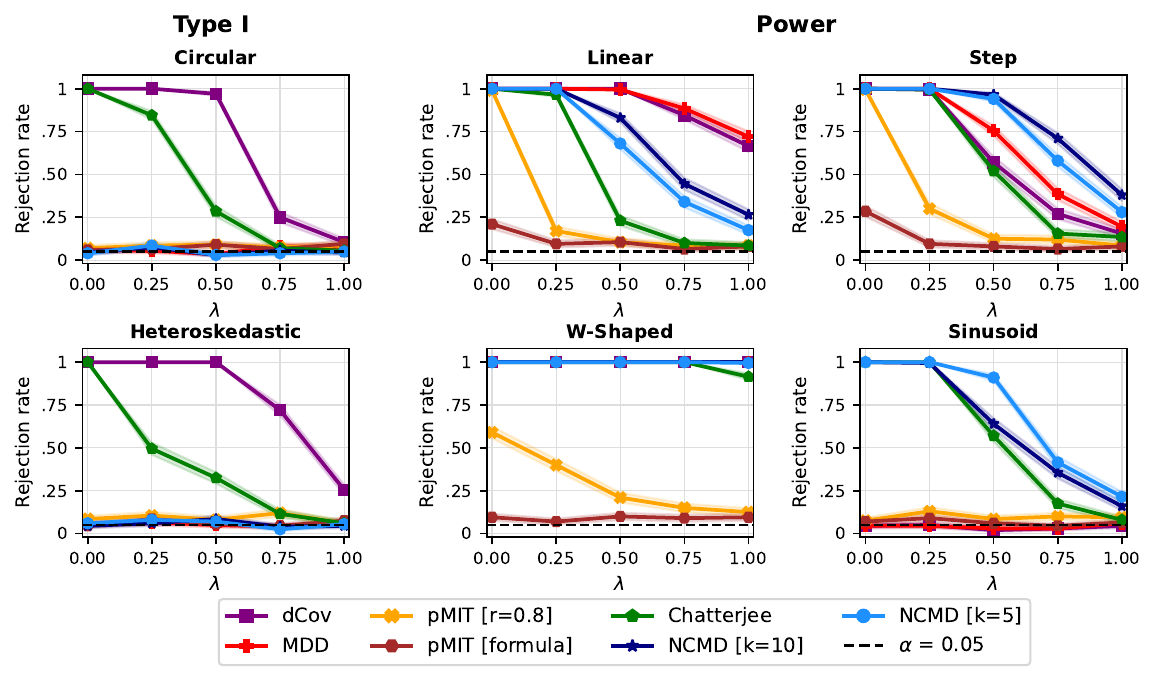}
        \small{(a)}
    \end{minipage}%
    \begin{minipage}{0.4\linewidth}
        \centering
        \includegraphics[width=\linewidth]{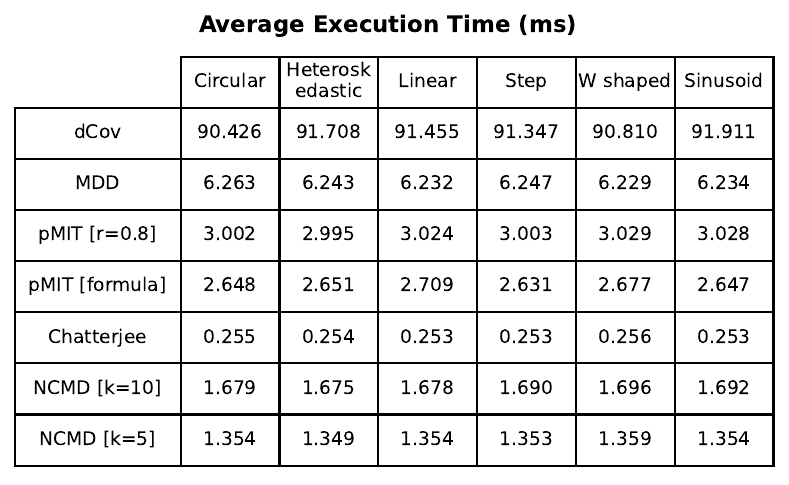}
        \small{(b)}
    \end{minipage} 
    
        \caption{\small{(a) Empirical Type I error/power and (b) computational time for conditional mean independence testing for the settings in Example \ref{example:univariate}, $X\sim N(0,1)$.}}

    \label{fig:power_time_normal}
\end{figure}

\begin{figure}[h]
    \centering
    \begin{minipage}{0.6\linewidth}
        \centering
        \includegraphics[width=\linewidth]{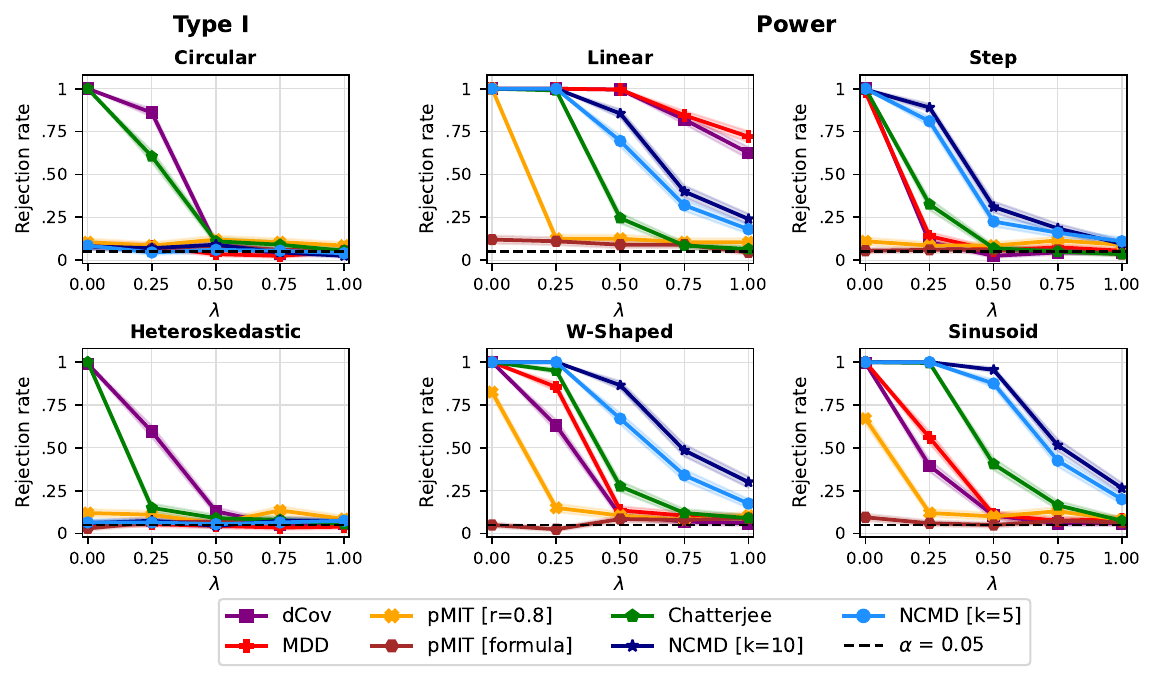}
        \small{(a)}
    \end{minipage}%
    \begin{minipage}{0.4\linewidth}
        \centering
        \includegraphics[width=\linewidth]{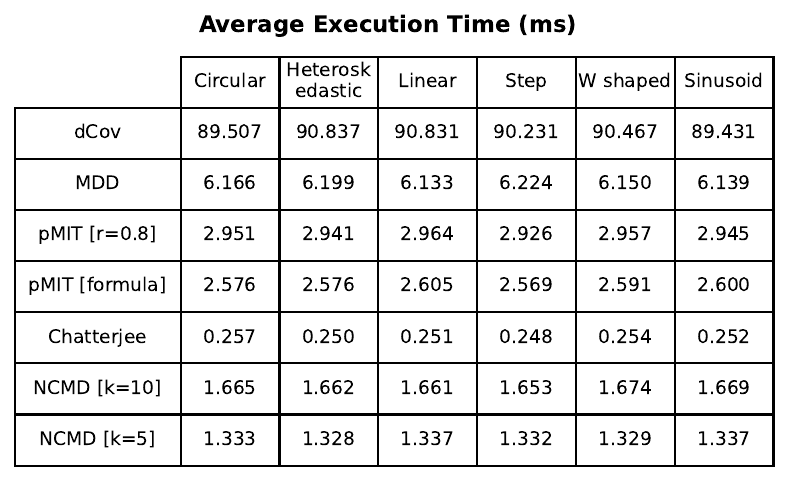}
        \small{(b)}
    \end{minipage}
    \caption{\small{(a) Empirical Type I error/power and (b) computational time for conditional mean independence testing for the settings in Example \ref{example:univariate}, when $X = 2U-1$ and $U \sim \mathrm{Beta}(0.1,0.1)$.}} 
    \label{fig:power_time_beta}
\end{figure}

\begin{figure}[h]
    \centering
    \begin{minipage}{0.6\linewidth}
        \centering
        \includegraphics[width=\linewidth]{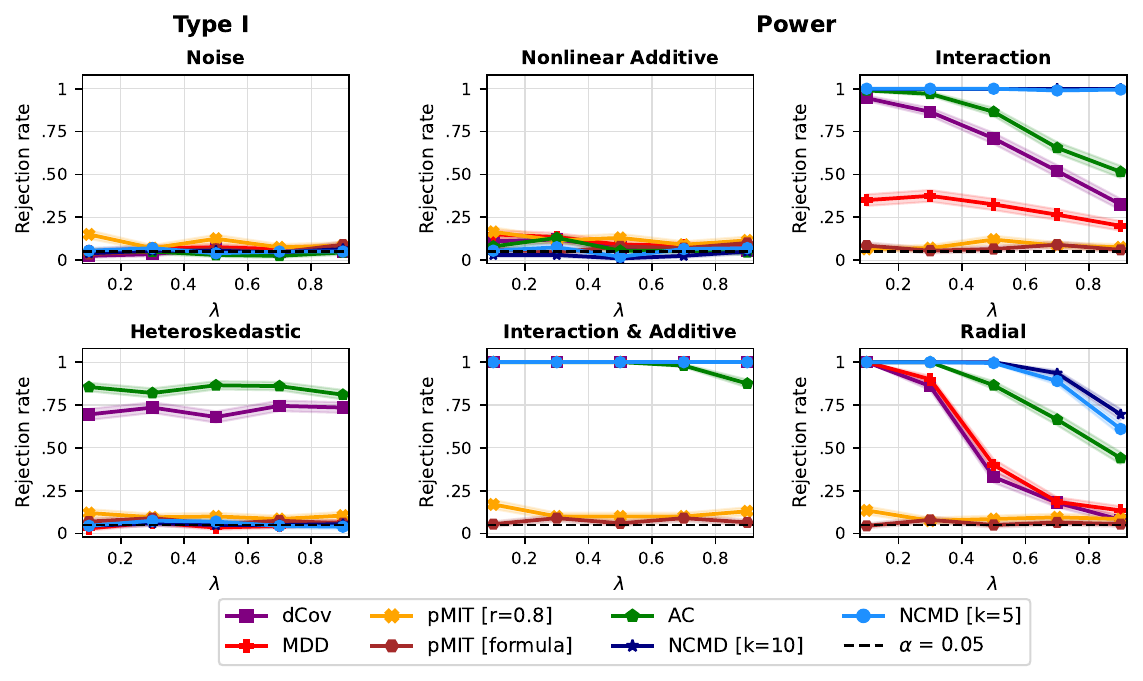}
        \small{(a)}
    \end{minipage}%
    \begin{minipage}{0.4\linewidth}
        \centering
        \includegraphics[width=\linewidth]{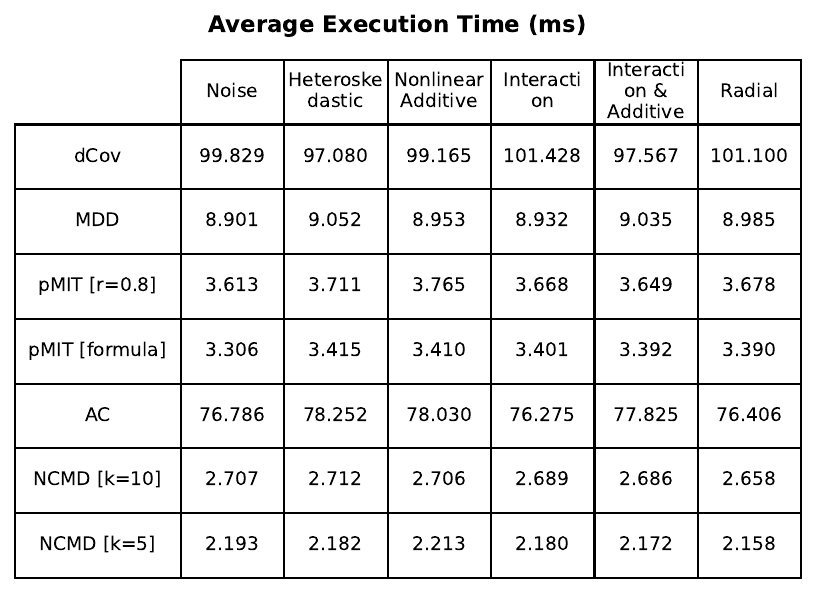}
        \small{(b)}
    \end{minipage} 
        \caption{\small{(a) Empirical Type I error/power and (b) computational time for conditional mean independence testing for the settings in Example \ref{example:multivariate}, when $\bX\sim  N_{10}(\bm 0, \bm{I}_{10})$.}} 
    \label{fig:power_time_normal_multivariate}
\end{figure}

\begin{figure}[h]
    \centering
    \begin{minipage}{0.6\linewidth}
        \centering
        \includegraphics[width=\linewidth]{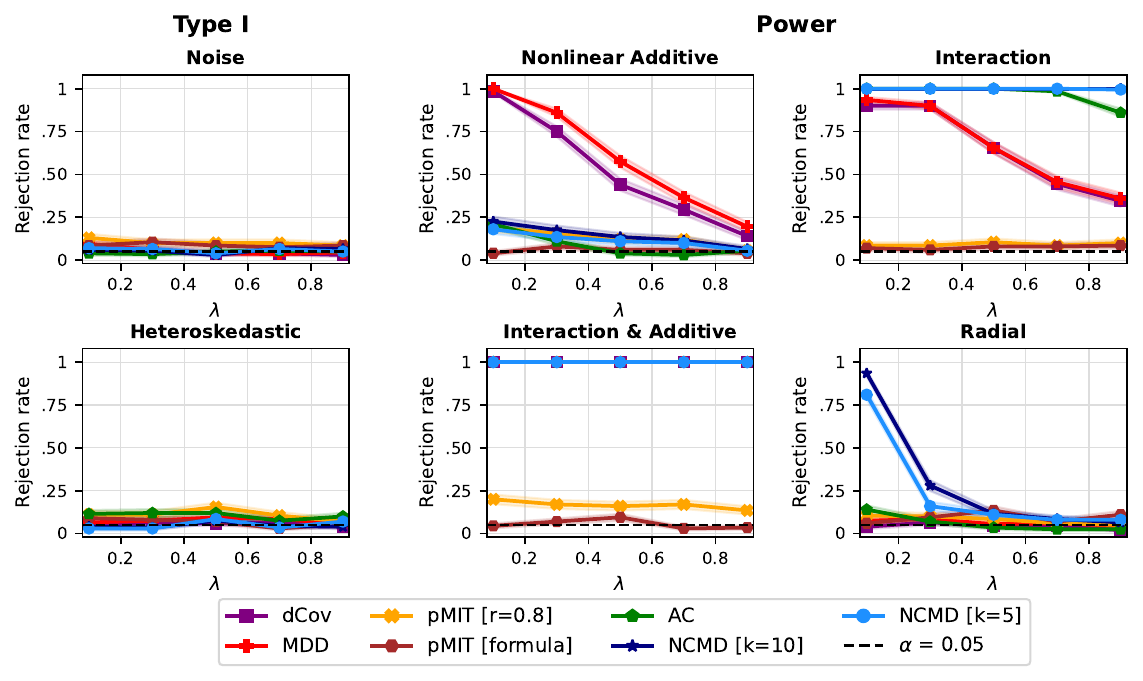}
    \end{minipage}%
    \begin{minipage}{0.4\linewidth}
        \centering
        \includegraphics[width=\linewidth]{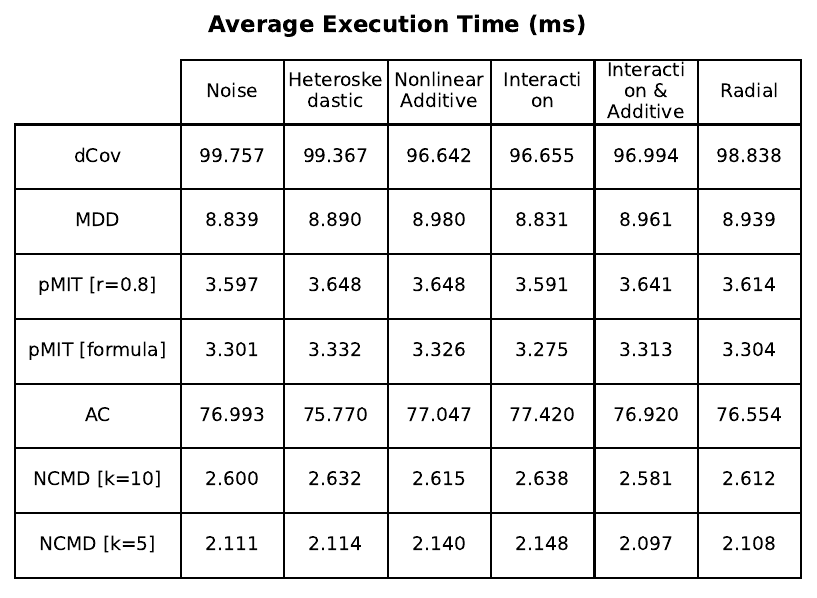}
    \end{minipage} 
            \caption{\small{(a) Empirical Type I error/power and (b) computational time for conditional mean independence testing for the settings in Example \ref{example:multivariate}, when $\bX = 2\bm{U}-1$ and the coordinates of $\bm{U}$ are i.i.d. $\mathrm{Beta}(0.1,0.1)$.}} 
    \label{fig:power_time_beta_multivariate}
\end{figure}

\subsection{Variable Screening for California Housing Dataset}
\label{sec:variablesigma}

In this section, we present additional experimental results comparing the variable screening methods on the California Housing dataset considered in Section~\ref{sec:variable_screening}. Specifically, Tables~\ref{table:variablesigma2} and \ref{table:variablesigma3} report the results for $\sigma = 0.1$ and $\sigma = 1$, respectively, where $\sigma$ denotes the standard deviation of the noise variables, as described in Section~\ref{sec:variable_screening}. The findings in these tables reinforce the trends observed earlier: the predictive performance based on the variables selected by \texttt{NNVS} and \texttt{KFOCI} exceeds that of all competing methods. By contrast, the other methods tend to select a larger proportion of noisy variables, which results in a deterioration in prediction accuracy. Notably, when $\sigma = 0.1$, \texttt{NNVS} selects four variables, including one noise variable, whereas \texttt{KFOCI} selects five variables, again including the same noise variable. Interestingly, despite the substantial overlap between the two selected subsets, the variables chosen by \texttt{NNVS} yield slightly better predictive accuracy than those selected by \texttt{KFOCI}.

\begin{table}[t]
\centering
\scriptsize
\setlength{\tabcolsep}{2pt}
\begin{minipage}{0.48\textwidth}
\centering
\begin{tabular}{lcccccc}
\toprule
 & \texttt{NNVS} & \texttt{NNVS} & \texttt{KFOCI} & \texttt{MDCSIS} & \texttt{BcorSIS} & \texttt{Kfilter} \\
\texttt{Feature} & $K{=}10$ & $K{=}25$ &  & (match)  &  &  \\
\midrule
\texttt{MedInc}         & \checkmark & \checkmark & \checkmark & \checkmark & \checkmark &            \\
\texttt{HouseAge}       &            &            &            &            &            & \checkmark \\
\texttt{AveRooms}       &            &            &            & \checkmark & \checkmark &            \\
\texttt{AveBedrms}      &            & \checkmark &            &            &            & \checkmark \\
\texttt{AveOccup}       &            & \checkmark & \checkmark & \checkmark & \checkmark &            \\
\texttt{Latitude}       & \checkmark & \checkmark & \checkmark &            &            &            \\
\texttt{Longitude}      & \checkmark & \checkmark & \checkmark &            &            & \checkmark \\
\midrule
\texttt{n\_MedInc}      &            &            &            & \checkmark & \checkmark &            \\
\texttt{n\_AveRooms}    &            &            &            & \checkmark & \checkmark &            \\
\texttt{n\_AveOccup}    & \checkmark & \checkmark & \checkmark & \checkmark & \checkmark & \checkmark \\
\texttt{n\_Population}  &            &            &            &            &            & \checkmark \\
\texttt{n\_AveBedrms}   &            &            &            &            &            & \checkmark \\
\bottomrule
\end{tabular}
\caption*{ \small{(a)} }
\end{minipage}
\hspace{5pt}
\begin{minipage}{0.48\textwidth}
\centering
\scriptsize 
\begin{tabular}{lccccc}
\toprule
Method & Selected & Real & Noise & MSE \\
\midrule
\texttt{NNVS} ($K{=}10$) & 4 & 3 & 1 & .2206 \\
\texttt{NNVS} ($K{=}25$) & 6 & 5 & 1 & .2292 \\
\texttt{KFOCI}            & 5 & 4 & 1 & .2246 \\
\texttt{BcorSIS}          & 6 & 3 & 3 & .4958 \\
\texttt{Kfilter}          & 6 & 3 & 3 & .7457 \\
\texttt{MDCSIS} (match)      & 6 & 3 & 3 & .4972 \\
\texttt{MDCSIS} (own)        & 263 & 8 & 255 & .2533 \\
Oracle                    & 8 & 8 & 0 & .2036 \\
\bottomrule
\end{tabular}
\caption*{ \small{(b)} }
\end{minipage}
\caption{\small{Variable screening in the California Housing Dataset: (a) features selected by each method with $\sigma = 0.1$ and (b) prediction performance of each method with $\sigma = 0.1$. }}  
\label{table:variablesigma2}
\end{table}

\begin{table}[t]
\centering
\scriptsize
\setlength{\tabcolsep}{2pt}
\begin{minipage}{0.48\textwidth}
\centering
\begin{tabular}{lcccccc}
\toprule
 & \texttt{NNVS} & \texttt{NNVS} & \texttt{KFOCI} & \texttt{MDCSIS} & \texttt{BcorSIS} & \texttt{Kfilter} \\
\texttt{Feature} & $K{=}10$ & $K{=}25$ &  & (match) &  &  \\
\midrule
\texttt{MedInc}         & \checkmark & \checkmark & \checkmark & \checkmark & \checkmark & \checkmark \\
\texttt{AveRooms}       &            &            &            & \checkmark & \checkmark & \checkmark \\
\texttt{AveBedrms}      &            & \checkmark &            &            &            &            \\
\texttt{AveOccup}       & \checkmark & \checkmark & \checkmark & \checkmark & \checkmark &            \\
\texttt{Latitude}       & \checkmark & \checkmark & \checkmark &            &            &            \\
\texttt{Longitude}      & \checkmark & \checkmark & \checkmark &            & \checkmark &            \\
\midrule
\texttt{n\_MedInc}      &            &            &            & \checkmark & \checkmark &            \\
\texttt{n\_AveRooms}    &            &            &            & \checkmark &            &            \\
\texttt{n\_Longitude}   &            &            &            &            &            & \checkmark \\
\texttt{n\_Latitude}    &            &            &            &            &            & \checkmark \\
\texttt{n\_AveBedrms}   &            &            &            &            &            & \checkmark \\
\bottomrule
\end{tabular}
\caption*{ \small{(a)} }
\end{minipage}
\hspace{5pt}
\begin{minipage}{0.48\textwidth}
\centering
\scriptsize 
\begin{tabular}{lccccc}
\toprule
Method & Selected & Real & Noise & MSE \\
\midrule
\texttt{NNVS} ($K{=}10$) & 4 & 4 & 0 & .2254 \\
\texttt{NNVS} ($K{=}25$) & 5 & 5 & 0 & .2242 \\
\texttt{KFOCI}            & 4 & 4 & 0 & .2254 \\
\texttt{BcorSIS}          & 5 & 4 & 1 & .3673 \\
\texttt{Kfilter}          & 5 & 2 & 3 & .5864 \\
\texttt{MDCSIS} (match)      & 5 & 3 & 2 & .5153 \\
\texttt{MDCSIS} (own)        & 263 & 8 & 255 & .2532 \\
Oracle                    & 8 & 8 & 0 & .2036 \\
\bottomrule
\end{tabular}
\caption*{ \small{(b)} }
\end{minipage} 
\caption{\small{Variable screening in the California Housing Dataset: (a) features selected by each method with $\sigma = 1$ and (b) prediction performance of each method with $\sigma = 1$. }}
\label{table:variablesigma3}
\end{table}

\subsection{Sobol' Indices}\label{sec:sobol_experiment}
In this section, we present experimental results to illustrate the performance of the nearest neighbor-based estimators of Sobol' indices introduced in Section \ref{sec:connect_sobol} and Section \ref{sec:higher_sobol}. We consider a simple model with three inputs of varying dimensions and a scalar output. Specifically, we generate $\bX_1 = (X_{11}, X_{12})\sim \text{Unif}[-2,2]^2$, $\bX_2 = (X_{21}, X_{22}, X_{23}, X_{24})\sim \text{Unif}[-2,2]^4$, and $X_3\sim \text{Unif}[-2,2]$. The scalar output $Y$ is then generated as
\begin{equation}\label{eq:sobol_example}
Y = \theta \one^\top\bX_1 + \theta \one^\top\bX_2 + \theta X_3
+ (2-\theta)\one^\top\bX_1 X_3 + (2-\theta)\one^\top\bX_1 \one^\top\bX_2 X_3 + \varepsilon, 
\end{equation} 
where $\varepsilon\sim\mathcal{N}(0,0.04)$. For this model, we estimate the main effects $\eta_{\bX_1}$ and $\eta_{\bX_2}$, as well as the second-order interaction effect $\eta_{2}$ between $\bX_1$ and $X_3$, using the nearest neighbor--based estimators $\hat \eta_{\bX_1}$ and $\hat \eta_{\bX_2}$ from \eqref{eq:def_eta_hat}, and $\hat \eta_2$ from \eqref{eq:def_hat_eta_2}. Exploiting the mutual independence of the inputs, a direct calculation yields
\begin{align*}
\Var[Y] = \frac{252\theta^2 + 608(2-\theta)^2}{27} + 0.04.
\end{align*}
Recalling the definitions of $\eta_{\bX_1}$, $\eta_{\bX_2}$, and $\eta_2$, we obtain
\begin{align*}
\eta_{\bX_1} = \frac{72\theta^2}{252\theta^2 + 608(2-\theta)^2 + 1.08},\qquad
\eta_{\bX_2} = \frac{144\theta^2}{252\theta^2 + 608(2-\theta)^2 + 1.08},
\end{align*}
and
\begin{align*}
\eta_2 = \frac{96(2-\theta)^2}{252\theta^2 + 608(2-\theta)^2 + 1.08}.
\end{align*}
Figure~\ref{fig:sobol_experiment} compares the estimated indices $\hat \eta_{\bX_1}$, $\hat\eta_{\bX_2}$, and $\hat\eta_2$ with their corresponding population values $\eta_{\bX_1}$, $\eta_{\bX_2}$, and $\eta_2$.
Specifically, in Figure~\ref{fig:sobol_experiment}(a), we plot the estimated and true indices as functions of $\theta\in [0.5,1.5]$, with $n=1000$ and $K=5$. The estimates are averaged over 25 iterations, and $\pm1$ standard deviation error bars are shown. We observe that the estimates closely align with the true Sobol' indices across the range of $\theta$ values considered. In Figure~\ref{fig:sobol_experiment}(b), we illustrate the asymptotic convergence of the estimators by fixing $\theta=1$, $K=5$, and varying the sample size $n$. As before, the estimates are averaged over 25 iterations, and $\pm1$ standard deviation error bars are shown. In all three cases, the estimates converge to their population counterparts as the sample size increases, validating the results in Theorems~\ref{thm:etan_consistent} and~\ref{thm:eta2_convg}.


\begin{figure}[h]
    \centering
    \hspace{-20pt}
    \begin{minipage}{0.4\linewidth}
        \centering
        \includegraphics[scale = 0.7]{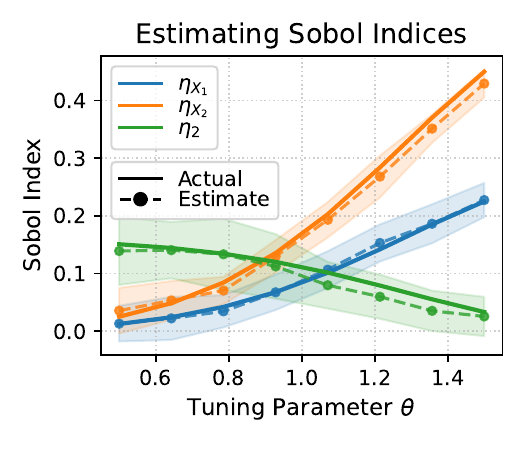}\\
        \small{(a)}
    \end{minipage}%
    \hspace{30pt}
    \begin{minipage}{0.4\linewidth}
        \centering
        \includegraphics[height=1\linewidth, width = 0.7\linewidth]{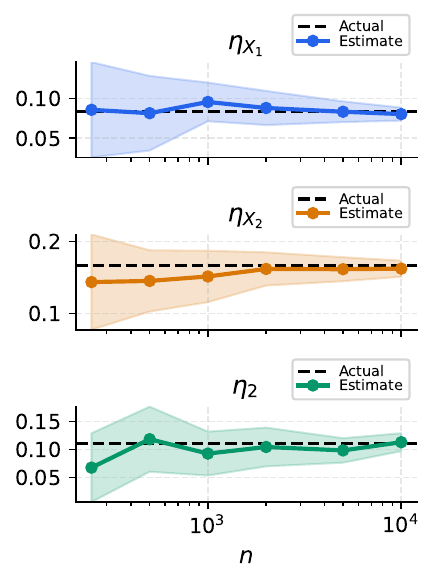}\\
        \small{(b)}
    \end{minipage} 
    \caption{\small{(a) True and estimated Sobol' indices as functions of $\theta\in [0.5,1.5]$ for model \eqref{eq:sobol_example} and (b) convergence of the estimated Sobol' indices at $\theta = 1$ for the model in \eqref{eq:sobol_example}.}}

    \label{fig:sobol_experiment}
\end{figure}

\section{Proof of Proposition \ref{ppn:normalized}} 
\label{sec:propertiespf}

The result in \ref{item:P1} follow directly from the decomposition: 
\begin{align*}
    \E\left[\big\|\bY - \E\left[\bY\right]\big\|_2^2\right] 
    = \E\left[\big\|\bY - \E\left[\bY\mid\bX\right]\big\|_2^2\right] 
    + \E\left[\big\|\E\left[\bY\mid\bX\right] - \E\left[\bY\right]\big\|_2^2\right].
\end{align*}
Property \ref{item:P2} follows immediately from the definition. Finally, for Property \ref{item:P3}, from the above decomposition we have, 
$$\eta = 1 \quad \Longleftrightarrow \quad  \E\left[\big\|\bY - \E\left[\bY\mid\bX\right]\big\|_2^2\right]  = 0 \quad \Longleftrightarrow \quad \bY \stackrel{a.s.}= \E\left[\bY\mid\bX\right].$$

\section{Proofs of Theorem \ref{thm:etan_consistent} and Theorem \ref{thm:rate_of_convg}}

We begin with the proof of Theorem \ref{thm:etan_consistent} in Appendix \ref{appendix:proofof_etan_consistent}. Then we prove Theorem \ref{thm:rate_of_convg} in Appendix \ref{sec:rate_of_convgpf}. Throughout the proofs, for two non-negative sequences $a_n$ and $b_n$, we use the notation $a_n \lesssim_{\square} b_n$ to denote that $a_n \le C(\square) \cdot b_n$, where $C(\square) > 0$ is a constant depending on the subscripted parameters.

\subsection{Proof of Theorem \ref{thm:etan_consistent}}\label{appendix:proofof_etan_consistent}

Recall the definition of $\hat\eta_n$ from \eqref{eq:def_eta_hat}. Then by the law of large numbers and recalling the decomposition from \eqref{eq:YY_simple} 
%
%
we can immediately conclude that the denominator $D_n\pto\E[\|\bY - \E[\bY]\|_2^2]$. Similarly, the second term in the numerator of $\hat\eta_n$ is a consistent estimate of $\left\|\E\left[\bY\right]\right\|_2^2$. Thus, recalling the decomposition from \eqref{eq:numerator_decomp}, to complete the proof it is enough to show,
\begin{align}\label{eq:nn_consistency}
    V_n := \frac{1}{n}\sum_{u=1}^{n} \frac{1}{K} \sum_{v \in N_{G\left(\sX_n\right)}(u)}\bY_u^\top \bY_v\pto \E\left[\left\|\E\left[\bY\mid\bX\right]\right\|_2^2\right].
\end{align}
To establish \eqref{eq:nn_consistency}, it suffices to show the following: 
\begin{align}\label{eq:consistencycondition}
\E[V_n]\ra \E[\|\E[\bY\mid\bX]\|_2^2] \quad \text{ and } \quad \Var[V_n] = o(1). 
\end{align}  
With this in mind, by exchangeability and recalling that $(\bY_1, \bX_1),\ldots, (\bY_n, \bX_n)$ are independent, we observe that
\begin{align}
    \E[V_n] = \E\left[\frac{1}{K}\sum_{u=1}^{n}\E\left[\bY_1\mid \bX_1\right]^\top\E\left[\bY_u\mid \bX_u\right]\one\left\{ u \in N_{G(\sX_n)}(1)\right\}\right],\label{eq:ETncondXn}
\end{align} 
Now, for $1\leq i\leq n$, define $g(\bX_u) = \E\left[\bY\middle|\bX_u\right]$ and suppose $N(1)$ is a vertex uniformly chosen from the neighbours of $1$ in the graph $G(\sX_n)$. Then recalling~\eqref{eq:ETncondXn} shows,
    \begin{align}
        \left|\E\left[V_n\right] - \E\left[\big\|g(\bX_1)\big\|_2^2\right]\right|
        &\leq \E\left[\frac{1}{K}\sum_{u=1}^{n}\left|g(\bX_1)^\top g(\bX_u) - \big\|g(\bX_1)\big\|_2^2\right|\one\left\{u \in N_{G(\sX_n)}(1)\right\}\right]\nonumber\\
        &= \E\left[\frac{1}{K}\sum_{u=1}^{n}\left|g(\bX_1)^\top \left(g(\bX_u) - g(\bX_1)\right)\right|\one\left\{u \in N_{G(\sX_n)}(1)\right\}\right]\nonumber\\
        &\leq \E\left[\left|g(\bX_1)^\top\left(g\left(\bX_{N(1)}\right) - g(\bX_1)\right)\right|\right]\nonumber\\
        &\leq \sqrt{\E\left[\big\|g(\bX_1)\big\|_2^2\right]}\sqrt{\E\left[\big\|g\left(\bX_{N(1)}\right) - g(\bX_1)\big\|_2^2\right]},\label{eq:unifoneneigh}
    \end{align}
    where the last inequality follows from Cauchy-Schwartz inequality. Now, by Lemma \ref{lemma:fXN1_bdd_fX1} notice that,
    \begin{align}\label{eq:1N1diff4bdd}
        \E\left[\big\|g\left(\bX_{N(1)}\right) - g(\bX_1)\big\|_2^4\right]\lesssim\E\left[\big\|g\left(\bX_{N(1)}\right)\big\|_2^4\right] + \E\left[\big\|g(\bX_1)\big\|_2^4\right]\lesssim \E\left[\big\|g(\bX_1)\big\|_2^4\right]<\infty,
    \end{align}
    where the finiteness follows from Jensen's inequality and the assumption $\E[\|\bY\|_2^{4+\delta}]<\infty$. The bound from \eqref{eq:1N1diff4bdd} implies that $\|g(\bX_{N(1)}) - g(\bX_1)\|_2^2$ is uniformly integrable. Then, applying \cite[Lemma D.3]{deb2020measuring} shows that $\E[\|g(\bX_{N(1)}) - g(\bX_1)\|_2^2] = o(1)$. Combined with the bound from \eqref{eq:1N1diff4bdd}, we conclude $\E[V_n]\ra \E[\|\E[\bY\mid \bX]\|_2^2]$. This proves the first assertion in \eqref{eq:consistencycondition}.

Now, we establish the second assertion in \eqref{eq:consistencycondition}. Towards that, \textcolor{black}{applying Lemma \ref{lemma:VarTn_bd} with $f(\bY_u,\bY_v) = \bY_u^\top\bY_v$} gives, 
    \begin{align}\label{eq:bound_var_Vn}
        \Var[V_n] \lesssim_{d}\frac{1}{n}\E\left[\max_{1\leq u\neq v\leq n}\left|\bY_u^\top \bY_v\right|^2\right]\leq \frac{1}{n}\E\left[\max_{1\leq u\leq n}\left\|\bY_u\right\|_2^4\right],
    \end{align}
    where the last inequality once again follows from the Cauchy-Schwartz inequality. For $\vep>0$ define $\vep_n := \vep n^{\frac{4}{4+\delta}}$. Then,
    \begin{align}
        \E\left[\max_{1\leq u\leq n}\left\|\bY_u\right\|_2^4\right]
        &\leq \vep_n + \E\left[\max_{1\leq u\leq n}\left\|\bY_u\right\|_2^4\one\left\{\max_{1\leq u\leq n}\left\|\bY_u\right\|_2^4>\vep_n\right\}\right]\nonumber\\
        &\leq \vep_n + n\int_{\vep_n}^{\infty}\P\left(\left\|\bY_1\right\|_2^4\geq t\right)\mathrm dt\tag*{ (by union bound) }\nonumber\\
        &\leq \vep_n + n\int_{\vep_n}^{\infty}\frac{\E\left[\left\|\bY_1\right\|_2^{4+\delta}\one\{\|\bY_1\|_2^4\geq \vep_n\}\right]}{t^{1+\frac{\delta}{4}}}\mathrm dt\tag*{ (by Markov's Inequality) }\nonumber\\
        &\lesssim \vep_n +  \vep_n\frac{\E\left[\left\|\bY_1\right\|_2^{4+\delta}\one\{\|\bY_1\|_2^4\geq \vep_n\}\right]}{\vep^{1+\frac{\delta}{4}}}.\nonumber
    \end{align}
Hence by DCT,
\begin{align}
    \lim_{n\ra\infty}n^{-\frac{4}{4+\delta}}\E\left[\max_{1\leq u\leq n}\left\|\bY_u\right\|_2^4\right]\lesssim\lim_{n\ra\infty}\vep\left(1 + \frac{\E\left[\left\|\bY_1\right\|_2^{4+\delta}\one\{\|\bY_1\|_2^4\geq \vep_n\}\right]}{\vep^{1+\frac{\delta}{4}}}\right) = \vep.\label{eq:max_Y_4_bd}
\end{align}
Since $\vep>0$ was arbitrary, using the limit from \eqref{eq:max_Y_4_bd} gives, $\frac{1}{n}\E[\max_{1\leq u\leq n}\|\bY_u\|_2^4] = o(1)$, implying that $\Var[V_n] = o(1)$.

\subsection{Proof of Theorem \ref{thm:rate_of_convg}}
\label{sec:rate_of_convgpf}

From \eqref{eq:def_eta_hat}, recall that
\begin{align*}
    \hat\eta_n = \dfrac{\frac{1}{n}\sum_{u=1}^{n}\frac{1}{K}\sum_{v \in N_{G\left(\sX_n\right)}(u)}\bY_u^\top \bY_v - \frac{1}{n(n-1)}\sum_{1 \leq u \neq v \leq n}\bY_u^\top \bY_v}{\frac{1}{n}\sum_{u=1}^{n}\left\|\bY_u\right\|^2 - \frac{1}{n(n-1)}\sum_{1 \leq u \neq v \leq n}\bY_u^\top \bY_v} = \frac{T_n}{D_n} , 
\end{align*} 
with $D_n$ and $T_n$ as defined in \eqref{eq:Dn} and \eqref{eq:Tn}, respectively. By the distributional convergence of $U$-statistics (see \cite[Theorem 12.3]{van2000asymptotic}), it immediately follows that 
\begin{align}\label{eq:rates_other_terms}
\begin{aligned}
    \left|\frac{1}{n(n-1)}\sum_{1 \leq u \neq v \leq n}\bY_u^\top \bY_v - \left\|\E[\bY]\right\|_2^2\right| & = O_P\left(\frac{1}{\sqrt{n}}\right) , \\ 
    \left|\frac{1}{n}\sum_{u=1}^{n}\left\|\bY_u\right\|_2^2 - \E\left[\left\|\bY\right\|_2^2\right]\right| 
    & = O_P\left(\frac{1}{\sqrt{n}}\right). 
    \end{aligned} 
\end{align}
These rates immediately imply that 
$$\left|D_n - \E\left[\left\|\bY - \E[\bY]\right\|_2^2\right]\right|= O_P\left(\frac{1}{\sqrt{n}}\right).$$
To complete the proof of the theorem, we aim to apply Lemma \ref{lemma:ratio_Op}. By \eqref{eq:rates_other_terms}, it therefore suffices to determine the rate at which 
$$V_n := \frac{1}{n}\sum_{u=1}^{n}\frac{1}{K}\sum_{v \in N_{G\left(\sX_n\right)}(u)}\bY_u^\top \bY_v$$ concentrates around
$\E[\|\E[\bY\mid\bX]\|_2^2]$, which we show in the following proposition.
\begin{proposition}\label{prop:rate_nn_est}
    Under the assumptions of Theorem \ref{thm:rate_of_convg},
    \begin{align*}
        \left|\frac{1}{n}\sum_{u=1}^{n}\frac{1}{K}\sum_{v \in N_{G\left(\sX_n\right)}(u)}\bY_u^\top \bY_v - \E\left[\left\|\E[\bY\mid\bX]\right\|_2^2\right]\right| = O_P\left(\max\left\{\frac{(\log n)^2}{\sqrt{n}},\frac{(\log n)^{1+\frac{1}{d}}}{n^{\frac{1}{d}}}\right\}\right)
    \end{align*}
\end{proposition}
Now combining \eqref{eq:rates_other_terms} and Proposition \ref{prop:rate_nn_est} we get the convergence rate,
\begin{align*}
    \left|T_n - \E[\|\E[\bY\mid\bX] - \E[\bY]\|_2^2]\right|= O_P\left(\max\left\{\frac{(\log n)^2}{\sqrt{n}},\frac{(\log n)^{1+\frac{1}{d}}}{n^{\frac{1}{d}}}\right\}\right).
\end{align*}
The proof is now completed by recalling that $\bY$ is not almost surely a constant and applying Lemma \ref{lemma:ratio_Op}.

\subsubsection{Proof of Proposition \ref{prop:rate_nn_est}}
Recalling \eqref{eq:Tn}, note that $T_n = V_n - \frac{1}{n(n-1)}\sum_{1 \leq u \neq v \leq n}\bY^\top\bY$, where,
\begin{align}\label{eq:Vn}
    V_n = \frac{1}{n K}\sum_{u=1}^{n}\sum_{v \in N_{G(\sX_n)}(u)}\bY^\top\bY.
\end{align}
First we control of variance of $V_n$. For this, applying \eqref{eq:bound_var_Vn} gives,
\begin{align*}
    \Var[V_n]\lesssim_d\frac{1}{n}\E\left[\max_{1\leq i\leq n}\|\bY_u\|_2^4\right]\lesssim_{d,p}\frac{1}{n} \inf_{\delta>0}\left\{\delta + n\int_{\delta}^{\infty}\P\left(\left\|\bY - \E\bY\right\|_2^4>t\right)\d t\right\} + \frac{1}{n} , 
\end{align*}
where the last inequality follows using union bound. Now, using Assumption \ref{assumption:rate_of_convg} (1) gives, 
\begin{align}\label{eq:bd_var_Vn}
    \Var[V_n]\lesssim_{d,p}\frac{1}{n}\inf_{\delta>0}\left\{\delta + n\int_{\delta}^{\infty}\exp\left(-C_2t^{1/4}\right)\d t\right\}\lesssim_{d,p}\frac{\left(\log n\right)^4}{n} . 
\end{align}
Next, we control the bias. For this, recall the definition of $g$ from Assumption \ref{assumption:rate_of_convg}, choose $\vep_n = (\log n)^{1 + 1/d}/n^{1/d}$, and 
\begin{align*}
    \nu_n = 
    \begin{cases}
            \frac{\left(\log n\right)^3}{n} & \text{ if }d = 1 , \\
            \frac{\left(\log n\right)^{4}}{n} & \text{ if }d = 2, \\
            \frac{\left(\log n\right)^{2 + 2/d}}{n^{2/d}} & \text{ if }d \geq 3.  
        \end{cases}
\end{align*}
Then by the argument from \eqref{eq:unifoneneigh}, 
\begin{align*}
    \left|\E[V_n] - \E\left[\left\|\E[\bY\mid\bX]\right\|_2^2\right]\right|
    & \lesssim \E\left[\left|g(\bX_1)^\top\left(g(\bX_{N(1)}) - g(\bX_1)\right)\right|\right]\\
    & \lesssim\E\left[\left(1 + \|\bX_1\|_2^{\beta} + \|\bX_{N(1)}\|_2^{\beta}\right)\left\|\bX_1 - \bX_{N(1)}\right\|\right] , 
\end{align*}
where $N(1)$ is an index selected uniformly from the neighbors of $\bX_1$ in $G(\sX_n)$ and the second inequality follows by the local Lipschitz assumption on $g$ in Assumption \ref{assumption:rate_of_convg}. 
Now, let $M_n = C \log n$, for a sufficiently large constant $C>0$, and define
\begin{align*}
\mathcal{E}_n := \left\{\max\{\|\bX_1\|_2,\|\bX_{N(1)}\|_2\}\leq M_n\right\}.
\end{align*}
Then consider the following decomposition: 
\begin{align}
\E\Big[\big(1+\|\bX_1\|_2^{\beta}+\|\bX_{N(1)}\|_2^{\beta}\big)
\|\bX_1-\bX_{N(1)}\|_2\Big]
\lesssim T_1 + T_2, \label{eq:decomp}
\end{align}
where
\begin{align*}
T_1 &:= \E\Big[\big(1+\|\bX_1\|_2^{\beta}+\|\bX_{N(1)}\|_2^{\beta}\big)
\|\bX_1-\bX_{N(1)}\|_2\one\{\mathcal{E}_n^c\}\Big],\\
T_2 &:= \E\Big[\big(1+\|\bX_1\|_2^{\beta}+\|\bX_{N(1)}\|_2^{\beta}\big)
\|\bX_1-\bX_{N(1)}\|_2\one\{\mathcal{E}_n\}\Big].
\end{align*}
For $T_1$, by Cauchy--Schwarz inequality, the tail condition in Assumption \ref{assumption:rate_of_convg} and Lemma \ref{lemma:fXN1_bdd_fX1} gives, 
\begin{align}\label{eq:I_1_bdd}
T_1 &\leq
\left(\E\Big[\big(1+\|\bX_1\|_2^{\beta}+\|\bX_{N(1)}\|_2^{\beta}\big)^2
\|\bX_1-\bX_{N(1)}\|_2^{2}\Big]\right)^{\frac{1}{2}}
\left(\P(\mathcal{E}_n^c)\right)^{\frac{1}{2}} \lesssim \frac{1}{n^2} ,
\end{align}
for $C$ chosen sufficiently large. Next,  for $T_2$, once again by Cauchy-Schwarz inequality,
\begin{align}\label{eq:I2_bdd}
    T_2
    &\lesssim  \left(\E\left[\left(1+\|\bX_1\|_2^{\beta}+\|\bX_{N(1)}\|_2^{\beta}\right)^2\right] \right)^{\frac{1}{2}} \left(\E\left[\|\bX_1 - \bX_{N(1)}\|_2^{2}\one\left\{\mathcal{E}_n\right\}\right] \right)^{\frac{1}{2}} \nonumber\\
    &\lesssim \left(\E\left[\|\bX_1 - \bX_{N(1)}\|_2^{2}\one\left\{\mathcal{E}_n\right\}\right]\right)^{\frac{1}{2}} , 
\end{align}
where the last inequality from the tail condition on $P_{\bX}$ in Assumption \ref{assumption:rate_of_convg}. 
%
%
Now, observe that
\begin{align}
\E\Big[\|\bX_1-\bX_{N(1)}\|_2^{2}\one\{\mathcal{E}_n\}\Big]
& \lesssim \int_{0}^{2M_n}\vep\P\left(\|\bX_1-\bX_{N(1)}\|_2\geq \vep,\, \mathcal{E}_n\right)\mathrm{d}\vep\nonumber\\
&\lesssim
\vep_n^{2}
+ \int_{\vep_n}^{2M_n}
\vep
\P\Big(\|\bX_1-\bX_{N(1)}\|_2\geq \vep,\, \mathcal{E}_n\Big)
\, \mathrm{d}\vep  ,  
\label{eq:tail_int}
\end{align} 
where the second inequality follows by noting that $\vep_n\leq M_n$ for large enough $C>0$. To control the second term, let $\cN = \cN(M_n,\vep)$ denote the covering number of the Euclidean ball $\cB(M_n) = \{\bx\in \R^d : \|\bx\|_2 \le M_n\}$ by $\|\cdot\|_2$-norm balls of diameter $\vep$.  Let $\{\cB_i\}_{1 \leq i \leq \cN}$ be a covering of $\cB(M_n)$ and define
\begin{align}\label{eq:def_S}
\cS = \left\{1 \leq i \leq \cN: P_{\bX}(\cB_i) \le \frac{C K \log n}{n} \right\} , 
\end{align}
where $P_{\bX}(A)$ denotes the measure of the set $A$ under $P_{\bX}$. Then, for $\vep \in (\vep_n, M_n)$, using the union bound and Lemma \ref{lemma:fXN1_bdd_fX1} gives, 
\begin{align}
    \P(\|\bX_1 - \bX_{N(1)}\|_2 \ge \vep, & \max(\|\bX_1\|_2,\|\bX_{N(1)}\|_2) \le M_n) \nonumber \\
&\lesssim \underbrace{\P\left(\bX_1,\bX_{N(1)} \in \bigcup_{i\notin \cS} \cB_i, \|\bX_1 - \bX_{N(1)}\|_2 \ge \vep \right)}_{\mathcal P_n} + \frac{K \log n}{n} \cN. 
\label{eq:Mn}
\end{align}
Note that if $\|\bX_1 - \bX_{N(1)}\|_2 \ge \vep$, then every other point outside the $K$ nearest neighbors must be at distance at least $\vep$ from $\bX_1$, that is, for any $v$ not among the $K$ nearest neighbors of $\bX_1$, we have $\|\bX_1 - \bX_v\|_2 \ge \vep$. Moreover, this implies there must exist $1\leq \ell\leq n$ and subset $J_\ell = \{j_1, j_2,\ldots, j_{n-K-1}\}\subseteq[n]\setminus\{\ell\}$ such that $\bX_{\ell}\in \cup_{i\not\in \cS}\cB_i$ and $\|\bX_\ell - \bX_{j_v}\|\geq \vep$ for all $1\leq v\leq n-K-1$. By a union bound,
\begin{align*}
    \mathcal P_n\leq\sum_{\ell=1}^{n}\sum_{\substack{J_\ell\subseteq[n]\setminus\{\ell\}\\|J_\ell| = n-K-1}}\P\left(\bX_{\ell}\in \bigcup_{i\not\in \cS}\cB_i, \min_{1\leq v\leq n-K-1}\|\bX_\ell - \bX_{j_v}\|\geq \vep\right).
\end{align*}
To further bound $\mathcal P_n$ let $\cB(\bX_{\ell}) = \cB_j$ such that $\bX_\ell\in \cB_j$ for some $j\not\in \cS$. Then,
\begin{align*}
    \mathcal P_n
    & \leq \sum_{\ell=1}^{n}\sum_{\substack{J_\ell\subseteq[n]\setminus\{\ell\}\\|J_\ell| = n-K-1}}\P\left(\bX_{j_v}\not\in \cB(\bX_\ell)\text{ for all }1\leq v\leq n-K-1\right)\\
    & = \sum_{\ell=1}^{n}\sum_{\substack{J_\ell\subseteq[n]\setminus\{\ell\}\\|J_\ell| = n-K-1}}\E\left[\left(1-\P\left(\bX\in\cB(\bX_\ell)\mid\bX_\ell\right)\right)^{n-K-1}\right]\le n^{K+1} \left(1 - \frac{C K \log n}{n}\right)^{n-K-1} \lesssim \frac{1}{n^2},
\end{align*}
where $\bX\sim P_{\bX}$ is generated independent of $\bX_\ell$. The equality follows from independence of $\bX_\ell,\bX_{j_1},\ldots,\bX_{j_{n-K-1}}$, while the penultimate inequality follows from the definition of $\cS$ in \eqref{eq:def_S}. Applying the above with the standard bound $\cN \lesssim_d (\log n)^{d}/\vep^d$ in \eqref{eq:Mn} gives, 
\[
\P(\|\bX_1 - \bX_{N(1)}\|_2 \ge \vep, \max(\|\bX_1\|_2,\|\bX_{N(1)}\|_2) \le M_n) \lesssim_d \frac{K (\log n)^{1+d}}{n \vep^d}.
\]
Plugging this into the tail integral in \eqref{eq:tail_int} gives, 
\begin{align}\label{eq:I2_rate}
    \E[\|\bX_1 - \bX_{N(1)}\|_2^{2} \mathbf{1}\{\mathcal{E}_n\}] 
\lesssim_d \vep_n^{2} + \frac{K (\log n)^{1+d}}{n} \int_{\vep_n}^{2M_n} \vep^{1 - d} \mathrm{d} \vep \lesssim_d \vep_n^{2} + \nu_n.
\end{align}
Combining \eqref{eq:decomp}, \eqref{eq:I_1_bdd}, \eqref{eq:I2_bdd}, and \eqref{eq:I2_rate} yields,
\begin{align}\label{eq:EVn_conc_rate}
\left|\E[V_n] - \E\left[\left\|\E[\bY\mid\bX]\right\|_2^2\right]\right| \lesssim_d \frac{1}{n^2} + \vep_n^{} + \sqrt{\nu_n}.
\end{align}
Finally, recalling the bound on $\Var[V_n]$ from \eqref{eq:bd_var_Vn} we conclude,
\begin{align*}
    \left|V_n - \E\left[\left\|\E[\bY\mid\bX]\right\|_2^2\right]\right| = O_P\left(\sqrt{\frac{(\log n)^4}{n} + \vep_n^2 + \nu_n}\right) = O_P\left(\max\left\{\frac{(\log n)^2}{\sqrt{n}},\frac{(\log n)^{1+\frac{1}{d}}}{n^{\frac{1}{d}}}\right\}\right) .  
\end{align*}
This completes the proof of Proposition \ref{prop:rate_nn_est}. \hfill $\Box$



\section{Proof of Theorem \ref{thm:null_dist}}
\label{sec:cltpf}

The first step in the proof is to establish the decomposition in \eqref{eq:RnTn}. This is stated in the following lemma, which we prove later in Appendix \ref{sec:Rnpf}.

\begin{lemma} Let $T_n$ and $R_n$ be as defined in \eqref{eq:Tn} and \eqref{eq:Rn}, respectively. Then the following hold: 
$$\sqrt n T_n = R_n + o_{L_2}(1).$$
\label{lm:Rn}
\end{lemma}

Next, we compute the variance of $R_n$ under $H_0$. The proof is given in Appendix \ref{sec:varianceRnpf}. 

\begin{lemma} Let $R_n$ be as defined in \eqref{eq:Rn} and $\sigma_n^2 := \Var_{H_0}\left[R_n\middle| \mathcal{F}(\sX_n)\right] $. Then 
\begin{align*}
       \sigma_n^2 = & \frac{1}{n K^2} \sum_{1 \leq u \ne v \leq n} (\E_{H_0}[(\bZ_u^\top \bZ_v)^2\mid  \bX_u, \bX_v]) \left( \one\left\{\bX_u \rightarrow \bX_v\right\} + \one\left\{ \bX_u \leftrightarrow \bX_v \right\} \right) \nonumber \\ 
& \hspace{1.35in} + \frac{1}{n}\sum_{u=1}^{n}\E_{\bm H_0}\left[\bY\right]^\top\Var_{\bm H_0}\left(\bY_u\middle|\bX_u\right)\E_{\bm H_0}\left[\bY\right]\left(\frac{\bar{d}_u}{K}-1\right)^2 , 
    \end{align*} 
    \label{lm:varianceRn} 
    with $\bZ_u = \bY_u - \E[\bY]$ and $\bar{d}_u: |\left\{v \in [n]:\bX_v\ra\bX_u \right\}|$, for $1\leq u \leq n$. 
\end{lemma}

The next result establishes the asymptotic normality of $\frac{R_n}{\sigma_n}$ under $H_0$. The proof is given in Appendix \ref{sec:cltH0sigmapf}.

\begin{proposition}\label{prop:clt_with_sigman} 
    Suppose Assumption \ref{assumption:normX_continuous} holds and $\E[\|\bY\|_2^{8+\delta}]<\infty$. Then under $\bm H_0$, as $n\ra\infty$,
    \begin{align*}
        \sup_{z\in\R}\left|\P\left(\frac{R_n}{\sigma_n}\leq z\right) - \Phi(z)\right| \ra 0 . 
    \end{align*} 
    \end{proposition}

Next, we show that $\hat\sigma_n^2$ (defined in \eqref{eq:varianceH0}) is a consistent estimate of $\sigma_n^2$. The proof is given in Appendix \ref{sec:varianceestimatepf}.

\begin{proposition}\label{ppn:varianceestimate}
    Assume that $\E[\|\bY\|_2^{8+\delta}]<\infty$ for some $\delta>0$. Then under $H_0$,
    \begin{align*}
        \left|\frac{\hat{\sigma}_n^2}{\sigma_n^2}-1\right| = o_P\left(n^{-\frac{\delta}{32 + 4\delta}}\right),
    \end{align*}
    where $\sigma_n^2$ and $\hat\sigma_n^2$ are defined in Lemma \ref{lm:varianceRn} and \eqref{eq:varianceH0}, respectively.
\end{proposition}

With the above results the proof of Theorem \ref{thm:null_dist} can be completed as follows. First, note that the lower bound in \eqref{eq:var_tn_xn_lb}, together with the convergences in \eqref{eq:var_sigman_t1_convg} and \eqref{eq:Var_t1_convg}, implies that $\sigma_n$ satisfies the lower-bound assumption of Lemma \ref{lemma:replace_lemma}. Then Lemma \ref{lm:Rn} and Proposition \ref{prop:clt_with_sigman}, in conjunction with Lemma \ref{lemma:replace_lemma}, gives 
\begin{align}\label{eq:tn_sigman_clt}
    \sup_{z\in \R}\left|\P\left(\frac{\sqrt{n}T_n}{\sigma_n}\leq z\right) - \Phi(z)\right|\ra 0\text{ as }n\ra \infty.
\end{align}
The proof is completed by replacing $\sigma_n$ with $\hat \sigma_n$ in \eqref{eq:tn_sigman_clt} using Proposition \ref{ppn:varianceestimate}.


\subsection{Proof of Lemma \ref{lm:Rn}}
\label{sec:Rnpf}
First note that,
\begin{align*}
    \sqrt{n}T_n - R_n = \frac{1}{\sqrt{n}(n-1)}\sum_{1 \leq u\neq v \leq n}\bY_u^\top\bY_v - \frac{1}{\sqrt{n}}\sum_{u=1}^{n}\E[\bY]^\top\left(2\bY_i - \E[\bY]\right).
\end{align*}
For notational simplification let $\bZ_i = \bY_i - \E[\bY]$ for all $1\leq i\leq n$. Then by a direct computation,
\begin{align*}
    \frac{1}{\sqrt{n}(n-1)}\sum_{1 \leq u\neq v \leq n}\bY_u^\top\bY_v = \frac{1}{\sqrt{n}(n-1)}\sum_{1 \leq u\neq v \leq n}\bZ_u^\top \bZ_v + \frac{2}{\sqrt{n}}\sum_{u=1}^{n}\E[\bY]^\top\bZ_u + \sqrt{n}\|\E[\bY]\|_2^2,
\end{align*}
and similarly,
\begin{align*}
    \frac{1}{\sqrt{n}}\sum_{u=1}^{n}\E[\bY]^\top\left(2\bY_i - \E[\bY]\right) = \frac{2}{\sqrt{n}}\sum_{u=1}^{n}\E[\bY]^\top\bZ_u + \sqrt{n}\|\E[\bY]\|_2^2.
\end{align*}
With the expansions from above the difference $\sqrt{n}T_n - R_n$ becomes,
\begin{align*}
    \sqrt{n}T_n - R_n = \frac{1}{\sqrt{n}(n-1)}\sum_{1 \leq u\neq v \leq n}\bZ_u^\top \bZ_v.
\end{align*}
Using independence and $\E[\bZ_i] = 0$ for all $1\leq i\leq n$, it follows immediately that,
\begin{align*}
    \E\left[\left(\sqrt{n}T_n - R_n\right)^2\right] = \frac{2}{n-1}\E\left[\left(\bZ_1^\top\bZ_2\right)^2\right] = o(1),
\end{align*}
which completes the proof.  \hfill $\Box$  

\subsection{Proof of Lemma \ref{lm:varianceRn}}
\label{sec:varianceRnpf}

Recall the definition of $R_n$ from \eqref{eq:Rn}. Then a direct computation gives, 
\begin{align*}
   R_n & = \frac{1}{\sqrt{n} K }\sum_{u=1}^{n} \sum_{v\in N_{G(\sX_n)}(u)}\left(\bY_u - \E\left[\bY\right]\right)^\top\left(\bY_v - \E\left[\bY\right]\right) \nonumber \\ 
   & \hspace{1.5in} + \frac{1}{\sqrt{n}}\sum_{u=1}^{n}\left(\frac{\bar{d}_u}{K}-1\right)\E\left[\bY\right]^\top\left(\bY_u - \E\left[\bY\right]\right) , 
\end{align*}
    We now decompose the conditional variance as follows:
    \begin{align}\label{eq:condvardecomp}
        \Var_{H_0}\left[R_n\middle| \mathcal{F}(\sX_n)\right] 
        = S_1 + S_2 + S_3,
    \end{align}
    where
    \begin{align*}
    S_1 
    & := \Var_{H_0}\left[\frac{1}{\sqrt{n}K}\sum_{u=1}^{n}\sum_{v\in N_{G(\sX_n)}(u)}\bZ_u^\top \bZ_v\middle| \mathcal{F}(\sX_n)\right] ,\\
    S_2 
    & := \Var_{H_0}\left[\frac{1}{\sqrt{n}}\sum_{u=1}^{n}\left(\frac{\bar{d}_u}{K}-1\right)\E\left[\bY\right]^\top\left(\bY_u - \E\left[\bY\right]\right)\middle| \mathcal{F}(\sX_n)\right] ,\\
    S_3 
    & := 2\E_{H_0}\left[\frac{1}{nK}\sum_{w=1}^{n}\sum_{u=1}^{n}\sum_{v\in N_{G(\sX_n)}(u)}\bZ_u^\top \bZ_v\left(\frac{\bar{d}_w}{K}-1\right)\E\left[\bY\right]^\top  \bZ_w\middle| \mathcal{F}(\sX_n)\right],
    \end{align*}
    with $\bZ_u = \bY_u - \E\left[\bY\right]$ for all $1\leq u \leq n$. A direct computation then shows that
    \begin{align}\label{eq:S2value}
        S_2 = \frac{1}{n}\sum_{v=1}^{n}\left(\frac{\bar{d}_v}{K}-1\right)^2 \E\left[\bY\right]^\top \Var\left[\bY_v\middle|\bX_v\right] \E\left[\bY\right].
    \end{align}
Next, recall that
    \begin{align*}
        S_3 
        = \frac{2}{nK}\sum_{w=1}^{n}\left(\frac{\bar{d}_w}{K}-1\right)\sum_{u=1}^{n}\sum_{v\in N_{G(\sX_n)}(u)}\E\left[\bZ_u^\top \bZ_v \bZ_w^\top \E[\bY]\middle| \mathcal{F}(\sX_n)\right].
    \end{align*}
    To simplify $S_3$, first note that by construction $1 \leq u\neq v \leq n$. Then using $\E_{H_0}\left[\bZ_u \middle| \mathcal{F}(\sX_n)\right]=0$, for all $1\leq u \leq n$, it follows that $S_3 = 0$. Finally, to simplify $S_1$ we consider the following decomposition,
    \begin{align*}
        S_1 = S_{11} + S_{12} + S_{13},
    \end{align*}
    where
    \begin{align}
    \begin{aligned} 
            \label{eq:S1expression} 
        S_{11}  
        & = \frac{1}{nK^2}\sum_{u=1}^{n}\sum_{v=1}^{n}\E_{H_0}\left[(\bZ_u^\top \bZ_v)^2\middle|\bX_u, \bX_v\right]\left( \one\left\{\bX_u \rightarrow \bX_v\right\} + \one\left\{ \bX_u \leftrightarrow \bX_v \right\} \right) , \\
        S_{12} 
        & = \frac{1}{nK^2}\sum_{\cS_2}\E_{H_0}\left[\bZ_u^\top \bZ_v \bZ_{u^\prime}^\top \bZ_{v^\prime}\middle|\bX_u, \bX_v, \bX_{u^\prime}, \bX_{v^\prime}\right]  \one\left\{\bX_u \rightarrow \bX_v\right\} \one\left\{\bX_{u'} \rightarrow \bX_{v'}\right\} , \\
        S_{13}
        & = \frac{1}{nK^2}\sum_{\cS_3}\E_{H_0}\left[\bZ_u^\top \bZ_v \bZ_{u^\prime}^\top \bZ_{v^\prime}\middle|\bX_u, \bX_v, \bX_{u^\prime}, \bX_{v^\prime}\right]  \one\left\{\bX_u \rightarrow \bX_v\right\} \one\left\{\bX_{u'} \rightarrow \bX_{v'}\right\} , 
          \end{aligned} 
    \end{align}
    with $\cS_2 = \{(u, v, u^\prime, v^\prime)\in [n]^4: |\{u, v\}\cap \{u^\prime, v^\prime\}|=1\}$ and $\cS_3 = \{(u, v, u^\prime, v^\prime)\in [n]^4: |\{u, v\}\cap \{u^\prime, v^\prime\}|=0\}$. Since $\E_{H_0}\left[\bZ_u \middle|\bX_u\right] = 0$, for all $1\leq u \leq n$, we have that $S_{13} = 0$. Furthermore, consider $(u, v, u^\prime, v^\prime)\in \cS_2$ and without loss of generality suppose $u = u^\prime$. Then, 
    \begin{align*}
        \E_{H_0}[\bZ_u^\top \bZ_v \bZ_{u^\prime}^\top \bZ_{v^\prime}|\bX_u, \bX_{u^\prime}, \bX_v, \bX_{v^\prime}]
        = \E_{H_0}[\bZ_v^\top \bZ_u \bZ_u^\top \bZ_{v^\prime}|\bX_u, \bX_v, \bX_{v^\prime}] = 0,
    \end{align*}
    implying $S_{12} = 0$. Collecting the above shows, $S_1 = S_{11}$. Combining this \eqref{eq:condvardecomp}, \eqref{eq:S2value}, and \eqref{eq:S1expression}, completes the proof of Lemma \ref{lm:varianceRn}.  \hfill $\Box$

\subsection{Proof of Proposition \ref{prop:clt_with_sigman}}
\label{sec:cltH0sigmapf}

Define, for each $1 \le u \le n$,
\begin{align*}
    V_u 
    = \frac{1}{\sqrt{n}}
    \left[
        \frac{1}{K}\sum_{v \in N_{G(\sX_n)}(u)} \bY_u^\top \bY_v
        - \E[\bY]^\top\bigl(2\bY_u - \E[\bY]\bigr)
    \right].
\end{align*}
Then $R_n$ (recall \eqref{eq:Rn}) can be written as
\begin{align*}
    R_n = \sum_{u=1}^{n} V_u .
\end{align*}
We will prove the result in Proposition \ref{prop:clt_with_sigman} using Stein's method for dependency graphs \cite{chen2004normal}. For this, we need to construct a dependency graph for the collection of random variables $\{V_1,\dots,V_n\}$. Denote by $\bar{G}(\sX_n)$ the undirected simple graph obtained from the the $K$-NN graph
$G(\sX_n)$, that is, we remove the directions from the edges and if for a pair of vertices there are directed edges in both directions, we keep only an undirected edge between them. We then define a dependency graph $\cG_n = (\cV_n, \cE_n)$ with vertex set
$\cV_n = \{V_1,\dots,V_n\}$, where an edge is placed between $V_u$ and $V_v$
if and only if there exists a path of length at most two between $u$ and $v$
in the undirected version of $\bar{G}(\sX_n)$. 
Let $D$ denote the maximum degree of $\cG_n$.
By \cite[Lemma 1]{jaffe2020randomized}, the maximum degree of $\bar{G}(\sX_n)$ is bounded by $c_{d} K$ for some constant $c_{d}>0$ depending only on the dimension $d$. Hence, $D \lesssim_d K^2$. Moreover, under $\bm H_0$,
\begin{align*}
    \E_{\bm H_0}\!\left[V_u \middle| \mathcal{F}(\sX_n) \right] = 0 .
\end{align*}
Then, applying the version Stein's method from \cite[Theorem 2.7]{chen2004normal} gives, 
\begin{align*}
    \sup_{z \in \R}
    \left|
        \P\left(\frac{R_n}{\sigma_n} \le z \middle| \mathcal{F}(\sX_n) \right) - \Phi(z)
    \right|
    \lesssim_{d}\frac{K^{20}}{\sigma_n^3}
        \E_{\bm H_0}\!\left[\sum_{u=1}^{n} |V_u|^3 \middle| \mathcal{F}(\sX_n) \right]
\end{align*}
where $\sigma_n^2 = \Var_{\bm H_0}(R_n \mid \sX_n)$. Using the tower property of conditional expectation, for any $\vep>0$,
\begin{align}
    \sup_{z \in \R}
    \left|
        \P\left(\frac{R_n}{\sigma_n} \le z \right) - \Phi(z)
    \right|
    &\lesssim_{d}
    \E\!\left[
            \frac{K^{20}}{\sigma_n^3}
            \E_{\bm H_0}\!\left[\sum_{u=1}^{n} |V_u|^3 \middle| \mathcal{F}(\sX_n) \right]
    \right] \nonumber\\
    &\le
    \E\!\left[\frac{K^{20}}{\sigma_n^3}
            \E_{\bm H_0}\!\left[\sum_{u=1}^{n} |V_u|^3 \middle| \mathcal{F}(\sX_n) \right]
        \one\{n^{\vep}\sigma_n^2 \ge 1\}
    \right]
    + o(1)\nonumber\\
    &\le
    K^{20} n^{\frac{3\vep-1}{2}}
    \E\!\left[
        \sqrt{n}\,
        \E_{\bm H_0}\!\left[\sum_{u=1}^{n} |V_u|^3 \middle| \mathcal{F}(\sX_n) \right]
    \right]
    + o(1), \label{eq:use_lemma_var_lim}
\end{align}
where the second inequality follows from Lemma~\ref{lemma:varTnlimpos}. To control the RHS of \eqref{eq:use_lemma_var_lim}, by the definition of $V_u$ and H\"{o}lder's inequality,
\begin{align}
    \sqrt{n}\sum_{u=1}^{n} |V_u|^3
    &\le
    \frac{1}{n}\sum_{u=1}^{n}
    \left|
        \E[\bY]^\top(\bY_u - \E[\bY])
    \right|^3
    + \frac{1}{n}\sum_{u=1}^{n}
    \left|
        \frac{1}{K}\sum_{v \in N_{G(\sX_n)}(u)} \bY_u^\top \bY_v
    \right|^3 \nonumber\\
    &\le
    \frac{1}{n}\sum_{u=1}^{n}
    \left|
        \E[\bY]^\top(\bY_u - \E[\bY])
    \right|^3
    + \frac{1}{n K}
    \sum_{u=1}^{n}\sum_{v \in N_{G(\sX_n)}(u)}
    \|\bY_u\|_2^3 \|\bY_v\|_2^3 .\label{eq:V_u_bdd}
\end{align}
Under the moment assumption in Theorem \ref{prop:clt_with_sigman}, a straightforward decomposition the first term shows, 
\begin{align*}
    \E\!\left[
        \E_{\bm H_0}\!\left[
            \frac{1}{n}\sum_{u=1}^{n}
            \left|
                \E[\bY]^\top(\bY_u - \E[\bY])
            \right|^3
            \middle| \mathcal{F}(\sX_n)
        \right]
    \right]
    = O(1).
\end{align*}
For the second term from \eqref{eq:V_u_bdd}, using exchangeability and letting $N(1)$ denote a uniformly chosen neighbor of $1$
in $G(\sX_n)$,
\begin{align*}
    \E\!\left[
        \E_{\bm H_0}\!\left[
            \frac{1}{n K}\sum_{u=1}^{n}\sum_{v \in N_{G(\sX_n)}(u)}
            \|\bY_u\|_2^3 \|\bY_v\|_2^3
            \middle| \mathcal{F}(\sX_n)
        \right]
    \right]
    = \E\!\left[\beta(\bX_1)\beta(\bX_{N(1)})\right],
\end{align*}
where $\beta(\bm x) = \E[\|\bY\|_2^3 \mid \bX=\bm x]$. Now, by the Cauchy-Schwarz inequality and Lemma \ref{lemma:fXN1_bdd_fX1},
\begin{align*}
    \E\!\left[\beta(\bX_1)\beta(\bX_{N(1)})\right]
    \le \E\!\left[\beta(\bX_1)^2\right]
    \le \E\!\left[\|\bY_1\|_2^6\right]
    < \infty .
\end{align*}
Combining the above bounds gives, 
\begin{align*}
    \E\!\left[
        \sqrt{n}\,
        \E_{\bm H_0}\!\left[\sum_{u=1}^{n} |V_u|^3 \middle| \mathcal{F}(\sX_n) \right]
    \right]
    = O(1).
\end{align*}
Therefore, from \eqref{eq:use_lemma_var_lim},
\begin{align*}
    \sup_{z \in \R}
    \left|
        \P\left(\frac{R_n}{\sigma_n} \le z \right) - \Phi(z)
    \right|
    \lesssim_{d}
    K^{20} n^{\frac{3\vep-1}{2}} .
\end{align*}
The proof of Proposition \ref{prop:clt_with_sigman} is now completed by choosing $\vep< \frac{1}{3}$. \hfill $\Box$

\subsection{Proof of Proposition \ref{ppn:varianceestimate}} 
\label{sec:varianceestimatepf}

To prove Proposition \ref{ppn:varianceestimate} we divide the proof into two steps. In the first step we show that the plug in estimate $\hat\sigma_n^2$ is close to the conditional variance $\sigma_n^2$ in $L_2$.

\begin{lemma}\label{lemma:L2_var_est}
    Suppose Assumption \ref{assumption:normX_continuous} holds and let $\E[\|\bY\|_2^{8+\delta}]<\infty$. Then under $H_0$,
    \begin{align*}
        \left|\sigma_{n}^2-\hat\sigma_n^2\right| = o_P\left(n^{-\frac{\delta}{16+2\delta}}\right).
    \end{align*}
\end{lemma}
The proof of Lemma \ref{lemma:L2_var_est} is postponed to Appendix \ref{appendix:proofof_lemma_L2_var_est}. Lemma \ref{lemma:L2_var_est} shows that $\hat\sigma_n^2$ is indeed an approximation of $\sigma_n^2$ in $L_2$. However to complete the proof of Proposition \ref{ppn:varianceestimate} we have to translate this to an approximation error in terms of the ratio $\hat\sigma_n^2/\sigma_n^2$. To that end in the following lemma (with proof postponed to Appendix \ref{appendix:proofof_lemma_varTnlimpos}) we show that $\sigma_n$ is asymptotically bounded away from $0$.

\begin{lemma}\label{lemma:varTnlimpos}
Suppose $\E\left[\|\bY\|_2^{8+\delta}\right]<\infty$ for some $\delta>0$ and assume that $\bY$ is almost surely not a constant. Then for all $\vep, t>0$,
\begin{align*}
\P\left(n^{\vep}\Var_{\bm{H}_0}\left(R_n\middle| \mathcal{F}(\sX_n)\right)>t\right) \to 1.
\end{align*}
\end{lemma}
Now to complete the proof of Proposition \ref{ppn:varianceestimate} consider $\gamma = \frac{\delta}{32 + 4\delta}$. Then by Lemma \ref{lemma:varTnlimpos} we know that $$\P\left(n^{\gamma}\Var_{\bm{H}_0}\left[R_n\middle| \mathcal{F}(\sX_n)\right] \leq 1\right) = o(1).$$ Hence, for all $\vep>0$,
    \begin{align*}
        \P\left(n^{\gamma}\left|\frac{\hat\sigma_n^2}{\sigma_n^2}-1\right| > \vep\right) 
        & \leq \P\left(n^{2\gamma}\left|\hat{\sigma}_n^2 - \sigma_n^2\right| > \vep\right) + o(1) = o(1).
    \end{align*}

\subsubsection{Proof of Lemma \ref{lemma:L2_var_est}}\label{appendix:proofof_lemma_L2_var_est}
To prove that $\hat\sigma_n^2$ is close to $\sigma_n^2$, we begin by first expanding both $\sigma_n^2$ and $\hat\sigma_n^2$ into five components and then show that the empirical counterpart of each component consistently estimates its population analogue. Combining these bounds yields Lemma \ref{lemma:L2_var_est}. We begin by decomposing $\sigma_n^2$ as follows: 
\begin{align}\label{eq:sigman_decomp}
    \sigma_n^2 =  Q^{(1)}_n  -  Q^{(2)}_n  +  Q^{(3)}_n  +  Q^{(4)}_n  -  Q^{(5)}_n , 
\end{align}
where  
\begin{align*}
     Q^{(1)}_n  & := \frac{1}{n K^2} \sum_{i,j=1}^n \E\Big[\big(\bY_u^\top \bY_v\big)^2 \,\big|\, \bX_u, \bX_v \Big] \big[\one\left\{\bX_u \rightarrow \bX_v\right\} + \one\{\bX_u\leftrightarrow\bX_j\}\big],\\
     Q^{(2)}_n  &:= \frac{1}{n K^2} \sum_{i,j=1}^n \E[\bY]^\top \Big(\E[\bY_u \bY_u^\top \,|\, \bX_u] + \E[\bY_v \bY_v^\top \,|\, \bX_v]\Big) \E[\bY] \big[\one\left\{\bX_u \rightarrow \bX_v\right\} + \one\{\bX_u\leftrightarrow\bX_j\}\big],\\
     Q^{(3)}_n &:=\frac{\|\E[\bY]\|_2^4}{n K^2} \sum_{i,j=1}^n \big[\one\left\{\bX_u \rightarrow \bX_v\right\} + \one\{\bX_u\leftrightarrow\bX_j\}\big], \\   
     Q^{(4)}_n &:=\frac{1}{n} \sum_{i=1}^n \Big(\frac{\bar{d}_i}{k}-1\Big)^2 \E[\bY]^\top \E[\bY_u \bY_u^\top \,|\, \bX_u] \E[\bY] , \\ 
     Q^{(5)}_n  & := \frac{\|\E[\bY]\|_2^4}{n} \sum_{i=1}^n \Big(\frac{\bar{d}_i}{k}-1\Big)^2.
\end{align*}
Next, we expand $\hat\sigma_n^2$ in a similar manner. In particular,
\begin{align}\label{eq:sigma_n_hat_decomp}
    \hat\sigma_{n}^2 := \hat Q^{(1)}_n  - \hat Q^{(2)}_n  + \hat Q^{(3)}_n  + \hat Q^{(4)}_n  - \hat Q^{(5)}_n  + \tilde Q_n,  
\end{align}
where  
\begin{align*}
\hat Q^{(1)}_n &:= \frac{1}{n K^2} \sum_{i,j=1}^n (\bY_u^\top \bY_v)^2 \, [\one\left\{\bX_u \rightarrow \bX_v\right\} + \one\{\bX_u\leftrightarrow\bX_j\}],\\
\hat Q^{(2)}_n &:= \frac{1}{n K^2} \sum_{i,j=1}^n \bar{\bY}^\top (\bY_u \bY_u^\top + \bY_v \bY_v^\top) \bar{\bY} \, [\one\left\{\bX_u \rightarrow \bX_v\right\} + \one\{\bX_u\leftrightarrow\bX_j\}] , \\
\hat Q^{(3)}_n &:= \frac{\|\bar{\bY}\|_2^4}{n K^2} \sum_{i,j=1}^n [\one\left\{\bX_u \rightarrow \bX_v\right\} + \one\{\bX_u\leftrightarrow\bX_j\}], \\ 
\hat Q^{(4)}_n &:= \frac{1}{n} \sum_{i=1}^n \Big(\frac{\bar{d}_i}{k} - 1\Big)^2 \bar{\bY}^\top \bY_u \bY_u^\top \bar{\bY}, \\ 
\hat Q^{(5)}_n &:= \frac{\|\bar{\bY}\|_2^4}{n} \sum_{i=1}^n \Big(\frac{\bar{d}_i}{k} - 1\Big)^2,
\end{align*}
and $\tilde Q_n = \hat\sigma_n^2 - \left(\hat Q^{(1)}_n  - \hat Q^{(2)}_n  + \hat Q^{(3)}_n  + \hat Q^{(4)}_n  - \hat Q^{(5)}_n\right)$. In the following, we first show that $\hat Q^{(i)}_n$ is close to $Q^{(i)}_n$, for $i \in \{1, 4, 5\}$.

\begin{lemma}\label{lemma:tau1}
    Suppose Assumption \ref{assumption:normX_continuous} holds and $\E[\|\bY\|_2^{8+\delta}]<\infty$ for some $\delta>0$. Then under $H_0$,
    \begin{align*}
        \max\left\{\E_{H_0}\left[\left| Q^{(1)}_n  - \hat Q^{(1)}_n \right|\right], \E_{H_0}\left[\left| Q^{(4)}_n  - \hat Q^{(4)}_n \right|\right], \E_{H_0}\left[\left| Q^{(5)}_n  - \hat Q^{(5)}_n \right|\right]\right\} = o\left(n^{-\frac{\delta}{16+2\delta}}\right)
    \end{align*}
\end{lemma}

\begin{proof}[Proof of Lemma \ref{lemma:tau1}]
We begin by showing that
\begin{align}\label{eq:var_term_1_convg}
    \E_{H_0}\left[\left| Q^{(1)}_n  - \hat Q^{(1)}_n \right|\right] = o\left(n^{-\frac{\delta}{16+2\delta}}\right)  .  
\end{align}
To that end notice that it is enough to prove the error bound,
    \begin{align}\label{eq:var_term-1}
        \E\left[\left(\frac{1}{nK^2}\sum_{1 \leq u \ne v \leq n}\left(\E\left[\left(\bY_u^\top \bY_v\right)^2\middle|\bX_u, \bX_v\right] - \left(\bY_u^\top \bY_v\right)^2\right)\one\left\{\bX_u \rightarrow \bX_v\right\}\right)^2\right] = o\left(n^{-\frac{\delta}{8+\delta}}\right)  .
    \end{align}
    Note that proof for $\one\left\{\bX_u \rightarrow \bX_v\right\}$ replaced by $\one\{\bX_u\leftrightarrow\bX_v\}$ is exactly similar and, hence, is omitted. Towards proving \eqref{eq:var_term-1} define,
    \begin{align*}
        \tilde Q^{(1)}_n  
        & = \frac{1}{nK^2}\sum_{1 \leq u \ne v \leq n}\E\left[(\bY_u^\top \bY_v)^2\middle|\bX_u, \bX_v\right]\one\left\{\bX_u \rightarrow \bX_v\right\},
    \end{align*}
    and
    \begin{align*}
        \hat{\tilde Q}_{n}^{(1)} 
        & = \frac{1}{nK^2}\sum_{1 \leq u \ne v \leq n}(\bY_u^\top \bY_v)^2\one\left\{\bX_u \rightarrow \bX_v\right\}.
    \end{align*}
    By definition it is easy to note that $\E[\tilde Q^{(1)}_n ]=\E[\hat{\tilde Q}_{n}^{(1)}]$. Hence, to show \eqref{eq:var_term-1} it is enough to show $\Var[\tilde Q^{(1)}_n ] = o(n^{-\frac{\delta}{8+\delta}})\text{ and }\Var[\hat{\tilde Q}_{n}^{(1)} ] = o(n^{-\frac{\delta}{8+\delta}})$. \textcolor{black}{Next, applying Lemma \ref{lemma:VarTn_bd} with $f(\bY_u, \bY_v) = (\bY_u^\top\bY_v)^2$ gives,} 
        \begin{align*}
        \Var[\hat{\tilde Q}_{n}^{(1)}]\lesssim_{d}\frac{1}{n}\E\left[\max_{1\leq u\neq v\leq n}\left|\bY_u^\top \bY_v\right|^4 \right]\leq \frac{1}{n}\E\left[\max_{1\leq u\leq n}\left\|\bY_u\right\|_2^8 \right].
    \end{align*}
    Now the arguments from \eqref{eq:max_Y_4_bd} can be easily adapted to show $ \Var[\hat{\tilde Q}_{n}^{(1)}] = o(n^{-\frac{\delta}{8+\delta}})$. 
    \textcolor{black}{Combinatorial arguments similar to Lemma \ref{lemma:VarTn_bd} along with the Efron-Stein inequality \cite{efron1981jackknife}} show that $\Var[\tilde Q^{(1)}_n] = o(n^{-\frac{\delta}{8+\delta}})$, completing the proof of \eqref{eq:var_term_1_convg}. 
    
Next, we show that,
\begin{align}\label{eq:var_term_4_convg}
    \E_{H_0}\left[\left|Q_n^{(4)} - \hat Q_n^{(4)}\right|\right] = O\left(1/\sqrt{n}\right).
\end{align}
First, we apply triangle inequality to get, 
\begin{align}\label{eq:term2_traingle_inequality}
    \E_{H_0}\left[\left|Q_n^{(4)} - \hat Q_n^{(4)}\right|\right]
    & \lesssim \E\Bigg[\Bigg|\frac{1}{n}\sum_{u=1}^{n} \Big(\frac{\bar{d}_u}{K}-1\Big)^2 
    \E[\bY]^\top \Big(\E[\bY_u\bY_u^\top \mid \bX_u] - \bY_u\bY_u^\top \Big) \E[\bY] \Bigg|\Bigg]\nonumber\\
    & + \E\Bigg[\Bigg|\frac{1}{n}\sum_{u=1}^{n} \Big(\frac{\bar{d}_u}{K}-1\Big)^2 
    \Big(\E[\bY]^\top \bY_u\bY_u^\top \E[\bY] - \bar{\bY}^\top \bY_u\bY_u^\top \bar{\bY}\Big) \Bigg|\Bigg] . 
\end{align}
We begin by showing the following,
\begin{align}\label{eq:tau4_1}
    \E\Bigg[\Bigg|\frac{1}{n}\sum_{u=1}^{n} \Big(\frac{\bar{d}_u}{K}-1\Big)^2 
    \E[\bY]^\top \Big(\E[\bY_u\bY_u^\top \mid \bX_u] - \bY_u\bY_u^\top \Big) \E[\bY] \Bigg|\Bigg] = O\left(\frac{1}{\sqrt{n}} \right)  , 
\end{align}
which shows that one can replace the conditional expectation $\E\left[\bY_u\bY_u^\top\mid\bX_u\right]$ in $ Q^{(4)}_n $ with  $\bY_u\bY_u^\top$ upto negligible error. Towards that, define $a_u = \E[\bY]^\top \bY_u$ for all $1 \le u \le n$. Then,
\begin{align*}
    &\E\Bigg[\Bigg|\frac{1}{n}\sum_{u=1}^{n} \Big(\frac{\bar{d}_u}{K}-1\Big)^2 
    \E[\bY]^\top \Big(\E[\bY_u\bY_u^\top \mid \bX_u] - \bY_u\bY_u^\top \Big) \E[\bY] \Bigg|^2\Bigg] \\
    & = \E\Bigg[\Bigg(\frac{1}{n}\sum_{u=1}^{n} \Big(\frac{\bar{d}_u}{K}-1\Big)^2 
    \big(a_u^2 - \E[a_u^2 \mid \bX_u]\big) \Bigg)^2\Bigg] 
    = \frac{1}{n^2} \sum_{u=1}^{n} \E\Big[\Big(\frac{\bar{d}_u}{K}-1\Big)^4 
    \big(a_u^2 - \E[a_u^2 \mid \bX_u]\big)^2\Big]  ,  
\end{align*}
where the last equality follows from the independence of $(a_u, \bX_u), 1 \le i \le n$. From \cite{jaffe2020randomized} we know that $\bar{d}_u\lesssim_d K$ for all $1\leq u\leq n$. Hence, recalling the moment assumptions we get,
\begin{align}\label{eq:tau4_1_sq_bd}
    \E\Bigg[\Bigg|\frac{1}{n}\sum_{u=1}^{n} \Big(\frac{\bar{d}_u}{K}-1\Big)^2 
    \E[\bY]^\top \Big(\E[\bY_u\bY_u^\top \mid \bX_u] - \bY_u\bY_u^\top \Big) \E[\bY] \Bigg|^2\Bigg] = O\left(\frac{1}{n} \right).
\end{align}
Next, we show
\begin{align}\label{eq:tau4_2}
    \E\Bigg[\Bigg|\frac{1}{n}\sum_{u=1}^{n} \Big(\frac{\bar{d}_u}{K}-1\Big)^2 
    \Big(\E[\bY]^\top \bY_u\bY_u^\top \E[\bY] - \bar{\bY}^\top \bY_u\bY_u^\top \bar{\bY}\Big) \Bigg|\Bigg] = O\left(\frac{1}{\sqrt{n}} \right) , 
\end{align}
which shows that we can replace the sample average $\bar{\bY}$ in $\hat Q_n^{(4)}$ with the population mean $\E[\bY]$ upto negligible error. Notice that the proof of \eqref{eq:var_term_4_convg} is completed by combining \eqref{eq:tau4_1} and \eqref{eq:tau4_2} with the triangle inequality. Now, to prove \eqref{eq:tau4_2}, by the triangle inequality, it is enough to show, 
\begin{align*}
    \E\Bigg[\Bigg|\frac{1}{n}\sum_{u=1}^{n} \Big(\frac{\bar{d}_u}{K}-1\Big)^2 
    (\E[\bY] - \bar{\bY})^\top \bY_u\bY_u^\top \E[\bY] \Bigg|\Bigg] = O(1/\sqrt{n}) 
\end{align*}
and
\begin{align*}
    \E\Bigg[\Bigg|\frac{1}{n}\sum_{u=1}^{n} \Big(\frac{\bar{d}_u}{K}-1\Big)^2 
    (\E[\bY] - \bar{\bY})^\top \bY_u\bY_u^\top \bar{\bY} \Bigg|\Bigg] = O(1/\sqrt{n}).
\end{align*}
By Cauchy-Schwartz inequality and bounds on $\bar{d}_u$ from \cite{jaffe2020randomized} gives,
\begin{align}\label{eq:triangle_2_Qn4_1}
    \E\Bigg[\Bigg|\frac{1}{n}\sum_{u=1}^{n} \Big(\frac{\bar{d}_u}{K}-1\Big)^2 
    &(\E[\bY] - \bar{\bY})^\top \bY_u\bY_u^\top \bar{\bY} \Bigg|\Bigg]\\
    &\le \sqrt{\E\left[\|\bar{\bY} - \E[\bY]\|_2^2\right]} 
    \sqrt{\E\left[\left\|\frac{\bar{\bY}^\top}{n} \sum_{u=1}^{n} \left(\frac{\bar{d}_u}{K}-1\right)^2 \bY_u\bY_u^\top \right\|_2^2\right]} \nonumber\\
    &\lesssim \sqrt{\E\left[\|\bar{\bY} - \E[\bY]\|_2^2\right]} 
    \sqrt{\E\left\|\bar{\bY}\|_2^2 \left(\frac{1}{n}\sum_{u=1}^{n} \|\bY_u\|_2^2\right)^2\right]} \nonumber\\
    &\le \sqrt{\E\left[\|\bar{\bY} - \E[\bY]\|_2^2\right]} \sqrt{\E\left[\left(\frac{1}{n}\sum_{u=1}^{n} \|\bY_u\|_2^2\right)^3\right]} = O\left(\frac{1}{\sqrt{n}} \right)  ,  \nonumber  \\  
\end{align}
where the last step uses H\"{o}lder's inequality and the moment assumption. Similarly, we can show
\begin{align}\label{eq:triangle_2_Qn4_2}
    \E\Bigg[\Bigg|\frac{1}{n}\sum_{u=1}^{n} \Big(\frac{\bar{d}_u}{K}-1\Big)^2 
    (\E[\bY] - \bar{\bY})^\top \bY_u\bY_u^\top \E[\bY] \Bigg|\Bigg] = O\left(\frac{1}{\sqrt{n}} \right)  .  
\end{align}
The proof of \eqref{eq:var_term_4_convg} is now completed by combining the upper bound from \eqref{eq:term2_traingle_inequality} with the bounds from \eqref{eq:tau4_1_sq_bd}, \eqref{eq:triangle_2_Qn4_1} and \eqref{eq:triangle_2_Qn4_2}. 

Finally, we show that 
\begin{align}\label{eq:Qn5_bd_needed}
    \E\left[\left|Q_n^{(5)} - \hat Q_n^{(5)}\right|\right] = O\left( \frac{1}{n} \right)  .  
\end{align}
To that end by definition,
\begin{align}\label{eq:Qn5_diff_equal}
    \E\left[\left|Q_n^{(5)} - \hat Q_n^{(5)}\right|\right] = \E\Bigg[\Bigg|\Big(\|\E[\bY]\|_2^4 - \|\bar{\bY}\|_2^4\Big)\frac{1}{n}\sum_{u=1}^{n}\Big(\frac{\bar{d}_u}{k} - 1\Big)^2\Bigg|\Bigg].
\end{align}
Notice that,
\begin{align}
    &\E\Bigg[\Bigg|\Big(\|\E[\bY]\|_2^4 - \|\bar{\bY}\|_2^4\Big)\frac{1}{n}\sum_{u=1}^{n}\Big(\frac{\bar{d}_u}{k} - 1\Big)^2\Bigg|\Bigg] 
    \lesssim \E\Big[|\|\bar{\bY}\|_2^4 - \|\E[\bY]\|_2^4|\Big]\nonumber\\
    &\lesssim\E\left[\bigg|\|\bar \bY\| - \|\E[\bY]\|\bigg|\bigg|\|\bar\bY\|_2^3 + \|\bar \bY\|_2^2\|\E[\bY]\|_2 + \|\bar \bY\|_2\|\E[\bY]\|_2^2 + \|\E[\bY]\|_2^3\bigg|\right]\nonumber\\
    &\lesssim\sqrt{\E\left[\bigg|\|\bar \bY\| - \|\E[\bY]\|\bigg|^2\right]}\nonumber\\
    &\leq\sqrt{\E\left[\left\|\bar \bY - \E[\bY]\right\|_2^2\right]}\label{eq:Qn5_diff_bd_1}
\end{align}
where the first inequality follows by recalling that $\bar{d}_u \lesssim_d K$, for all $1\leq u\leq n$ \citep{jaffe2020randomized} and the penultimate inequality follows by Cauchy-Schwarz inequality and the moment assumption. Recall that $\bZ_u = \bY_u - \E[\bY]$, for all $1 \le u \le n$. Then notice that
\begin{align}\label{eq:bddfourthmoment}
    \E\left[\|\bar{\bZ}\|_2^2\right] = \E\left[\left\langle \frac{1}{n} \sum_{u=1}^{n} \bZ_i, \frac{1}{n} \sum_{u=1}^{n} \bZ_i \right\rangle \right] = \E\left[\frac{1}{n^2} \sum_{1 \leq u, v \leq n} \bZ_u^\top \bZ_v \right] = O(1/n),
\end{align}
where the last equality follows from the definition of $\bZ_u$ and the moment assumptions. The proof of \eqref{eq:Qn5_bd_needed} is now completed by combining the identity from \eqref{eq:Qn5_diff_equal} with the bounds from \eqref{eq:Qn5_diff_bd_1} and \eqref{eq:bddfourthmoment}.
\end{proof}

Finally, using arguments similar to proofs of \eqref{eq:var_term_4_convg} and \eqref{eq:Qn5_bd_needed} one can show that under Assumption \ref{assumption:normX_continuous} and the assumptions of Proposition \ref{ppn:varianceestimate} 
\begin{align}\label{eq:Qn23_order}
    \E_{H_0}\left[\left| Q^{(2)}_n  - \hat Q^{(2)}_n \right|\right] = O\left(n^{-\frac{1}{2}}\right) \quad \text{ and } \quad \E_{H_0}\left[\left| Q^{(3)}_n  - \hat Q^{(3)}_n \right|\right] = O\left(n^{-\frac{1}{2}}\right),
\end{align}
respectively. Recalling the decomposition in \eqref{eq:sigma_n_hat_decomp}, it now suffices to show that
\[
\left|\tilde Q_n\right|
= o_P\left(n^{-\frac{\delta}{16+2\delta}}\right)
\]
to complete the proof of Lemma \ref{lemma:L2_var_est}. We show this in the following result.

\begin{lemma}\label{lemma:tildeQn_o1}
    Suppose Assumption \ref{assumption:normX_continuous} holds and $\E[\|\bY\|_2^{8+\delta}]<\infty$ for some $\delta>0$. Then under $H_0$ from \eqref{eq:H0}, $|\tilde Q_n| = o_P(n^{-\frac{\delta}{16+2\delta}})$.
\end{lemma}

\begin{proof}[Proof of Lemma \ref{lemma:tildeQn_o1}]
To begin with, recall that by the CLT $\bar\bY_n = \E[\bY] + o_p(1/\sqrt{n})$. Hence, by a direct substitution along with moment bound $\E[\|\bY\|_2^{8+\delta}]<\infty$ and the bounds on in and out degrees of $G(\sX_n)$ from \cite{jaffe2020randomized} we can show,
\begin{align*}
    |\tilde Q_{n} - \tilde Q_{n}(\bm \mu)| = O_p(1/\sqrt{n})
\end{align*}
where $Q_{n}(\bm \mu)$ is defined analogous to $\tilde Q_n$ with $\bar\bY_n$ replaced by $\bm \mu = \E[\bY]$. To complete the proof it is now enough to show that $|\tilde Q_{n}(\bm \mu)| = o_P(n^{-\frac{2+\delta}{16+2\delta}})$. To that end note that we can write $\tilde Q_{n}(\bm \mu)$ as,
\begin{align*}
    \tilde Q_{n}(\bm \mu) = \frac{1}{nK^2}
    & \sum_{1 \leq u, v \leq n}T_{\bm\mu}^{(1)}(\bY_u, \bY_v)[\one\left\{\bX_u \rightarrow \bX_v\right\} + \one\{\bX_u\leftrightarrow\bX_j\}]\\
    & + \frac{1}{n}\sum_{u=1}^{n}\left(\frac{\bar d_u}{K}-1\right)^2T_{\bm\mu}^{(2)}(\bY_u) , 
\end{align*}
where 
\begin{align*}
    T_{\bm\mu}^{(1)}(\bY_u, \bY_v) = ((\bY_u - \bm\mu)^\top(\bY_v - \bm\mu))^2 - (\bY_u^\top\bY_v)^2 + \bm\mu^\top(\bY_u\bY_u^\top + \bY_v\bY_v^\top)\bm\mu - \|\bm\mu\|_2^4
\end{align*}
and,
\begin{align*}
    T_{\bm\mu}^{(2)}(\bY_u) = \bm\mu^\top(\bY_u - \bm\mu)(\bY_u-\bm\mu)^\top\bm\mu - \bm\mu^\top\bY_u\bY_u^\top\bm\mu + \|\bm\mu\|_2^4.
\end{align*}
Now for $u\neq v$, under $H_0$ note that $\E_{H_0}[T_{\bm\mu}^{(1)}(\bY_u, \bY_v)\mid \bX_u, \bX_v ] = 0$. \textcolor{black}{Moreover, by a direct computation we can easily verify that under the moment assumptions $\E[(T_{\bm\mu}^{(1)}(\bY_u,\bY_v))^2]<\infty$. Hence, we can now apply Lemma \ref{lemma:VarTn_bd} to conclude,}
\begin{align*}
    \Var\left[\frac{1}{nK^2} \sum_{1 \leq u, v \leq n}T_{\bm\mu}^{(1)}(\bY_u, \bY_v)\one\left\{\bX_u \rightarrow \bX_v\right\}\right] & \lesssim\frac{1}{n}\E\left[\max_{1\leq u\neq v\leq n}T_{\bm\mu}^{(1)}(\bY_u, \bY_v)^2\right]  \nonumber \\  
    & \lesssim\frac{1}{n} + \frac{1}{n}\E\left[\max_{u=1}^{n}\|\bY_u\|_2^6\right].
\end{align*}
Following the arguments from \eqref{eq:max_Y_4_bd} we can conclude that,
\begin{align*}
    \Var\left[\frac{1}{nK^2}\sum_{1 \leq u, v \leq n}T_{\bm\mu}^{(1)}(\bY_u, \bY_v)\one\left\{\bX_u \rightarrow \bX_v\right\}\right]\lesssim n^{-\frac{\delta}{6+\delta}}.
\end{align*}
Similarly we can show,
\begin{align*}
    \Var\left[\frac{1}{nK^2}\sum_{1 \leq u, v \leq n}T_{\bm\mu}^{(1)}(\bY_u, \bY_v)\one\{\bX_u\leftrightarrow\bX_j\}\right]\lesssim n^{-\frac{\delta}{6+\delta}}.
\end{align*}
Hence by Markov inequality,
\begin{align}\label{eq:tildeQ1_1st}
    \frac{1}{nK^2}
    & \sum_{1 \leq u, v \leq n}T_{\bm\mu}^{(1)}(\bY_u, \bY_v)[\one\left\{\bX_u \rightarrow \bX_v\right\} + \one\{\bX_u\leftrightarrow\bX_j\}] = o_P\left(n^{-\frac{\delta}{12+2\delta}}\right).
\end{align}
Next, to complete the proof of Lemma \ref{lemma:tildeQn_o1} note that $\E_{H_0}[T_{\bm\mu}^{(2)}(\bY_u)\mid\bX_u] = 0$. Hence by law of total variance,
\begin{align*}
    \Var\left[\frac{1}{n}\sum_{u=1}^{n}\left(\frac{\bar d_u}{K}-1\right)^2T_{\bm\mu}^{(2)}(\bY_u)\right] 
    & = \E\left[\Var\left[\frac{1}{n}\sum_{u=1}^{n}\left(\frac{\bar d_u}{K}-1\right)^2T_{\bm\mu}^{(2)}(\bY_u)\mid \sX_n\right]\right]\\
    & = \E\left[\frac{1}{n^2}\sum_{u=1}^{n}\left(\frac{\bar d_u}{K}-1\right)^4\Var\left[T_{\bm\mu}^{(2)}(\bY_u)\mid \bX_u\right]\right]  ,  
\end{align*}
where the last equality follows by independence of $\{(\bY_i, \bX_i): 1\leq i\leq n\}$. Using the bound $\bar d_u\lesssim K$ for all $1\leq u\leq n$ from \cite{jaffe2020randomized} we get,
\begin{align*}
    \Var\left[\frac{1}{n}\sum_{u=1}^{n}\left(\frac{\bar d_u}{K}-1\right)^2T_{\bm\mu}^{(2)}(\bY_u)\right] 
    &\lesssim \frac{1}{n^2}\sum_{u=1}^{n}\E\left[\Var\left[T_{\bm\mu}^{(2)}(\bY_u)\mid \bX_u\right]\right] = \frac{1}{n}\Var\left[T_{\bm\mu}^{(2)}(\bY_1)\right]\lesssim \frac{1}{n}  ,  
\end{align*}
where the last equality once again follows by recalling $\E_{H_0}[T_{\bm\mu}^{(2)}(\bY_u)\mid\bX_u] = 0$ and the last inequality follows by recalling the moment condition on $\bY_1$. Once again applying Markov inequality,
\begin{align}\label{eq:tildeQ1_2nd}
    \frac{1}{n}\sum_{u=1}^{n}\left(\frac{\bar d_u}{K}-1\right)^2T_{\bm\mu}^{(2)}(\bY_u) = o_P(1/\sqrt{n}).
\end{align}
The proof is now completed by combining \eqref{eq:tildeQ1_1st} and \eqref{eq:tildeQ1_2nd}.
\end{proof}

The proof of Lemma \ref{lemma:L2_var_est} is now completed by recalling the decompositions of $\sigma_n^2$ and $\hat\sigma_n^2$ from \eqref{eq:sigman_decomp} and \eqref{eq:sigma_n_hat_decomp} along with the convergences from Lemma \ref{lemma:tau1}, \eqref{eq:Qn23_order} and Lemma \ref{lemma:tildeQn_o1}.\hfill $\Box$

\subsubsection{Proof of Lemma \ref{lemma:varTnlimpos}}\label{appendix:proofof_lemma_varTnlimpos}

Recalling the definition of $\sigma_n = \Var_{\bm{H}_0}(R_n| \mathcal{F}(\sX_n))$ from Lemma \ref{lm:varianceRn}, we have
\begin{align}\label{eq:var_tn_xn_lb}
    \sigma_n^2 \geq \frac{1}{n K^2}\sum_{1 \leq u \ne v \leq n}\E\left[\left(\bZ_u^\top \bZ_v\right)^2\middle| \bX_u, \bX_v\right]\one\left\{\bX_u \rightarrow \bX_v\right\},
\end{align}
where $\bZ_i = \bY_u - \E[\bY]$ for all $1\le u \le n$. Notice that
\begin{align}\label{eq:var_term_1_trace}
    & \frac{1}{n K^2}\sum_{1 \leq u \ne v \leq n}\E\left[\left(\bZ_u^\top \bZ_v\right)^2\middle|\bX_u, \bX_v\right]\one\left\{\bX_u \rightarrow \bX_v\right\} \nonumber \\ 
    &= \frac{1}{n K^2}\sum_{1 \leq u \ne v \leq n}\E\left[\tr\left(\bZ_u\bZ_u^\top \bZ_v\bZ_v^\top\right)\middle|\bX_u,\bX_v\right]\one\left\{\bX_u \rightarrow \bX_v\right\} \nonumber\\
    &= \frac{1}{n K^2}\sum_{1 \leq u \ne v \leq n}\tr\left(\E\left[\bZ_u\bZ_u^\top\middle|\bX_u\right]\E\left[\bZ_v\bZ_v^\top\middle|\bX_v\right]\right)\one\left\{\bX_u \rightarrow \bX_v\right\}.
\end{align} 
For notational convenience, define $h(\bX_u) = \E\left[\bZ_u\bZ_u^\top \middle| \bX_u\right]$.  
\textcolor{black}{By the Efron Stein inequality \citep{efron1981jackknife} and following the combinatorial arguments from Lemma \ref{lemma:VarTn_bd},} we get
\begin{align*}
    \Var\left[\frac{1}{n K^2}\sum_{1 \leq u \ne v \leq n}\E\left[\left(\bZ_u^\top \bZ_v\right)^2\middle|\bX_u, \bX_v\right]\one\left\{\bX_u \rightarrow \bX_v\right\}\right]
    &\lesssim_{d} \frac{1}{n K^2}\E\left[\max_{1\le u \neq v \le n} \left|\tr\left(h(\bX_u)h(\bX_v)\right)\right|^2\right] \\
    &\le \frac{1}{n K^2}\E\left[\max_{1\le u \le n} \|h(\bX_u)\|_F^4\right].
\end{align*}
where $\|\cdot\|_F$ is the Frobenius norm. By the tower property of conditional expectation,
\begin{align*}
    \E\left[\max_{1 \leq u \leq n}\|h(\bX_u)\|_F^4\right] \le \E\left[\max_{1 \leq u \leq n}\E\left[\|\bZ_i\|_2^4\middle|\bX_u\right]^2\right] 
    \le \E\left[\max_{1 \leq u \leq n}\|\bZ_i\|_2^8\right] \lesssim \E\left[\max_{1 \leq u \leq n}\|\bY_u\|_2^8\right].
\end{align*}
Repeating arguments similar to \eqref{eq:max_Y_4_bd} and using the moment assumption, for all $\vep>0$,
\begin{align*}
    \frac{1}{n K^2}\E\left[\max_{1 \leq u \leq n}\|\bY_u\|_2^8\right] = o\left(n^{-\frac{\delta}{8+\delta}}\right),
\end{align*}
which gives
\begin{align}\label{eq:var_sigman_t1_convg}
    \Var\left[\frac{1}{n K^2}\sum_{1 \leq u \ne v \leq n}\E\left[\left(\bZ_u^\top \bZ_v\right)^2\middle|\bX_u, \bX_v\right]\one\left\{\bX_u \rightarrow \bX_v\right\}\right] = o\left(n^{-\frac{\delta}{8+\delta}}\right).
\end{align}
Now to complete the proof, it is enough to show that there exists $c>0$ such that
\begin{align*}
    \E\left[\frac{1}{n K^2}\sum_{1 \leq u \ne v \leq n}\E\left[\left(\bZ_u^\top \bZ_v\right)^2\middle|\bX_u, \bX_v\right]\one\left\{\bX_u \rightarrow \bX_v\right\}\right] \to c.
\end{align*}
Using arguments similar to \eqref{eq:unifoneneigh} and \eqref{eq:var_term_1_trace},
\begin{align*}
    &\left|\E\left[\frac{1}{n K^2}\sum_{1 \leq u \ne v \leq n}\E\left[\left(\bZ_u^\top \bZ_v\right)^2\middle|\bX_u, \bX_v\right]\one\left\{\bX_u \rightarrow \bX_v\right\}\right] - \frac{1}{K}\E\left[\|h(\bX_1)\|_F^2\right]\right| \\
    &\quad \le \frac{1}{K}\sqrt{\E\left[\|h(\bX_1)\|_F^2\right]} \sqrt{\E\left[\|h(\bX_1) - h(\bX_{N(1)})\|_F^2\right]},
\end{align*}
where $N(1)$ is a uniformly chosen neighbor of $1$ in $G(\sX_n)$.  Using the tower property, the moment assumptions, and Lemma \ref{lemma:fXN1_bdd_fX1}, we conclude that $\E\left[\|h(\bX_1) - h(\bX_{N(1)})\|_F^2\right] = o(1)$ and $\E\left[\|h(\bX_1)\|_F^2\right] < \infty$. Thus,
\begin{align}\label{eq:Var_t1_convg}
    \E\left[\frac{1}{n K^2}\sum_{1 \leq u \ne v \leq n}\E\left[\left(\bZ_u^\top \bZ_v\right)^2\middle|\bX_u, \bX_v\right]\one\left\{\bX_u \rightarrow \bX_v\right\}\right] \to \frac{1}{K}\E\left[\|h(\bX_1)\|_F^2\right].
\end{align}
Finally, $\E\left[\|h(\bX_1)\|_F^2\right] > 0$ because if it were $0$, we would have $\E[\bZ_1\bZ_1^\top | \bX_1] = 0$ almost surely, implying $\Var[\bY_1] = 0$, contradicting the assumption that $\bY$ is not almost surely constant.


\section{Proof of Corollary \ref{cor:test_consistency}}
\label{sec:consistencypf}
The proof of asymptotic validity of the test $\phi_n$ from \eqref{eq:conditional_mean_test} follows immediately from Theorem \ref{thm:null_dist}. To prove consistency we note that under $H_1$, by \eqref{eq:nn_consistency} and law of large numbers,
\begin{align*}
    T_n\pto \E\left[\left\|\E\left[\bY\mid\bX\right] - \E\left[\bY\right]\right\|_2^2\right] > 0.
\end{align*}
The proof of universal consistency is completed once we can show that $\hat\sigma_n^2 = O_P(1)$ (see \eqref{eq:varianceH0} for definition of $\hat\sigma_n$). Using the Cauchy-Schwarz inequality and bounds on in-degree from \cite{jaffe2020randomized},
\begin{align*}
    \hat\sigma_n^2
    & \lesssim\frac{1}{nK^2}\sum_{u=1}^{n}\sum_{v\in N_{G(\sX_n)}(u)}\left(\|\bY_u-\bar \bY_n\|_2^4 + \|\bY_v - \bar \bY_n\|_2^4\right) + \frac{1}{n}\sum_{u=1}^{n}\left|\bar{\bY}^\top (\bY_u - \bar{\bY}_n) (\bY_u - \bar{\bY}_n)^\top \bar{\bY}\right|\\
    & \lesssim \frac{1}{n}\sum_{u=1}^{n}\|\bY_u - \bar\bY_n\|_2^4 + \frac{1}{n}\sum_{u=1}^{n}\left(\bar{\bY}^\top (\bY_u - \bar{\bY}_n)\right)^2\\
    &\lesssim \|\bar \bY_n\|_2^4 + \frac{1}{n}\sum_{i=1}^{n}\|\bY_i\|_2^4 + \frac{\|\bar\bY_n\|_2^2}{n}\sum_{i=1}^{n}\|\bY_i\|_2^2 = O_p(1)
\end{align*}
where the final equality follows from the moment condition on $\bY$.\hfill $\Box$

\section{Proof of Theorem \ref{thm:variable_screening}}
\label{sec:vspf}


The proof proceeds by showing that $\hat V_n$ is close to $V$ with high probability, and that this event is contained in the event ${\hat S \text{ is sufficient}}$, (as in the proofs of \cite[Theorem 5.1]{huang2022kernel} and \cite[Theorem 6.1]{azadkia2021simple}). We begin by defining a few notations. Let $s_1,\ldots, s_d$ be the indices collected (in order) by Algorithm \ref{alg:greedy_screening}, if it is allowed to run without the stopping rule with $\hat V_n$. Let $S_t = \{s_1,\ldots, s_t\}$, for $1\leq t \leq \kappa$, and define $S_t = S_\kappa$, for $t \geq \kappa+1$ and $S_0 = \emptyset$. Define $\theta = \kappa - \frac{M}{\delta}$, $\vep_1 = \frac{\delta\theta}{12\kappa}$, and $\vep_2 =\frac{\theta}{3\kappa}$. With the above notations, consider the event,
\begin{align*}
    \sE = \left\{\left|\hat V_n(S) - V(S)\right|\leq \frac{\delta\theta}{12\kappa}, \text{ for all }S\subseteq[d]\text{ such that }|S|\leq \kappa\right\} . 
\end{align*}
The following lemma shows that if $\hat S$ is sufficient, then the event $\sE$ will happen.

\begin{lemma}\label{lemma:S_suff_sub_E}
    Under the assumptions of Theorem \ref{thm:variable_screening},
    \begin{align*}
         \sE\subseteq\left\{\hat S\text{ is sufficient}\right\}.
    \end{align*}
\end{lemma}

The proof of Lemma \ref{lemma:S_suff_sub_E} is given in Appendix \ref{sec:sufficientEpf}. To complete the proof of Theorem \ref{thm:variable_screening} using this lemma, it is now enough to show that $\sE$ happens with high probability. To that end consider
\begin{align*}
    W_n(S) = \frac{1}{n}\sum_{u=1}^{n}\frac{1}{K}\sum_{v\in N_{G(\sX_S)}(u)}Y_uY_v \text{ for all }S\subseteq[d].
\end{align*}
Note that by definition $|W_n(S)| = \hat V_n(S)$. By \eqref{eq:EVn_conc_rate} with $\bX$ replaced with $\bX_{S}$ we get,
\begin{align}\label{eq:EVhat_conc}
    \left|\E\left[W_n(S)\right] - V(S)\right|\lesssim_\kappa \max\left\{\frac{(\log n)^2}{\sqrt{n}},\frac{(\log n)^{2}}{n^{\frac{1}{\kappa}}}\right\} , 
\end{align}
for all $S\subseteq[d]$ such that $|S|\leq \kappa$. The following lemma shows that, for any $S$ such that $|S|\leq \kappa$,  $W_n(S)$ concentrates about its mean. The proof is given in Appendix \ref{sec:algorithmpf}.

\begin{lemma}\label{lemma:conc_Vnhat}
    Under the assumptions of Theorem \ref{thm:variable_screening}, for any $S$ with $|S|\leq \kappa$,
    \begin{align*}
        \P\left(\left|W_n(S) - \E\left[W_n(S)\right]\right|>t\right)\lesssim e^{-cnt^2}
    \end{align*} 
    for all $t\in(0,1)$ with some constant $c>0$ depending on $\kappa, \delta, M, \beta, C_1,C_2,C_3$ and $d$.
\end{lemma}

Combining \eqref{eq:EVhat_conc} and Lemma \ref{lemma:conc_Vnhat} we now conclude,
\begin{align*}
    \P\left(\left|W_n(S) - V(S)\right|\gtrsim_\kappa \max\left\{\frac{(\log n)^2}{\sqrt{n}},\frac{(\log n)^{2}}{n^{\frac{1}{\kappa}}}\right\} + t\right)\lesssim  e^{-cnt^2}, \text{ for all }S\text{ such that }|S|\leq\kappa. 
\end{align*}
Now note that $|V(S)| = V(S)$ and recall $|W_n(S)| = \hat V_n(S)$. Then using the inequality,
\begin{align*}
    \left|\hat V_n(S) - V(S)\right| \leq \left|W_n(S) - V(S)\right|
\end{align*}
we can show,
\begin{align*}
    \P\left(\left|\hat V(S) - V(S)\right|\gtrsim_\kappa \max\left\{\frac{(\log n)^2}{\sqrt{n}},\frac{(\log n)^{2}}{n^{\frac{1}{\kappa}}}\right\} + t\right)\lesssim  e^{-cnt^2}, \text{ for all }S\text{ such that }|S|\leq\kappa. 
\end{align*}
Hence, by a union bound,
\begin{align*}
    \P\left(\bigcup_{|S|\leq \kappa}\left|\hat V_n(S) - V(S)\right|\gtrsim_\kappa \max\left\{\frac{(\log n)^2}{\sqrt{n}},\frac{(\log n)^{2}}{n^{\frac{1}{\kappa}}}\right\} + t\right)\lesssim d^\kappa e^{-cnt^2} . 
\end{align*}
Then choosing $t = \delta\theta/24\kappa$, $n$ large enough such that $\max\left\{\frac{(\log n)^2}{\sqrt{n}},\frac{(\log n)^{2}}{n^{\frac{1}{\kappa}}}\right\}\leq \frac{\delta\theta}{24\kappa}$, and Lemma \ref{lemma:S_suff_sub_E} gives, 
\begin{align*}
\P\left( \hat S\text{ is sufficient} \right) \geq \P(\mathscr{E})\geq 1-Cd^\kappa e^{-cn} , 
\end{align*}
where $C$ and $c$ are constants depending on $\kappa, \delta, M, \beta, C_1,C_2,C_3$ and $d$. This completes the proof of Theorem \ref{thm:variable_screening}.

\subsection{Proof of Lemma \ref{lemma:S_suff_sub_E}} 
\label{sec:sufficientEpf}

We consider two cases depending on whether the algorithm stops before or after $\kappa$.

\begin{itemize} 

\item[(1)] The algorithm stops at $t<\kappa$. If $t+1>d$, then it immediately follows that $\hat S = S_t = S_d$ is sufficient. Now, suppose $t + 1\leq d$, then by the stopping rule $\hat V_n(S_{t+1})<\hat V_n(S_t)$. Then for any $s \in [d]\setminus S_t$,
\begin{align}\label{eq:V_t_u_j_t_diff}
    V\left(S_t\cup\{s\}\right) - V(S_t)
    &\leq \hat V_n\left(S_t\cup\{j\}\right) - \hat V_n\left(S_t\right) + \frac{\delta\theta}{6\kappa}\nonumber\\
    &\leq \hat V_n\left(S_{t+1}\right) - \hat V_n\left(S_t\right) + \frac{\delta\theta}{6\kappa}\leq \frac{\delta\theta}{6\kappa}<\delta,
\end{align}
where the first inequality follows by recalling the definition of event $\sE$, the second inequality follows from the construction of $S_{t+1}$  in Algorithm \ref{alg:greedy_screening}. Recalling the definition of $\delta$, the final inequality implies that $S_t$ is sufficient. 

\item[$(2)$] Now, suppose the algorithm stops at $t\geq \kappa$. In this case, $S_\kappa\subseteq S_t$, hence, it is enough to show that $S_\kappa$ is sufficient. By arguments similar to \eqref{eq:V_t_u_j_t_diff} it follows that that if $\hat V_n\left(S_t\right) - \hat V_n\left(S_{t-1}\right)\leq \delta\left(1-\theta/3\kappa\right)$, for some $1\leq t\leq \kappa$, then $S_{t-1}$ is sufficient, which implies $S_\kappa$ is sufficient. Thus, to complete the proof suppose $\hat V_n\left(S_t\right) - \hat V_n\left(S_{t-1}\right)>\delta\left(1-\theta/3\kappa\right)$, for all $1\leq t\leq \kappa$. Then using the definition of the set $\sE$, 
\begin{align*}
    V(S_{t}) - V(S_{t-1}) \geq \hat V_n\left(S_t\right) - \hat V_n\left(S_{t-1}\right) - \frac{\delta\theta}{6\kappa} > \delta\left(1-\theta/3\kappa\right) - \frac{\delta\theta}{6\kappa} = \delta - \frac{\delta\theta}{2\kappa}.
\end{align*}
By a considering a telescoping sum,
\begin{align*}
    V(S_\kappa) = \sum_{t=1}^{\kappa} \left(V(S_{t}) - V(S_{t-1})\right) + V(S_0)\geq \kappa\left(\delta - \frac{\delta\theta}{2\kappa}\right)>M.
\end{align*}
However, note that $V(S_\kappa)\leq \E[Y^2]\leq M$, which is a contradiction. This completes the proof of Lemma \ref{lemma:S_suff_sub_E}. \hfill $\Box$
\end{itemize}

\subsection{Proof of Lemma \ref{lemma:conc_Vnhat}} 
\label{sec:algorithmpf} 

For $1\leq u \leq n$, define $\eta_{u} = Y_u - m_S(\bX_{S, u})$, where $m_S(\bx_S) = \E\left[Y\mid\bX_S = \bx_S\right]$. Moreover let $A_S$ be the adjacency matrix of $G(\sX_S)$, the $K$-NN constructed using $\sX_S$ and take $\bA_S = \frac{1}{n K} A_S$. To keep notations compact, we will use $c$ to denote universal constants, whose may change across the proof. Then, by definition, 
\begin{align}\label{eq:hatVn_decomp}
    W_n(S) = \bY_n^\top\bA_S\bY_n 
    & = \bm{\eta}_n^\top\bA_S\bm{\eta}_n + \bm{m}_S^\top\bA_S\bm{m}_S + \bm{m}_S^\top\bA_S\bm{\eta}_n + \bm{\eta}_n^\top\bA_S\bm{m}_S.
\end{align}
    where $\bm{\eta}_n = (\eta_1,\ldots, \eta_n)^\top, \bm{m}_S = \left(m_S(\bX_{S,1}),\ldots, m_S\left(\bX_{S,n}\right)\right)$. and $\bY_n = (Y_1,\ldots, Y_n)^\top$. 
    We begin with the analysis of the first term in \eqref{eq:hatVn_decomp}. To this end, fix $t \in(0, 1) $ and observe that, for $1 \leq u \leq n$, 
\begin{align}\label{eq:sub_gaussian_eta}
    \E\left[e^{t\eta_{u}} \mid \sX_{S}\right] 
    & = \E\left[\E\left[e^{t\left(Y_u - m(\bX_u)\right) + t\left(m(\bX_u) - m_S(\bX_{S, u})\right)}\mid \sX_n\right]\mid \sX_S\right]\nonumber\\
    & = \E\left[\E\left[e^{t\left(Y_u - m(\bX_u) \right) } \mid \sX_n\right]e^{t\left(m(\bX_u) - m_S(\bX_{S, u})\right)} \mid \sX_{S}\right]    \lesssim e^{\lambda^2t^2}  ,   
\end{align}  
for some $\lambda>0$. Observe that the last inequality in \eqref{eq:sub_gaussian_eta} follows from the assumption (d) in the statement of Theorem \ref{thm:variable_screening}, recalling the uniform bound on $m$, and applying Hoeffding's Lemma. Hence, we can now apply a conditional version of the Hanson-Wright inequality  \citep{rudelson2013hanson} to get, 
\begin{align*}
    \P\left(\left|\bm{\eta}_n^\top\bA_S\bm{\eta}_n\right|>t\mid \sX_{S}\right)\leq 2 e^{-c\min\left\{\frac{t^2}{\|\bA_S\|_F^2},\frac{t}{\|\bA_S\|_{\mathrm{op}}}\right\}} ,
\end{align*}
for some $c>0$ depending on the uniform bound on $m$ and $\lambda$, where $\|\cdot\|_{F}$ is the Frobenius norm of a matrix and $\|\cdot\|_{\mathrm{op}}$ is the operator norm. In the following, we bound the two matrix norms of $\bA_S$. By a direct computation, it is easy to show that $\|\bA_S\|_F^2\lesssim \frac{1}{nK}$. For the operator norm, we recall the following inequality:
\begin{align*}
    \|\bA_S\|_{\mathrm{op}}\leq \frac{1}{nK}\sqrt{\left(\max_{1\leq j\leq n}\sum_{i=1}^{n}|a_{S, ij}|\right)\left(\max_{1\leq i\leq n}\sum_{j=1}^{n}|a_{S, ij}|\right)}  ,  
\end{align*}
with $\bA_S = \frac{1}{nK}(a_{S, ij})_{1\leq i,j\leq n}$. For a $K$-NN graph, note that the first and second terms in the above bound correspond to the maximum in-degree and maximum out-degree of $G(\sX_S)$, respectively. Applying Lemma 1 of \cite{jaffe2020randomized}, we know that the maximum degree of a vertex in a $K$-NN graph is bounded above by $c_{|S|}K$, for some constant depending on the dimension $|S|$. Hence, we immediately conclude that $\|\bA_S\|_{\mathrm{op}}\lesssim_{d} \frac{1}{n}$. Then,
\begin{align}\label{eq:eta_quad_form_conc}
    \P\left(\left|\bm{\eta}_n^\top\bA_S\bm{\eta}_n\right|>t\right)\leq 2 e^{-c\min\left\{n Kt^2, nt\right\}} . 
\end{align} 
Next, we consider the third term in \eqref{eq:hatVn_decomp}. Since $m_S$ is uniformly bounded, $\|\bA_S^\top\bm{m}_S\|_2\lesssim 1/\sqrt{n}$. Then the conditional sub-gaussianity from \eqref{eq:sub_gaussian_eta} gives, 
\begin{align*}
    \P\left(\left|\bm{m}_S^\top\bA_S\bm{\eta}_n\right|>t\mid\sX_S\right)\leq 2e^{-cnt^2}.
\end{align*}
Similarly, for the fourth term in \eqref{eq:hatVn_decomp}, 
\begin{align*}
    \P\left(\left|\bm{\eta}_n^\top\bA_S\bm{m}_S\right|>t\mid\sX_S\right)\leq 2e^{-cnt^2}.
\end{align*}
Combining the above gives, 
\begin{align}\label{eq:eta_linear_term_conc}
    \P\left(\left|\bm{m}_S^\top\bA_S\bm{\eta}_n + \bm{\eta}_n^\top\bA_S\bm{m}_S\right|>t\right)\leq 2e^{-cnt^2}
\end{align}
Finally, we consider the second term in \eqref{eq:hatVn_decomp}. For this, note that, 
\begin{align*}
    \bm{m}_S^\top\bA_S\bm{m}_S = \frac{1}{n K}\sum_{u=1}^{n}\sum_{v\in N_{G(\sX_S)}(v)}m_S(\bX_{S, u})m_S(\bX_{S, v}) . 
\end{align*}
Now, from the arguments in \cite[Appendix B.2]{chatterjee2024kernel} and McDiarmid's bounded difference inequality, 
\begin{align}\label{eq:mS_quad_form_conc}
    \P\left(\left|\bm{m}_S^\top\bA_S\bm{m}_S - \E\left[\bm{m}_S^\top\bA_S\bm{m}_S\right]\right|>t\right)\leq 2e^{-cnt^2},
\end{align}
where the constant $c$ depends on $\kappa$. To complete the proof observe that $\E[W_n(S)] = \E[\bm{m}_S^\top\bA_S\bm{m}_S]$. Then combining \eqref{eq:eta_quad_form_conc}, \eqref{eq:eta_linear_term_conc} and \eqref{eq:mS_quad_form_conc} we conclude that for all $t\in (0,1)$,
\begin{align*}
    \P\left(\left|W_n(S) - \E\left[W_n(S)\right]\right|>t\right)\lesssim e^{-cnt^2}.
\end{align*}
for some constant $c>0$ depending on $\kappa$. \hfill $\Box$

\section{Technical Lemmas } 

In this section we collect the proofs of a few technical lemmas that are used in the proofs of the main results. \color{black}
The first result is about bounds on variance of general nearest neighbor based functionals. To set up notations recall the notations from Section \ref{sec:estimation_nn}. Suppose we have observation $\{(\bY_i, \bX_i): 1\leq i\leq n\}$ from the joint distribution $P_{\bY\bX}$ on $\R^p\times\R^d$. Fix $K\geq 1$, and let $G(\sX_n)$ to be the directed $K-$nearest neighbor ($K-$NN) graph constructed using $\sX_n = \{\bX_1,\ldots, \bX_n\}$. For a function $f:\R^p\times \R^p\ra\R$ define,
\begin{align}\label{eq:Tnf}
    T_{n,f} = \frac{1}{nK}\sum_{u=1}^{n}\sum_{v=1}^nf\left(\bY_u, \bY_v\right)\one\left\{\bX_u\ra\bX_v\right\},
\end{align}
where $\bX_u\ra\bX_v$ denotes a directed edge from $\bX_u$ to $\bX_v$ in $G(\sX_n)$. In the following lemma we provide a bound on $\Var(T_{n,f})$. 
\begin{lemma}\label{lemma:VarTn_bd}
    Suppose Assumption \ref{assumption:normX_continuous} holds and let $f:\R^p\times \R^p$ be such that $\E[f^2(\bY_u, \bY_v)]<\infty$. Then,
    \begin{align*}
        \Var\left(T_{n,f}\right)\lesssim_d\frac{1}{n}\E\left[\max_{1\leq u\neq v\leq n}f^2(\bY_u, \bY_v)\right].
    \end{align*}
\end{lemma}
\begin{proof}
    The proof proceeds by an application of the Efron-Stein inequality \citep{efron1981jackknife}. To that end let $\{(\bY_i^\prime,\bX_i^\prime): 1\leq i\leq n\}$ be an independent copy of $\{(\bY_i, \bX_i): 1\leq i\leq n\}$ and for $1\leq s\leq n$ define,
    \begin{align}\label{eq:XY_ns}
        (\bY_{u,s}, \bX_{u,s}) = 
        \begin{cases}
            (\bY_s^\prime, \bX_s^\prime) \text{ if } u = s\\
            (\bY_u, \bX_u) \text{ if } u\neq s.
        \end{cases}
    \end{align}
    Note that we can rewrite $T_{n,f}$ from \eqref{eq:Tnf} as,
    \begin{align*}
        T_{n,f} = \frac{1}{nK}\sum_{(u,v)\in E(G(\sX_n))}f\left(\bY_u, \bY_v\right)
    \end{align*}
    where $E(G(\sX_n)) = \{(u,v): \bX_u\ra\bX_v\}$ is the set of directed edges in $G(\sX_n)$. With the definition from \eqref{eq:XY_ns} we can now define,
    \begin{align*}
        T_{n,s,f} = \frac{1}{nK}\sum_{(u,v)\in E(G(\sX_{n,s}))}f(\bY_{u,s}, \bY_{v,s})
    \end{align*}
    where $G(\sX_{n,s})$ is the directed $K-$NN constructed using samples $\sX_{n,s} = \{\bX_{1,s},\ldots, \bX_{n,s}\}$ and $E(G(\sX_{n,s}))$ is defined analogous to $E(G(\sX_{n}))$. The Efron-Stein inequality now implies,
    \begin{align}\label{eq:efron_stein_f}
        \Var(T_{n,f})\leq \sum_{s=1}^{n}\E\left[\left(T_{n,f} - T_{n,s,f}\right)^2\right].
    \end{align}
    Fix $1\leq s\leq n$. To isolate the terms that cancel in the difference $T_{n,f} - T_{n,s,f}$ define,
    \begin{align*}
        \mathcal{C}_{n,s} = \left\{(u,v)\in E(G(\sX_n))\bigcap E(G(\sX_{n,s})): u\neq s, v\neq s\right\}
    \end{align*}
    to be the collection of edges common to $G(\sX_{n})$ and $G(\sX_{n,s})$ such that none of the end points are the index $s$. Let $\mathcal D_{n} = E(G(\sX_n))\setminus \mathcal C_{n,s}$ and $\mathcal D_{n,s} = E(G(\sX_{n,s}))\setminus \mathcal C_{n,s}$ to be the remaining edges in each graph. Then by definition,
    \begin{align*}
        T_{n,f} - T_{n,s,f} = \frac{1}{nK}\left[\sum_{(u,v)\in \mathcal D_{n}}f(\bY_u, \bY_v) - \sum_{(u,v)\in \mathcal D_{n,s}}f(\bY_{u,s}, \bY_{v,s})\right].
    \end{align*}
    Hence,
    \begin{align}\label{eq:Tnf_diff_bd_1}
        \left|T_{n,f} - T_{n,s,f}\right|\leq \frac{\max\{\left|\mathcal D_{n}\right|, \left|\mathcal D_{n,s}\right|\}}{nK}\left[\max_{1\leq u\neq v\leq n}\left|f(\bY_u, \bY_v)\right|+ \max_{1\leq u\neq v\leq n}\left|f(\bY_{u,s}, \bY_{v,s})\right|\right].
    \end{align}
    To complete the proof in the following we bound the counts $|\mathcal D_n|$ and $|\mathcal D_{n,s}|$. By definition,
    \begin{align*}
        |\mathcal D_{n}|\leq \left|\left\{(u,v)\in E(G(\sX_n)):s\in \{u,v\}\right\}\right| + \sum_{u\neq s}\left|\{v:(u,v)\in \mathcal D_n\}\right|.
    \end{align*}
    To further bound the second term consider $\mathcal A_{n,s} = \{u: (u,s)\in E(G(\sX_n))\bigcup E(G(\sX_{n,s}))\}$ to be the collection of all indices $u$ such that there is an edge to index $s$ in either $G(\sX_{n})$ or $G(\sX_{n,s})$. Now note that by construction if $u\not\in \mathcal A_{n,s}$, then all edges from $u$ belong to $\mathcal C_{n,s}$. However if $u\in \mathcal A_{n,s}$, then at most one edge from $u$ lies in $\mathcal D_n$ and $\mathcal D_{n,s}$. Hence,
    \begin{align*}
        |\mathcal D_{n}|\leq \left|\left\{(u,v)\in E(G(\sX_n)):s\in \{u,v\}\right\}\right| + \sum_{u\neq s}\one\{u\in \mathcal A_{n,s}\}\lesssim_d K,
    \end{align*}
    where the final bound follows by recalling that maximum degree of any vertex in a $K-$NN is bounded above by $c_dK$ for some constant $c_d$ dependent on the dimension $d$ (see Lemma 1 in \cite{jaffe2020randomized}). Similarly, we can show $|\mathcal D_{n,s}|\lesssim_d K$. Subsequently, recalling the bound on $T_{n,f} - T_{n,s,f}$ from \eqref{eq:Tnf_diff_bd_1} we get,
    \begin{align*}
        \left|T_{n,f} - T_{n,s,f}\right|\lesssim_d\frac{1}{n}\left[\max_{1\leq u\neq v\leq n}\left|f(\bY_u, \bY_v)\right|+ \max_{1\leq u\neq v\leq n}\left|f(\bY_{u,s}, \bY_{v,s})\right|\right].
    \end{align*}
    Substituting the upper bound in \eqref{eq:efron_stein_f} implies,
    \begin{align*}
        \Var(T_{n,f})\lesssim_d\frac{1}{n}\E\left[\max_{1\leq u\neq v\leq n}f^2(\bY_u, \bY_v)\right] ,  
    \end{align*}
    which completes the proof.
\end{proof}

The following lemma, which provides a bound on $\mathbb{E}[f(X_{N(1)})]$ for measurable functions $f$, where $N(1)$ is a uniformly sampled vertex in $G(\sX_n)$. This appears in \cite{deb2020measuring} for more general class of geometric graphs. We include it here for completeness.

\begin{lemma}\label{lemma:fXN1_bdd_fX1}
    Let $f:\R^d\ra[0,\infty)$ be any measurable function and suppose Assumption \ref{assumption:normX_continuous} holds. Then for $\sX_n = \{\bX_1,\ldots, \bX_n\}$ generated independently from $P_{\bX}$ let $G(\sX_n)$ be the directed $K-$nearest neighbor graph constructed using $\sX_n$. Let $N(1)$ be a vertex uniformly chosen from the neighbors of index $1$ in the graph $G(\sX_n)$. Then,
    \begin{align*}
        \E\left[f(X_{N(1)})\right]\lesssim_{d} \E\left[f(X_1)\right].
    \end{align*}
\end{lemma}

\begin{proof}
    Note that by definition
    \begin{align*}
        \E\left[f(X_{N(1)})\right] 
        & = \E\left[\frac{1}{K}\sum_{v=1}^{n}f(X_v)\one\left\{\bX_1\ra\bX_v\right\}\right]\\
        & = \E\left[\frac{1}{nK}\sum_{u=1}^{n}\sum_{v=1}^{n}f(X_v)\one\{\bX_u\ra\bX_v\}\right]
    \end{align*}
    where the last equality follows by exchangeability. Now recalling that maximum degree of a vertex in $K-$NN is bounded above by $c_dK$ for some constant $c_d$ depending on the dimension $d$ we get,
    \begin{align*}
        \E\left[f(X_{N(1)})\right] \lesssim_d\E\left[\frac{1}{n}\sum_{v=1}^{n}f(X_v)\right] = \E[f(X_1)].
    \end{align*}
This completes the proof of the result. 
\end{proof}

\color{black}

\begin{lemma}\label{lemma:replace_lemma}
    Consider sequences of random variables $M_n, R_n,\sigma_n,t_n$ such that $M_n = R_n + o_P(1)$, $\sigma_n\geq t_n$, and $t_n = c + o_P(1)$, where $c>0$ is a constant. Suppose 
    \begin{align}\label{eq:R_n_kol_convg}
        \sup_{z\in \R}\left|\P\left(\frac{R_n}{\sigma_n}\leq z\right) - \Phi(z)\right|\ra 0 , 
    \end{align}
as $n\ra \infty$. Then,
    \begin{align*}
        \sup_{z\in \R}\left|\P\left(\frac{M_n}{\sigma_n}\leq z\right) - \Phi(z)\right|\ra 0 , 
    \end{align*} 
as $n\ra \infty$. 
\end{lemma}

\begin{proof} By definition, $M_n = R_n + s_n$, where $s_n = o_P(1)$. Fix $z\in \R$ and $\vep>0$. By assumption of the lemma note that $\P(\sigma_n\leq c/2)\ra 0$. Note that,
    \begin{align*}
        \P\left(\frac{M_n}{\sigma_n}\leq z\right)  \leq \P\left(\frac{R_n}{\sigma_n}\leq z + \frac{2\vep}{c}\right) + o(1)
        &\leq \Phi\left(z + \frac{2\vep}{c}\right) + \gamma_n + o(1)\\
        &\leq \Phi(z) + \frac{2\vep}{\sqrt{2\pi}c} + \gamma_n + o(1),
    \end{align*}
    where the first inequality follows from the convergence in \eqref{eq:R_n_kol_convg} with,
    \begin{align*}
        \gamma_n = \sup_{z\in \R}\left|\P\left(\frac{R_n}{\sigma_n}\leq z\right) - \Phi(z)\right|
    \end{align*}
    and the second inequality follows from the Lipschitz property of $\Phi$. Similarly, 
    \begin{align*}
        \P\left(\frac{M_n}{\sigma_n}\leq z\right)  \geq \P\left(\frac{R_n}{\sigma_n}\leq z - \frac{2\vep}{c}\right) + o(1)
        &\geq \Phi\left(z - \frac{2\vep}{c}\right) - \gamma_n + o(1)\\
        &\geq \Phi(z) - \frac{2\vep}{\sqrt{2\pi}c} - \gamma_n + o(1).
    \end{align*}
    Hence,
    \begin{align*}
        \sup_{z\in \R}\left|\P\left(\frac{M_n}{\sigma_n}\leq z\right) - \Phi(z)\right|\leq \frac{2\vep}{\sqrt{2\pi}c} + \gamma_n + o(1)\leq \frac{2\vep}{\sqrt{2\pi}c} + \gamma_n + o(1), 
    \end{align*}
    The proof is now completed by recalling that $\gamma_n = o(1)$ and $\vep>0$ is chosen arbitrarily.
\end{proof}

\begin{lemma}\label{lemma:ratio_Op} 
    Let $A_n, B_n$ be sequences of real valued random variables and $a, b \in \R$ be constants, with $b \neq 0$. Suppose that there exists deterministic sequences $\alpha_n, \beta_n\ra\infty$, as $n\ra\infty$, such that,
    \begin{align*}
        \alpha_n|A_n-a| = O_P(1) \quad \text{ and } \quad \beta_n|B_n-b| = O_P(1).
    \end{align*}
    Then,
    \begin{align*}
        \min\{\alpha_n, \beta_n\}\left|\frac{A_n}{B_n} - \frac{a}{b}\right| = O_P(1).
    \end{align*}
\end{lemma}

\begin{proof} The proof by follows observing the following: 
    \begin{align*}
        \left|\frac{A_n}{B_n} - \frac{a}{b}\right| \leq \frac{b \left|A_n - a \right| + |a| \left| B_n - b \right|}{b |B_n|} = O_P\left(\frac{1}{\min\{\alpha_n, \beta_n\}}\right) , 
    \end{align*}
    where the last inequality follows from the rate assumptions on $A_n$ and $B_n$ and noticing that $B_n\pto b >0$.
\end{proof}

\end{document}